\documentclass{article}

\usepackage{adjustbox}
\usepackage{amsmath}
\usepackage{amsthm}
\usepackage{graphicx}
\usepackage[algo2e,vlined,linesnumbered]{algorithm2e}

\usepackage{subfig}
\usepackage{rotating}
\usepackage{multirow}
\usepackage{hyperref}
\usepackage{booktabs}

\usepackage[a4paper,left=3cm,right=2cm,top=2.5cm,bottom=2.5cm]{geometry}

\theoremstyle{definition}
\newtheorem{definition}{Definition}[section]
\theoremstyle{theorem}

\theoremstyle{lemma}
\newtheorem{lemma}{Lemma}[section]

\newcommand{\depstop}{\ensuremath{\mathrm{dep\_stop}}}
\newcommand{\arrstop}{\ensuremath{\mathrm{arr\_stop}}}
\newcommand{\deptime}{\ensuremath{\mathrm{dep\_time}}}
\newcommand{\arrtime}{\ensuremath{\mathrm{arr\_time}}}

\newcommand{\trip}{\ensuremath{\mathrm{trip}}}
\newcommand{\dur}{\ensuremath{\mathrm{dur}}}
\newcommand{\enter}{\ensuremath{\mathrm{enter}}}
\newcommand{\exit}{\ensuremath{\mathrm{exit}}}
\newcommand{\leg}{\ensuremath{\mathrm{leg}}}
\newcommand{\change}{\ensuremath{\mathrm{change}}}
\newcommand{\first}{\ensuremath{\mathrm{first}}}

\usepackage{tabularx,environ,amsmath,amssymb}
\makeatletter
\newcommand{\problemtitle}[1]{\gdef\@problemtitle{#1}}
\newcommand{\probleminput}[1]{\gdef\@probleminput{#1}}
\newcommand{\problemoutput}[1]{\gdef\@problemoutput{#1}}
\NewEnviron{problem}{
  \problemtitle{}\probleminput{}\problemoutput{}
  \BODY
  \par\addvspace{.5\baselineskip}
  \noindent
  \begin{tabularx}{\textwidth}{@{\hspace{\parindent}} l X c}
    \multicolumn{2}{@{\hspace{\parindent}}l}{\textbf{\@problemtitle}} \\
    \textbf{Input:} & \@probleminput \\
    \textbf{Output:} & \@problemoutput
  \end{tabularx}
  \par\addvspace{.5\baselineskip}
}
\makeatother

\begin{document}

\title{Connection Scan Algorithm\thanks{Support by DFG grant WA654/16-2}}

\author{%
Julian Dibbelt\footnotemark[3]{~}, Thomas Pajor\footnotemark[3]{~}, Ben Strasser\footnotemark[2]{~}, Dorothea Wagner\footnotemark[2]{~}\\
\footnotemark[2]{~}Karlsruhe Institute of Technology (KIT), Germany\\
\footnotemark[3]{~}work done while at KIT\\
\texttt{algo@dibbelt.de}~~~~\texttt{thomas@tpajor.com}\\
\texttt{strasser@kit.edu}~~~~\texttt{dorothea.wagner@kit.edu}
}

\date{March 2017}

\maketitle

\begin{abstract}
We introduce the Connection Scan Algorithm (CSA) to efficiently answer queries to timetable information systems.
The input consists, in the simplest setting, of a source position and a desired target position.
The output consist is a sequence of vehicles such as trains or buses that a traveler should take to get from the source to the target.
We study several problem variations such as the earliest arrival and profile problems.
We present algorithm variants that only optimize the arrival time or additionally optimize the number of transfers in the Pareto sense.
An advantage of CSA is that is can easily adjust to changes in the timetable, allowing the easy incorporation of known vehicle delays.
We additionally introduce the Minimum Expected Arrival Time (MEAT) problem to handle possible, uncertain, future vehicle delays.
We present a solution to the MEAT problem that is based upon CSA.
Finally, we extend CSA using the multilevel overlay paradigm to answer complex queries on nation-wide integrated timetables with trains and buses.
\end{abstract}

\section{Introduction}

We study the problem of efficiently answering queries to timetable information systems. 
Efficient algorithms are needed as the foundation of complex web services such as the Google Transit or bahn.de - the German national railroad company's website. 
To use these websites, the user enters his desired departure stop, arrival stop and a vague moment in time and the system should compute a journey telling the user when to take which train. 
In practice, trains do not adhere perfectly to the timetable and therefore it is necessary to be able to quickly adjust the scheduled timetable to the actual situation or account in advance for possible delays.

At its core, the studied problem setting consists of the classical shortest path problem.
This problem is usually solved using Dijkstra's algorithm \cite{d-ntpcg-59} which is build around a priority queue.
Algorithmic solutions that reduce timetable information systems to variation of the shortest path problem that are solved with extensions of Dijkstra's algorithm are therefore common.
The time-dependent and time-expanded graph \cite{pswz-emtip-08} approaches are prominent examples.

In this work, we present an alternative approach to the problem, namely the \emph{Connection Scan Algorithm} (CSA).
The core idea consists of doing away with the priority queue and replacing it with a list of trains sorted by departure time.
Contrary to most competitors, CSA is therefore not build upon Dijkstra's algorithm.
The resulting algorithm is comparatively simple because the complexity inherent to the queue is missing.
Further, Dijkstra's algorithm spends most of its execution time within queue operations.
Our approach replaces these with faster more elementary operations on arrays.
The resulting algorithm is therefore also able of achieving low query running times.
A further advantage of our approach is that the data structure consists primarily of an array of trains sorted by departure time.
Maintaining a sorted array is easy even when train schedules change.

Modern timetable information systems do not only optimize the arrival time.
A common approach consists of optimizing several criteria in the Pareto sense \cite{mswz-tima-07,dms-mcspt-08,bm-somcs-09}.
The practicality of this approach was shown by~\cite{mw-pspof-01}.
The most common second criterion is the number of transfers.
Another often requested criterion is the price~\cite{ms-pltcm-06} but we omit this criterion from our study because of very complex realworld pricing schemes.
A further commonly considered problem variant consists of profile queries.
In this variant the input does not contain a departure time.
Instead, the output should contain all optimal journeys between two stops for all possible departure times.
As further problem variant, we propose and study the minimum expected arrival time (MEAT) problem setting to compute delay-robust journeys.

CSA is very fast as it does not possess a heavyweight preprocessing step.
This makes the algorithm comparatively simple but it also makes the running time inherently dependent on the timetable's size.
For very large instances this can be a problem.
We therefore study an algorithmic extension called Connection Scan Accelerated (CSAccel) which combines a multilevel overlay approach~\cite{sww-daola-99,hsw-emlog-08,dgpw-crprn-13} with CSA.

\paragraph{Related Work.}

Finding routes in transportation networks is the focus of many research projects and thus many publications on this subject exist.
The published papers can be roughly divided into two categories depending on whether the studied network is timetable-based.
As our research focuses on timetable routing, we restrict our exposition to it and refer to a recent survey~\cite{bdgmpsww-rptn-16} for other routing problems.

Some techniques are processing-based and have an expensive and slow startup phase.
The advantage of preprocessing is, that it decreases query running times.
A major problem with preprocessing-based techniques is that the preprocessing needs to be rerun each time that the timetable changes.
We start by proving an overview over techniques without preprocessing and afterwards describe the preprocessing-based techniques.

The traditional approach consists of extending Dijkstra's algorithm.
Two common methods exist and are called the time-dependent and time-expanded graph models~\cite{pswz-emtip-08}.
In \cite{dkp-pcbcp-12} the time-dependent model has been refined by coloring graph elements.
The authors further introduce SPCS, an efficient algorithm to answer earliest arrival profile queries.
A parallel version called PSPCS is also introduced.
We experimentally compare CSA to SPCS, to the colored time-dependent model and the basic time-expanded model.

Another interesting preprocessing-less technique is called RAPTOR and was introduced in \cite{dpw-rbptr-14}.
Just as CSA it does not employ a priority queue and therefore is not based on Dijkstra's algorithm.
It inherently supports optimizing the number of transfers in the Pareto-sense in addition to the arrival time.
A profile extension called rRAPTOR also exists.
We experimentally compare CSA with RAPTOR and rRAPTOR.

Adjusting the time-dependent and time-expanded graphs to account for realtime delays is conceptually straightforward but the details are non-trivial and difficult as the studies of~\cite{ms-etipd-09} and~\cite{cddfgpz-egbmd-14} show.

In \cite{bgm-fdsut-10} SUBITO was introduced. 
This is an acceleration of Dijkstra's algorithm applied to the time-dependent graph model.
It works using lower bounds on the travel time between stops to prune the search.
As slowing down trains does not invalidate the lower bounds, most realworld train delays can be incorporated.
However, CSA supports more flexible timetable updates. 
For example, contrary to SUBITO CSA supports the efficient insertion of connections between stops that were previously not directly connected. 

In \cite{w-tbptr-15} trip-based routing (TB) was introduced.
It works by computing all possible transfers between trains in a preprocessing step.
The preprocessing running times are still well below those of other preprocessing-based techniques but non-negligible.
Unfortunately, the achieved query speedup lacks behind techniques with more extensive preprocessing.
In~\cite{w-tbptr-16} the technique was extended with a significantly more heavy-weight preprocessing algorithm that stores a large amount of trees to achieve higher speedups.

Many more preprocessing-based techniques exist.
For example, in \cite{g-ctnrt-10} Contraction Hierarchy, a very successful technique for road routing, was adapted for timetable-based routing.
In~\cite{ddpw-ptl-15}, Hub-labeling, another successful technique for roads, was also adapted for timetable-based routing.
Another labeling-based approach was proposed in~\cite{wlyxz-erppt-15}.
In addition to SUBITO, \cite{bgm-fdsut-10} introduces $k$-flags. 
$k$-flags is an adaptation of Arc-Flags \cite{l-aefea-04}, a further successful technique for roads, to timetables.
Another well-known preprocessing-based technique is called Transfer Patterns (TP).
It was introduced in~\cite{bceghrv-frvlp-10} and was refined since then over the course of several papers.
In~\cite{bs-fbspt-14} the authors combined frequency-based compression with routing and used it to decrease the TP preprocessing running times.
In~\cite{bhs-stp-16} TP was combined with a bilevel overlay approach to further decrease preprocessing running times.
CSAccel is not the first technique to combine multilevel routing with timetables.
This was already done in~\cite{swz-umlgt-02}.

We postpone giving an overview over the existing papers related to the MEAT problem until Section~\ref{sec:related-work-meat}, as the details of the MEAT problem are described in Section~\ref{sec:MEAT}.

\paragraph{Previous publications.}

This paper is the aggregated journal version of three conference papers. 
In \cite{dpsw-isftr-13}, we introduced CSA and the very basic MEAT problem.
In \cite{sw-csa-13}, we first described CSAccel.
In \cite{dsw-drjtn-14}, we present a more in-depth description and evaluation of the MEAT problem setting.

\paragraph{Contribution.}

We describe the Connection Scan family of algorithms (CSA) to solve various routing problems in timetable-based networks.
We describe profile and non-profile variants.
Algorithm variants are described that optimize arrival time and optionally the number of transfers in the Pareto sense.
We further describe Connection Scan Accelerated (CSAccel) a combination of CSA with multilevel overlay techniques.
Finally, we define the Minimum Expected Arrival Time (MEAT) problem and describe how it can be solved using CSA.
All algorithm descriptions are accompanied by an in-depth experimental analysis and experimental comparison with relevant related work.

\paragraph{Outline.}

Our paper is organized into five sections.
The first section contains the preliminaries.
It consists of the formal timetable definition and precisely states nearly all problem settings considered in the following sections.
The second section describes the basic CSA without profiles.
The third section extends CSA to profiles.
In the fourth section, we describe CSAccel, a multilevel extension of CSA.
The fifth section formalizes the MEAT problem and describes how it can be solved within the CSA framework.
The final section is a conclusion. 

\section{Preliminaries}

We describe the Connection Scan algorithm in terms of train networks.
Fortunately, many other transportation networks exist with the same timetable-based structure.
Flight, ship, and bus networks are examples thereof. 
We could therefore formulate our work in more abstract terms such as vehicles.
However, to avoid an unnecessary clumsy language, we refrain from it, and just refer to every vehicle as train.

\subsection{Timetable Formalization}

In this section, we formalize the notion of timetable, which is part of the input of nearly every algorithm presented in this paper.
We are not the first to present a formalization.
However, even though many previous works exist, they differ when it comes to notation and details.
We therefore explain our terminology and the model used in our work in detail to avoid confusion. 

A timetable encodes what trains exist, when they drive, where they drive, and how travelers can transfer between trains.
Especially, the details of the last part --- changing trains --- vary significantly across related work.
Unfortunately, unlike one intuitively might expect, these details impact the algorithm design and can have a huge impact on the running time behavior.
Further, these details can make a timetable description verbose.
Therefore, we first describe the entities not related to transfers, give examples for these, and only afterwards describe the transfer details.

A \emph{timetable} is a quadruple $(\mathcal{S},\mathcal{C},\mathcal{T},\mathcal{F})$ of stops $\mathcal{S}$, connections $\mathcal{C}$, trips $\mathcal{T}$, and footpaths $\mathcal{F}$.
The footpaths are used to model transfers.
We therefore postpone their description until we describe transfers.
A \emph{stop} is a position outside of a train where a traveler can stand. 
At a stop, trains can halt and passengers can enter or leave trains.
A \emph{trip} is a scheduled train. 
A \emph{connection} is a train that drives from one stop to another stop without intermediate halt. 
Formally, a connection $c$ is a five tuple $(c_\depstop,c_\arrstop,c_\deptime,c_\arrtime,c_\trip)$
We refer to these attributes as $c$'s \emph{departure stop}, \emph{arrival stop}, \emph{departure time}, \emph{arrival time}, and \emph{trip}, respectively.
We require from every connection $c$ that $c_\depstop\neq c_\arrstop$ and $c_\deptime<c_\arrtime$.

All connections with the same trip form a set. 
We require that this set can be ordered into a sequence $c^1,c^2\ldots c^k$ such that $c^i_\arrstop=c^{i+1}_\depstop$ and $c^i_\arrtime<c^{i+1}_\deptime$ for every $i$.
In a slight abuse of notation, we sometimes identify a trip with its corresponding sequence of connections.

\paragraph{Examples.}

Examples for stops are the train main stations, such as ``Karlsruhe Hbf''.
Other examples include subway or tram stations.

Trips include high speed trains, subway trains, trams, buses, ferries, and more.
An example for a trip is the ``ICE 104'' from Basel to Amsterdam that departs at 15:13 on the 2-nd of August 2016. 
Note, that the description ``ICE 104'' without the departure time does not uniquely identify a trip as such a train exists on every day of August 2016.
In our model, there is a trip for every day, even though these trips share the same sequence stop and the operator refers to all trains by the same name.

Pick one of the ``ICE 104'' trips and name it $x$.
The first three stops at which $x$ halts are Basel, Freiburg, and Offenburg. 
There is a connection with departure stop Basel, arrival stop Freiburg, and trip $x$.
There further is a connection with departure stop Freiburg, arrival stop Offenburg, and trip $x$.
However, there is no connection with departure stop Basel, arrival stop Offenburg, and trip $x$, as we require that the train of a connection does not halt at an intermediate stop.

\paragraph{Transfers.}

A traveler standing at stop $s$ at the time point $\tau$ can be described using a pair $(s,\tau)$.
To lighten our notation, we denote these pairs as $s@\tau$.
Denote by $P$ the infinit set of these pairs.
A \emph{transfer model} is a relation on $P$, which we denote using the $\rightarrow$ symbol.
A traveler sitting in an incoming connection $c$, wishing to transfer to an outgoing connection $c'$ of another trip, can do so by definition if and only if $c_\arrstop@c_\arrtime \rightarrow c'_\depstop@c'_\deptime$ holds.

Many transfer models exist and the details vary significantly across the literature. 
Unfortunately, there is no consent on what the best model is.
In the following, we focus our description on the model used in our work, which is based upon footpaths.
We also briefly discuss the differences to other models.

A \emph{footpath} $f$ is a triple $(f_\depstop,f_\arrstop,f_\dur)$, which we refer to as $f$'s \emph{departure stop}, $f$'s \emph{arrival stop}, and $f$'s \emph{duration}.
We require all footpath durations to be positive, i.e., $f_\dur > 0$.
The set of footpaths $\mathcal{F}$ is the last element of the quadruple that characterizes timetables.
These footpaths can be viewed as weighted, directed \emph{footpath graph} $G_\mathcal{F}=(\mathcal{S},\mathcal{F})$, where the stops are the nodes, the footpaths the arcs, and the duration the weights.
We define the transfer relation as follows: $a@\tau_a \rightarrow b@\tau_b$ holds, if and only if there is a path from $a$ to $b$ whose length is at most $\tau_b - \tau_a$.

Having a large connected footpath graph makes the considered problems significantly harder than having only loosely connected components. 
Following \cite{dpw-rbptr-14}, we therefore introduce two restrictions on the footpath graph.
It must be transitively closed and fulfill the triangle inequality.
Transitively closed means that if there is an edge $ab$ and an edge $bc$, then there is an edge $ac$. 
The triangle inequality further requires that $ab_\dur + bc_\dur \ge ac_\dur$.
From these two properties one can show that if there is a path from $a$ to $b$, then there is a shortest $ab$-path with a single edge.
The transfer relation in this special case therefore boils down to
\[ (a@\tau_a \rightarrow b@\tau_b) \iff \exists f \in F : \tau_b - \tau_a \ge f_\dur \text{ and } a=f_\depstop\text{ and }b=f_\arrstop \]
which allows us to limit our searches to single-edge paths.
These restrictions come at a price.
In each connected component there is a quadratic number of edges because of the transitive closure.
As a quadratic memory consumption is prohibitive in practice, we can therefore have no large components.

Our footpath-based transfer model is transitive, i.e., if $a@\tau_a \rightarrow b@\tau_b$ and $b@\tau_b \rightarrow c@\tau_c$ then $a@\tau_a \rightarrow c@\tau_c$.
We exploit this property in our algorithms.
While transfer model transitivity sounds like a very reasonable and desirable property, there is a common class of competitor transfer models that do not have it.
They are similar to our model, except that instead of requiring transitive closure and triangle inequality, they limit the maximum path length by some constant $m$.
It possible that one can walk within time $m$ from $a$ to $b$ and within time $m$ from $b$ to $c$ but require longer than time $m$ to get from $a$ to $c$, which demonstrates that transitivity breaks. 
The missing transitivity is the main reason why we chose a different model.

An interesting special case are \emph{loops} in the footpath graph. 
Note that without a loop at a stop $s$, a traveler cannot exit at $s$ and enter another train at $s$. 
In practice, all stops have therefore loops.
The duration of the loop footpath at stop $s$ is called the change time\footnote{Several other works refer to $s^\change$ as minimum change time.} $s^\change$.
Some competitor works even assume that there are no footpaths beside these loops, which is a significant restriction compared to our model.
Footpaths that are not loops are \emph{interstop footpaths}.

Our transfer model is in general not reflexive, i.e., it is possible that there are stops $s$ and time points $\tau$ such that $s@\tau\not\rightarrow s@\tau$.
However, one can study the special case of reflexive transfer models.
This requirement translates to every stop having a change time of 0.
The London benchmark instance of \cite{dpw-rbptr-14}, which we also use, has this additional property.

\paragraph{Examples.}

In our Germany instance, the Karlsruhe main station is modeled as two stops. 
There is a stop that represents the main tracks used by the long distance trains.
Further, there is a stop that represents the tracks where the local trams halt.
Both are connected using a footpath per direction.
Further, both stops have loop footpaths.
The loop of the main track stop has a duration of 5min and the loop of the local tram stop has a duration of 4min.
The footpaths between the two stops have a duration of 6min.

Transferring between local trams is therefore possible within 4min.
To transfer between long distance trains, the traveler needs 5min.
Finally, to transfer from tram to long distance train 6min are needed.

Other main stations can be modeled using more stops.
For example many stations have an additional stop per subway line.

Within cities, it can make sense to insert footpaths between neighboring tram stops.
However, one has to be careful to not create large connected components in the footpath graph by doing so.

It is also possible to model stations in greater detail using a stop per platform.
The London instance uses this approach.
This approach gives more precise transfer times at the expense of more stops.

\subsection{Journeys}

A journey describes how a passenger can travel through a timetable network.
They are composed of legs, which are pairs of connections $(l^i_\enter,l^i_\exit)$ within the same trip.
$l^i_\enter$ must appear before $l^i_\exit$ in the trip.
Formally, a journey consists of alternating sequence of legs and footpaths $f^0,l^0,f^1,l^1\ldots f^{k-1},l^{k-1},f^k$.
A journey must start and end with a footpath.
All intermediate transfers must be feasible according to the transfer model, i.e., for all $i$, $(l^{i-1}_\exit)_\arrstop @ (l^{i-1}_\exit)_\arrtime \rightarrow (l^i_\enter)_\depstop @ (l^i_\enter)_\deptime$ must hold.
We refer to $f^0$ as \emph{initial footpath} and to $f^k$ as \emph{final footpath}.
The remaining footpaths are called \emph{transfer footpaths}.
Further, for a journey $j$ we refer to $f^0_\depstop$ as \emph{$j$'s departure stop}, to $f^k_\arrstop$ as the \emph{$j$'s arrival stop}, to $(l^0_\enter)\deptime - f^0_\dur$ as \emph{$j$'s departure time}, to \emph{$(l^{k-1}_\exit)\arrtime + f^k_\dur$ as $j$'s arrival time}, and to $k$ as \emph{$j$'s number of legs}.
We also use $j_\leg$ to refer to the number of legs, i.e., $k$.
Finally, we refer to $j_\arrtime-j_\deptime$ as \emph{$j$'s travel time}.
Formally, journeys are allowed to consist of a single footpath and no leg.
However, we forbid this special case in certain problem settings to avoid unnecessary, simple but cumbersome special cases in our algorithms.

A journey $j$ that is missing its initial footpath, i.e., a sequence $l^0,f^1,l^1\ldots f^{k-1},l^{k-1},f^k$ is called a \emph{partial journey}.
We say that $j$ departs in the connection $l^0_\enter$.

Note, that the number of legs and the number of transfers differ slightly.
For every journey with at least one leg, the number of transfers is $j_\leg-1$.
The numbers are therefore essentially the same, except for a subtle difference.
A journey without leg has 0 legs but also has 0 transfers and not -1 transfers.
Counting legs eliminates some special cases in our algorithms and avoids some -1/+1-operations.
Hence, for simplicity, we count legs.

\subsection{Considered Problem Settings}

In this section, we describe most problem settings studied in this paper. 
Several of these problems are defined in terms of Pareto-optimization.
We therefore first recapitulate the definition of Pareto-optimal and domination and then state the problems considered in our paper.
Section~\ref{sec:MEAT} introduces another problem setting called Minimum Expected Arrival Time problem.
As its details are more involved, we introduce the problem setting in its own section.

\begin{definition}
A tuple $x$ \emph{dominates} a tuple $y$ if there is no component in which $y$ is strictly smaller than $x$ and there is component in which $x$ is strictly smaller than $y$.
\end{definition} 

Pareto-optimal is defined in terms of domination.

\begin{definition}
Denote by $P$ a multi-set of $n$-dimensional tuple with scalar components.
A tuple $x$ is \emph{Pareto-optimal} with respect to $P$, if no other tuple $y\in P$ exists, such that $y$ dominates $x$.
\end{definition} 

In our setting, the tuples are journey attributes such as a journey's travel time.
$P$ is the set of attribute-tuples of all journeys.

The easiest problem, that we consider, asks when a traveler will arrive the earliest possible.
Formally, it can be stated as follows:

\begin{problem}
\problemtitle{Earliest Arrival Problem}
\probleminput{timetable, source stop $s$, target stop $t$, source time $\tau$}
\problemoutput{The minimum arrival time over all journeys that depart after $\tau$ at $s$ and arrive at $t$}
\end{problem}

While simple, the earliest arrival problem has several downsides.
For one, a traveler often does not have a fixed departure time, but is flexible and has a range of possible departure times.
One can resolve this issue by iteratively solving the earliest arrival problem with varying source times.
Fortunately, we can do better and therefore formalize the aggregated problem as follows:

\begin{problem}
\problemtitle{Earliest Arrival Profile Problem}
\probleminput{timetable, source stop $s$, target stop $t$, minimum departure time $\tau_s$, maximum arrival time $\tau_t$}
\problemoutput{The set of all $(j_\deptime,j_\arrtime)$ over journeys $j$ such that
\begin{itemize}
\item $j$ departs not before $\tau_s$ at $s$,
\item $j$ arrives not after $\tau_t$ at $t$,
\item the pair $(-j_\deptime,j_\arrtime)$ is Pareto-optimal among all journeys, and
\item $j$ contains at least one leg.
\end{itemize}
}
\end{problem}

\begin{figure}

\begin{center}
\includegraphics{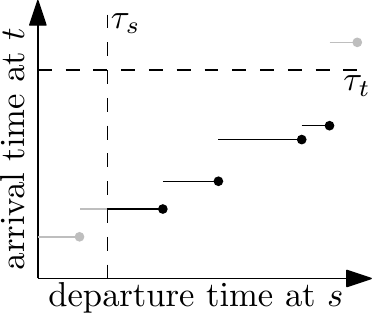}
\end{center}

\caption{Profile function that maps the departure times at a stop $s$ onto the arrival times at stop $t$. The black dots represent the solution to the earliest arrival profile problem. Only the black part needs to be computed. The gray part is excluded by the minimum departure time or maximum arrival time.}
\label{fig:profile-plot}
\end{figure}

The result of the profile problem can be represented using a plot such as the one in Figure~\ref{fig:profile-plot}. 
The result is a compact representation of the functions that maps a departure time at $s$ onto the earliest arrival time at $t$.
We refer to this function as \emph{profile function}.
Formulated differently, the profile problem asks to simultaneously solve the earliest arrival problem for all source times.

We require $j$ to have at least one leg, to be able to guarantee that the profile function is a step function.
Dropping this restriction, can break this property if $s$ and $t$ are connected via a footpath $f$.
At least in our setting, handling such a situation is trivial but requires special case handling in our algorithm.
To simplify our descriptions and to focus on the algorithmically interesting aspects, we decided to forbid journeys without leg.

\begin{figure}

\begin{center}
\includegraphics{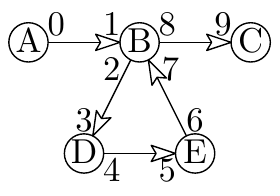}
\end{center}

\caption{Example for an ``optimal'' journey that visits a stop twice. Circles depict stops, arrows depict connections and are annotated with their departure and arrival times. The journey $A{\rightsquigarrow}B{\rightsquigarrow}D{\rightsquigarrow}E{\rightsquigarrow}B{\rightsquigarrow}C$ visits stop $B$ twice and has a minimum arrival time. The journey $A{\rightsquigarrow}B{\rightsquigarrow}C$ has the same arrival time but uses fewer legs.}
\label{fig:loop-journey}
\end{figure}

An issue common with the earliest arrival problem and with its profile counterpart is that solely optimizing arrival time can lead to very absurd but ``optimal'' journeys.
For example, Figure~\ref{fig:loop-journey} depicts a journey that is ``optimal'' with respect to its arrival time but visits a stop twice.
Similarly, it ``optimal'' journeys exist that enter a trip multiple times.
When computing earliest arrival journeys and not just their arrival time, one therefore usually also require that the journeys visit no stop or trip twice.

A simple solution to this problem consists of picking among all journeys with a minimum arrival time one that minimizes the number legs.
This implies that no stop or trip is used twice.
We say that the first optimization criterion is arrival time and the second criterion is the number of legs.
This slight change is enough to guarantee that no stop is visited twice.

While this small change solves many transfer-related problems, some remain.
Suppose, for example that there are two journeys whose arrival times differ by one second but the earlier one needs significantly more legs.
In this case one would like to pick the journey that arrives slightly later. 
This problem can be mitigated by rounding the arrival times at the target stop.
However, in many application one wants to find both journeys.
We therefore also consider the following problem setting.

\begin{problem}
\problemtitle{Pareto Profile Problem}
\probleminput{timetable, source stop $s$, target stop $t$, minimum departure time $\tau_s$, maximum arrival time $\tau_t$, maximum number of legs $\max_\leg$}
\problemoutput{The set of all $(j_\deptime,j_\arrtime,j_\leg)$ over journeys $j$ such that
\begin{itemize}
\item $j$ departs not before $\tau_s$ at $s$,
\item $j$ arrives not after $\tau_t$ at $t$, 
\item $j$ has at most $\max_\leg$ legs,
\item the pair $(-j_\deptime,j_\arrtime,j_\leg)$ is Pareto-optimal among all journeys, and
\item $j$ contains at least one leg.
\end{itemize}
}
\end{problem}

Besides the profile problem setting, we also consider \emph{range problem} variants.
In these, we set $\tau_t$ to $\tau_s + 2\cdot (x-\tau_s)$, where $x$ is the earliest arrival time.
Formulated differently, we are only interested in journeys that are at most two times as long as possible.
The solution to the range problems is a subset of the solution to the profile problems.
The range problems can therefore often be solved faster.
Fortunately, travelers usually do not want to arrive significantly later than the earliest arrival time. 
The solution to the range problem thus often consists of the journeys that actually interest a traveler.
The range problem special cases are therefore of high practical relevance.

Beside determining the attributes of optimal journeys, i.e., departure time, arrival time, and number of legs, we also consider the problem of computing corresponding journeys in Sections \ref{sec:ea-extraction} and \ref{sec:extraction}.
Note that optimal journeys are usually not unique.
There usually are multiple journeys for a combination of departure time, arrival time, and number of legs.
We regard all of them as being equal and only extract one of them.
Extracting all journeys for a combination is a different problem setting. 

\section{Earliest Arrival Connection Scan}
\label{sec:base-csa}

In this section, we describe the earliest arrival Connection Scan variant. 
It assumes that the connections are stored as array of quintuples that are sorted by departure time.
Further, the footpaths must be stored in a data structure that allows an efficient iteration over the incoming and outgoing footpaths of a stop, such as for example an adjacency array.  
Similar to Dijkstra's algorithm, CSA maintains a tentative arrival time array, that stores for each stop the earliest known arrival time.
A connection is called \emph{reachable} if there is a way for the traveler to sit in the connection.
Contrary to Dijkstra's algorithm, ours does not employ a priority queue.
Instead, it iterates over all connections increasing by departure time.
The algorithm tests for every connection whether it is reachable.
For each reachable connection, the algorithm adjusts the tentative arrival times of the stops reachable by foot from the connection's arrival stop.
After the execution of our algorithm, the output is $t$'s tentative arrival time.
Contrary to most adaptations of Dijkstra's algorithm, our algorithm touches more connections.
But the work required per connection does not involve a priority queue operation and is therefore significantly faster.

\begin{figure}
\begin{algorithm2e}[H]
\lFor{all stops $x$}{$S[x]\gets \infty$}
\lFor{all trips $x$}{reset $T[x]$}
\lFor{all footpaths $f$ from $s$}{$S[f_\arrstop] \gets \tau + f_\dur$}
\BlankLine
\For{all connections $c$ increasing by $c_\deptime$}{
	\If{$T[c_\trip]$ is set or $S[c_\depstop]\le c_\deptime$}{
		raise $T[c_\trip]$\;
		\For{all footpaths $f$ from $c_\arrstop$}{
			$S[f_\arrstop] \gets \min \{ S[f_\arrstop], c_\arrtime + f_\dur \}$\; 
		}
	}
}
\end{algorithm2e}

\caption{Unoptimized earliest arrival Connection Scan algorithm. $s$ is the source stop and $\tau$ the source time.}
\label{fig:unoptimized-earliest-arrival-connection-scan}
\end{figure}

Our algorithm maintains two arrays $S$ and $T$.
The array $S$ stores for every stop the tentative arrival time.
The array $T$ stores for every trip a bit indicating whether the traveler was able to reach any of the connections in the trip.
Testing whether a connection $c$ is reachable boils down to testing, whether $S[c_\depstop]\le c_\deptime$ or $T[c_\trip]$ is set.
To adjust the tentative arrival times, our algorithm relaxes all footpaths outgoing from $c_\arrstop$.
The algorithm is described in pseudo-code form in Figure~\ref{fig:unoptimized-earliest-arrival-connection-scan}.

\subsection{Optimizations}

\begin{figure}
\begin{algorithm2e}[H]
\lFor{all stops $x$}{$S[x]\gets \infty$}
\lFor{all trips $x$}{reset $T[x]$}
\lFor{all footpaths $f$ from $s$}{$S[f_\arrstop] \gets \tau + f_\dur$}
\BlankLine
Find first connection $c^0$ departing not before $\tau$ using a binary search\;
\For{all connections $c$ increasing by $c_\deptime$ starting at $c^0$}{
	\If{$S[t]\le c_\deptime$}{
		Algorithm is finished\;
	}
	\If{$T[c_\trip]$ is set or $S[c_\depstop]\le c_\deptime$}{
		raise $T[c_\trip]$\;

		\If{$c_\arrtime < S[c_\arrstop]$}{
			\For{all footpaths $f$ from $c_\arrstop$}{
				$S[f_\arrstop] \gets \min \{ S[f_\arrstop], c_\arrtime + f_\dur \}$\; 
			}
		}
	}
}
\end{algorithm2e}
\caption{Optimized earliest arrival Connection Scan algorithm. $s$ is the source stop, $\tau$ the source time, and $t$ that target stop.}
\label{alg:optimized-earliest-arrival-connection-scan}
\end{figure}

In this subsection, we describe three optimizations to the earliest arrival Connection Scan algorithm.
Figure~\ref{alg:optimized-earliest-arrival-connection-scan} presents pseudo-code that incorporates all three optimizations.
In the following, $c$ always denotes the connection currently being processed.

\paragraph{Stopping criterion.}
We can abort the execution of the algorithm as soon as $S[t] \le c_\deptime$.
This is correct because processing a connection $c$ never assigns a value below $c_\deptime$ to any tentative arrival time.
Further, as we process the connections increasing by $c_\deptime$, it follows that $S[t]$ will not be changed by our algorithm after the inequality holds.

\paragraph{Starting criterion.}
No connection departing before the source time $\tau$ is reachable, as for every journey $j$, $\tau \le j_\deptime < j_\arrtime$ must hold.
The proposed optimization exploits this.
It runs a binary search to determine the first connection $c^0$ departing no later than $\tau$.
The iteration is started from $c^0$ instead of the first connection in the timetable.

\paragraph{Limited Walking.}
If $S[c_\arrstop]$ cannot be improved even with an instant transfer, i.e., $S[c_\arrstop]\le c_\arrtime$ holds, then no tentative arrival time can be improved.
The optimization consist of not iterating over the outgoing footpaths of $c_\arrstop$ in this case.

The correctness of this optimization relies on the transitivity of the transfer model.
Denote by $y=c_\arrstop$. 
As $S[y]\neq\infty$ a journey $j$ ending at $y$ has already been found.
Denote by $f^{xy}$ the last footpath of $j$ departing at $x$.
Note, that it is possible that $x=y$ and that $f^{xy}$ is a loop.
For every outgoing footpath $f^{yz}$ of $y$ to some stop $z$, there exists a footpath $f^{xz}$ from $x$ to $z$ such that $f^{xz}_\dur \le f^{xy}_\dur + f^{yz}_\dur$.
We can replace the last footpath of $j$ by $f^{xz}$ and have obtained a journey arriving at $z$ no later than the journey involving $c$.
As this argumentation works for all outgoing footpaths, no tentative arrival time can be improved.
Iterating over the outgoing footpaths is therefore superfluous.
The optimization is thus correct.

\begin{figure}
\begin{center}
\includegraphics{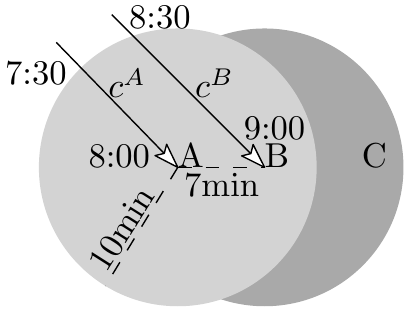}
\end{center}
\caption{Counterexample for the correctness of limited walking optimization in combination with a maximum path length transfer model. The example includes three stops $A$, $B$, and $C$, two connections $c^A$, and $c^B$ with annotated departure and arrival times, and walking radii of 10min. The light gray area is reachable by foot from $A$. The dark grey area is reachable by foot from $B$ but not from $A$.}
\label{fig:limited-walking-max-path-counter-example}
\end{figure}

Note that the limited walking optimization crucially depends on the transitivity of the transfer model.
For example, it does not hold for a transfer model with a maximum path length.
Consider the example depicted in Figure \ref{fig:limited-walking-max-path-counter-example}.
Assume that both $c^A$ and $c^B$ are reachable connections.
When processing the connection $c^A$ arriving at $A$ the tentative arrival time at $B$ is set to 8:07.
However, the tentative arrival time at $C$ remains $\infty$ as the path is too long.
Because the tentative arrival time at $B$ is smaller than 9:00 the limited walking optimization activates when processing $c^B$.
The tentative arrival time at $C$ therefore remains at $\infty$, which is clearly incorrect.

\subsection{Journey Extraction}
\label{sec:ea-extraction}

The algorithm described in the previous section only computes the earliest arrival time.
In this section, we describe how to compute an earliest arrival journey in a post processing step.
Our algorithm guarantees that the extracted journey visits no stop nor trip twice.

The algorithm comes in two variants.
The first variant augments the data structures used during the Connection Scan with additional \emph{journey pointers} that can be used to reconstruct a journey.
The second variant leaves the earliest arrival scan untouched but needs to perform more complex tasks to reconstruct an earliest arrival journey. 
The trade-off between the two variants is that the former is conceptually slightly more straight-forward and therefore easier to implement.
Further, the former has a lower extraction running time, which comes at the cost of a higher scan running time.
Finally, the later requires additional data structures, which must be computed in a fast preprocessing step. 
If only a journey towards one target stop should be extracted, then the later variant is faster.
If journeys from one source stop to many target stops should be extracted, the former can be faster.
 
\subsubsection{With Journey Pointers}

\begin{figure}
\begin{algorithm2e}[H]
\lFor{all stops $x$}{$S[x]\gets \infty$}
\lFor{all trips $x$}{$T[x]\gets \bot$}
\lFor{all stops $x$}{$J[x]\gets (\bot,\bot,\bot)$}
\lFor{all footpaths $f$ from $s$}{$S[f_\arrstop] \gets \tau + f_\dur$}
\BlankLine
\For{all connections $c$ increasing by $c_\deptime$}{
	\If{$T[c_\trip]\neq \bot$ is set or $S[c_\depstop]\le c_\deptime$}{
		\If{$T[c_\trip]=\bot$}{
			$T[c_\trip]\gets c$\;
		}
		\For{all footpaths $f$ from $c_\arrstop$}{
			\If{
				$c_\arrtime + f_\dur < S[f_\arrstop]$
			}{
				$S[f_\arrstop] \gets c_\arrtime + f_\dur$\;
				$J[f_\arrstop] \gets (T[c_\trip], c, f)$\;
			}
		}
	}
}
\BlankLine
$j\gets\{\}$\;
\While{$J[t] \neq (\bot,\bot,\bot)$}{
	Prepend $j$ with $J[t]$\;
	$t\gets J[t]^\enter_\depstop$\;
}
Prepend $j$ with the footpath from $s$ to $t$\;
Output $j$\;
\end{algorithm2e}

\caption{Earliest arrival Connection Scan algorithm with journey extraction. $s$ is the source stop. $\tau$ the source time, and $t$ the target stop.}
\label{fig:unoptimized-earliest-arrival-connection-scan-with-journey-pointer}
\end{figure}

Our algorithm, which is illustrated in Figure~\ref{fig:unoptimized-earliest-arrival-connection-scan-with-journey-pointer}, stores for every stop $x$ a triple $J[x]$ of final enter connection, final exit connection, and final footpath of an earliest arrival journey towards $x$.
We refer to this triple as \emph{journey pointer}.
If no optimal journey exists, then the journey pointer is set to an invalid value.

A journey $j$ from a source stop $s$ to a target stop $t$ can be constructed backwards.
Initially $j$ is empty.
If $t$ has a valid journey pointer, we the prepend $t$'s journey pointer to to $j$.
Further, we set $t$ to the departure stop of the journey pointer's enter connection and iterate.
If $t$ does not have a valid journey pointer, we prepend $j$ with a footpath from $s$ to $t$ and the journey extraction terminates.

\paragraph{Journey Pointer Construction.}
When a tentative arrival time is modified, our algorithm stores a corresponding journey pointer.
To this end, our algorithm must determine the three elements of the triple.
Determining the exit connection and the final footpath is easy.
These are the values denoted by the variables $c$ and $f$ in the code depicted Figure~\ref{fig:unoptimized-earliest-arrival-connection-scan-with-journey-pointer}.
Computing the enter connection is more difficult.

We replace the bit array $T$ in the base algorithm by an array that contains for every trip a connection ID.
This connection ID indicates the earliest connection reachable in a trip.
It may be invalid, if no connection was reached.
The ID being valid corresponds to the bit being set in the base algorithm.
We set an ID when a trip is first reached.

It remains to show that the extracted journey does not visit a stop or trip twice.
A trip cannot be visited twice by the extracted journey because $T$ is set to the first connection reachable in a trip.
A stop cannot be visited twice as our algorithms stores at each stop the first journey pointer found towards it.
Fortunately, a journey pointer leading to a journey with a loop cannot be the first.

\subsubsection{Without Journey Pointers}

A journey can be extracted without storing journey pointers.
However, additional data structures are necessary.

\paragraph{Additional Datastructures.}

Our algorithm needs to enumerate the connections in a trip that precede a given connection.
We therefore construct an adjacency array that maps a trip ID onto the IDs of the connections in the trip.
The connections are sorted by position in the trip. 
Our algorithm can thus enumerate all connections in a trip rapidly and stop once the given connection is found.

Further, our algorithm needs to enumerate the connections arriving at a given stop at a given timepoint.
We therefore construct a second adjacency array that maps a stop ID onto the IDs of the connections arriving at the stop.
The connections are sorted by arrival time.
We can use a binary search to efficiently enumerate all requested connections.

\paragraph{Extraction.}

Our algorithm works similar to the one using journey pointers.
However, the journey pointer is generated on the fly.
We therefore need a subroutine to determine a triple of enter connection $c^\enter$, exit connection $c^\exit$, and final footpath $f$.
We start by constructing a set of candidates for $c^\exit$.
This set is then pruned.
Finally, our algorithm iterates over the candidates and tries to find a corresponding $c^\enter$.

To generate the candidate set our algorithm enumerates all incoming footpaths $f$ of the stop $t$.
For every $f$, all connections arriving at $f_\depstop$ at $S[t]-f_\dur$ are added to the candidate set.

A candidate can only be a valid $c^\exit$ if it is reachable.
If it is reachable then the trip bit must be set.
We can therefore prune all candidate connections $c$ for which $T[c_\trip]$ is false.
Note, that the bit being set does not imply that the candidate is reachable.
It is also possible that only a later connection in the same trip is reachable.

Finally, our algorithm iterates over the remaining candidates $c$.
For each candidate, it enumerates all connections $x$ in $c_\trip$ not after $c$.
It then checks whether $S[x_\depstop]\le S[x_\deptime]$.
If it holds, then $x$ is a valid enter connection, $c$ a valid exit connection and $f$ a valid final footpath.
Further, as $x$ is the first connection in the trip, the extracted journey cannot visit a trip twice.
As our approach constructs a journey for the earliest timepoint where $t$ is reachable, we can guarantee that the extracted journey does not visit a stop twice.

If no journey pointer can be generated, then $t$ was reached by foot from $s$.
This corresponds to $J[t]$ being invalid in the algorithm of Figure~\ref{fig:unoptimized-earliest-arrival-connection-scan}.

\subsection{Experiments}

We experimentally evaluated the earliest arrival Connection Scan algorithm and compare it with competing algorithms.
Beside, only measuring the query running times, we also report how much time is needed to setup the data structures.
The setup time is an upper bound to the time needed to update a timetable.

The section is structured as follows: 
We first describe the machines on which we run our experiments. 
We then describe the test instances and how we generate our test queries.
Afterwards, we report the running times needed by the Connection Scan algorithm.
Finally, we compare the achieved running times with related work.

\subsubsection{Experimental Setup}

\label{sec:experimental-setup}

\paragraph{Machine.}
Unless specified otherwise, we ran all experiments on a single pinned thread of an Intel Xeon E5-1630v3, with 10 MiB of L3 cache and 128\,GiB of DDR4-2133MHz.
This is a CPU with Haswell architecture.
Some experiments were executed on an older dual 8-core Intel Xeon E5-2670, with 20 MiB of L3 cache and 64\,GiB of DDR3-1600 RAM, a CPU with Sandy Bridge architecture.
Hyperthreading was deactivated in all experiments.
Our implementation is written in C++ and is compiled using g++ 4.8.4 with the optimization flags \texttt{-O3 -march=native}.

\paragraph{Instances.}

\begin{table}
\begin{center}
\begin{tabular}{rrrrrr}
\toprule
Instance & Stops & Connections & Trips & Routes & Interstop Footpaths\\
\midrule
Germany & 252\,374 & 46\,218\,148 & 2\,395\,656 & 248\,261 & 103\,535 \\
London & 20\,843 & 4\,850\,431 & 125\,537 & 2\,135 & 45\,652 \\
\bottomrule
\end{tabular}
\end{center}

\caption{Instance sizes.}

\label{tab:instances}
\end{table}

We performed our experiments on two main benchmark instances.
Table~\ref{tab:instances} reports the sizes. 
The first instance is based on the data of \url{bahn.de} during winter 2011/2012. 
The data was provided to use by Deutsche Bahn (DB), the German national railway company.
We thank DB for making this data accessible to us for research purposes.
The data contains European long distance trains, German local trains, and many buses inside of Germany. 
The data includes vehicles of local operators beside DB.
The raw data contains for every vehicle a day of operation.
Unfortunately, no day exists every local operator operates. 
The planning horizon of some operators ends before the reported data of other operators begins.
To avoid holes in our timetable, we therefore extract all trips regardless of their day of operation and assume that they depart within the first day.
Our extracted instance contains therefore more connections per day than the instance in productive use.
Further, to support night trains, we consider two successive identical days.
The raw data contains footpaths. 
We did not generate additional ones based upon geographic positions but did add footpaths to make the graph transitively closed.
We removed data errors such as exactly duplicated trips, vehicles driving at more than 300 km/h or footpaths at more than 50 km/h. 

The second instance is based on open data made available by Transport for London (TfL). 
The raw input data is available in the London data store\footnote{\url{http://data.london.gov.uk}}.
We thank TfL for making this data openly available.
The data includes tube (subway), bus, tram, Dockland Light Rail (DLR).
The data corresponds to a Tuesday of the periodic summer schedule of 2011.
In contrast to the Germany instance, the London instance thus only contains data for a single day.
Stops correspond to platforms in this data set.
As a consequence all change times are zero, i.e., the transfer model is reflexive.
This data set is the main instance used in \cite{dpw-rbptr-12}, one of our main competitor algorithms.
We removed some obvious data errors from the data.
The instance sizes we report are therefore slightly smaller than in \cite{dpw-rbptr-12}.
 
\paragraph{Test Query Generation.}
To evaluate our algorithms, we generate random test queries.
The source and target stops are chosen uniformly at random.
The source time is chosen uniformly at random within the first 24 hours.
Unless noted otherwise all reported running times are averaged over $10^4$ queries.

\subsubsection{Earliest Arrival Connection Scan}

\begin{table}
\begin{center}
\begin{tabular}{cccccr}
\toprule
& Start & Stop & Limited & Journey  & Running \\
Instance & Crit. & Crit. & Walk. & Extraction & Time [ms]\\
\midrule
Germany & $\circ$ & $\circ$ & $\circ$ & $\circ$ & 329.0 \\
Germany & $\bullet$ & $\circ$ & $\circ$ & $\circ$ & 298.9 \\
Germany & $\bullet$ & $\bullet$ & $\circ$ & $\circ$ & 67.9 \\
Germany & $\bullet$ & $\bullet$ & $\bullet$ & $\circ$ & 44.9 \\
Germany & $\bullet$ & $\bullet$ & $\bullet$ & $\bullet$ & 47.1 \\
\midrule
London & $\circ$ & $\circ$ & $\circ$ & $\circ$ & 41.2 \\
London & $\bullet$ & $\circ$ & $\circ$ & $\circ$ & 37.9 \\
London & $\bullet$ & $\bullet$ & $\circ$ & $\circ$ & 2.7 \\
London & $\bullet$ & $\bullet$ & $\bullet$ & $\circ$ & 1.2 \\
London & $\bullet$ & $\bullet$ & $\bullet$ & $\bullet$ & 1.3 \\
\bottomrule
\end{tabular}

\end{center}
\caption{
Earliest arrival Connection Scan running times.
}
\label{tab:ea-time}
\end{table}

We experimentally evaluated the earliest arrival Connection Scan algorithm and report the average running time in Table~\ref{tab:ea-time}.
We successively activate the proposed optimizations.
Further, we evaluate the running time of the journey extraction without journey pointers.

The start and stop criteria drastically reduce the running times. 
The explanation is that significantly fewer connections have to be scanned.
On the London instance the speedup is 15 times whereas the speedup on the Germany instance is ``only'' 5.
This is due to the differences in journey lengths.
In London a traveler needs on average less time to traverse the whole network than in Germany.
The stop criterion therefore activates sooner reducing the number of scanned connections.
The limited walking optimization further reduces running times by 1.5 to 2.0 times.

Finally, we report the running time needed to perform a journey extraction in addition to the earliest arrival Connection Scan.
As we only extract a single journey per scan, we use the extraction process that does not store journey pointers.
The extraction process is very faster compared to the scan.
On the Germany instance, it only need about 1.2ms and on the London instance 0.1ms.

\begin{table}
\begin{center}
\begin{tabular}{crr}
\toprule
Instance & Sort [s] & Journey [s]\\
\midrule
Germany & 3.56 & 6.15 \\
London & 0.35 & 0.39 \\
\bottomrule
\end{tabular}
\end{center}
\caption{
Datastructure construction running time averaged over $100$ runs. 
``Sort'' is the time needed to sort the connection array by departure time. 
``Journey'' is the additional time needed to construct the journey extraction data structures. 
}
\label{tab:ea-construction}
\end{table}

\paragraph{Datastructure Construction.}

In Table Table~\ref{tab:ea-construction}, we report the running time needed to sort the connection array and the running time needed to construct the journey extraction data structures.
To avoid accelerating the sort algorithm by providing it with nearly sorted data, we randomly permute the array before sorting it. 
We use GCC's std::sort implementation.
If the timetable significantly changes, then these two steps need to be rerun.
If the changes are only small, then it is probably faster to patch the existing data structures.

In practice, when delays occurs, the operator needs to simulate how the delay propagates through the network.
This propagation is in practice probably slower, than the few seconds needed to construct the data structures needed by our algorithms.

\paragraph{Comparison with Related Work.}

\begin{table}
\begin{center}
\begin{tabular}{cccr}
\toprule
Instance & Algorithm & Pareto & Running Time [ms]\\
\midrule
Germany & TED & $\circ$ & 1\,996.6 \\
Germany & TD & $\circ$ & 448.5 \\
Germany & TD-col & $\circ$ & 163.3 \\
Germany & RAPTOR & $\bullet$ & 325.8 \\
Germany & CSA & $\circ$ & 44.9 \\
\midrule
London & TED & $\circ$ & 29.3 \\
London & TD & $\circ$ & 9.5 \\
London & TD-col & $\circ$ & 3.7 \\
London & RAPTOR & $\bullet$ & 6.4 \\
London & CSA & $\circ$ & 1.2 \\
\bottomrule
\end{tabular}

\end{center}
\caption{
Comparison with related work with respect to the earliest arrival time problem.
}
\label{tab:ea-related-work}
\end{table}

In Table~\ref{tab:ea-related-work}, we compare our algorithms with related work. 
The employed implementations are based upon the code of~\cite{dpw-rbptr-14}.
All competitors are run with stopping criterion active.

We compare the Connection Scan algorithm's running times against three extensions of Dijkstra's algorithm and RAPTOR.
The first extension is based on a time-expanded graph model.
The second uses a time-dependent graph model.
We refer to~\cite{pswz-emtip-08} for a detailed exposition of these models.
The third uses an optimized time-dependent graph model, proposed in \cite{dkp-pcbcp-12}, that merges nodes based on colored timetable elements.
Finally, we compare against RAPTOR~\cite{dpw-rbptr-14}, an algorithm that does not employ a graph based model.
Instead, it operates directly on the timetable, similarly to the Connection Scan algorithm.

We experimentally compare the performance of the algorithms with respect to the earliest arrival time problem.
However, RAPTOR does not fit precisely into this category.
It is designed in a way that inherently optimizes the number of transfers in the Pareto-sense.
It can and must thus solve a more general problem.
It does not benefit from restricting the problem setting.
We therefore report its running times alongside the other earliest arrival time algorithms.

Table~\ref{tab:ea-related-work} shows that the non-graph based algorithms clearly dominate the base versions of the time-dependent and time-expanded extensions of Dijkstra's algorithm.
The time-dependent extension can be engineered to be about a factor of 2 faster than RAPTOR. 
The Connection Scan algorithm is faster than all of the competitors.

\paragraph{Section Conclusions.}

CSA enables answering earliest arrival time queries in mere milliseconds.
A corresponding earliest arrival journey can be extracted afterwards in a nearly negligible amount of addition query running time.
Even on the large Germany network with integrated local transit average query running times below 50ms are possible.
The data structures can be constructed in less than 10 seconds even for the large Germany instance.
This enables an easy straightforward and fast integration of realtime train delays.

\section{Profile Connection Scan}

The Connection Scan algorithm can be extended to solve the profile problem variants.
The algorithm is very flexible and, compared with many other algorithms to solve the profile problem, comparatively easy.

We first present the algorithm on a very high level in the form of an abstract framework.
Afterwards, we illustrate how this framework can be used to solve the various profile problem variants.
We start with a very restricted problem setting to simplify the exposition.
We then extend the algorithm, iteratively dropping these restrictions.
The initial simplifications are:
\begin{itemize}
\item The time horizon is unbounded, i.e., there is no minimum departure nor maximum arrival time in the input. 
\item We solve the all-to-one problem, i.e., there the input contains only a target stop and the profile functions from every stop to this target should be computed. 
\item We assume that there are no interstop footpaths, i.e., there are only change times, i.e., there are only loops in the footpath graph.
\item We solve the earliest arrival profile problem, i.e., we do not optimize the number of transfers.
\end{itemize}

\begin{figure}
\begin{algorithm2e}[H]
\lFor{all stops $x$}{Initialize stop data structure $S[x]$}
\lFor{all trips $x$}{Initialize trip data structure $T[x]$}
\BlankLine
\For{connections $c$ decreasing by $c_\deptime$}{
	\tcc{1. Determine arrival time when starting in $c$}
	$\tau_1\gets$~arrival time when walking to the target\;
	$\tau_2\gets$~arrival time when remaining seated, uses $T[c_\trip]$\;
	$\tau_3\gets$~arrival time when transferring, uses $S[c_{\mathrm{arr\_stop}}]$\;
	\BlankLine
	\tcc{$\tau_c$ = arrival time when starting in $c$}
	$\tau_c\gets \min \{\tau_1,\tau_2,\tau_3\}$\;	
	\BlankLine
	\tcc{2. Incorporate $\tau_c$ into the data structures}
	Incorporate $\tau_c$ into $S[x]$ for all stops $x$ with footpath $(x,c_\depstop)$\;
	Incorporate $\tau_c$ into $T[c_\trip]$\;
} 
\end{algorithm2e}
\caption{Pseudo-Code of the Connection Scan profile framework.}
\label{alg:profile-connection-scan-framework} 
\end{figure}

\subsection{Framework}

Figure~\ref{alg:profile-connection-scan-framework} depicts the high level framework of all Connection Scan based profile algorithms.
Understanding this structure is crucial to understand any of the algorithms.
At its core the algorithm uses dynamic programming.
It constructs journeys from late to early and exploits that an early journey can only have later journeys as subjourneys.
Further, it exploits the observation that a traveler sitting in a connection only has three options to continue his journey.
The three options to continue his journey are:
\begin{itemize}
\item The traveler can exit the train and, if there is a footpath to the target, walk there, or
\item he can remain seated reaching the next connection in the trip, if there is a next connection, or
\item he can exit the train and use a footpath towards some other stop and enter another train.
\end{itemize}
The two ways how a traveler can have reached a connection $c$ are:
\begin{itemize}
\item He can have been sitting in the train, i.e, he reached a connection before $c$ in the same trip, or
\item or he entered the train at $c_\depstop$ proceeded by a footpath.
\end{itemize}
The algorithm scans the connections decreasing by departure time.
We say that the algorithm iteratively \emph{scans} the connections.
In the following, we always use the letter $c$ to indicating the connection currently being scanned.
The algorithm stores at each stop $x$ a profile from $x$ to the target $t$ and at each trip the earliest arrival time over all partial journeys departing in a connection of the trip.
The algorithm's structure is depicted in the pseudo-code of Figure~\ref{alg:profile-connection-scan-framework} which mirrors this high level description very closely.

\subsection{Earliest Arrival Profile Algorithm without Interstop Footpaths}

Figure~\ref{alg:profile-connection-scan-framework} contains the pseudo-code of the basic Connection Scan profile framework.
In this section we describe, how to instantiate this framework to obtain an algorithm to solve the earliest arrival profile algorithm.
The pseudo-code of the instantiated algorithm is depicted in Figure~\ref{alg:earliest-arrival-profile-connection-scan}.

We start our description by describing how the stop data structure $S$ and the trip data structure $T$ are implemented.
Afterwards, we describe the operations that modify $S$ and $T$.

For every trip, our algorithm stores one integer, i.e., $T$ is an array of integers whose size is the number of trips.
This number represents the earliest arrival time for the partial journey departing in the earliest scanned connection of the corresponding trip.

For every stop, we store a profile function.
A function is stored as sorted array of pairs of departure and arrival times.
This means that $S$ is an array whose size is the number of stops.
The elements of $S$ are arrays with a dynamic size.
The elements of these inner arrays are pairs of departure and arrival times.
After the execution of the algorithm, $S[x]$ contains the $xt$-profile.

We initialize all elements of $T$ with $\infty$ and all elements of $S$ with a singleton array containing a $(\infty,\infty)$-pair.
This algorithm state encodes that all travel times are $\infty$, i.e., the traveler cannot get anywhere.
This would also be the correct solution, if the timetable contained no connections.
When scanning the connection $c$, we modify $S$ and $T$ to account for all journeys that use $c$.
One can thus view our approach as maintaining profiles corresponding to the timetable consisting of only the latest connections.
We start with no connection and iteratively add connections.
The scanned connection $c$ is the connection currently being added.

Scanning $c$ consists of two parts.
First, $\tau_c$ must be computed and then $\tau_c$ must be integrated into $T$ and $S$.
Computing $\tau_c$ consists of the already mentioned three subcases and the integration has two subcases. 
Luckily, most of these cases are trivial in the simplified problem variant considered here.

Computing the arrival time at the target $\tau_1$ is trivial: 
Either $c$ arrives at the target stop, in which case the arrival time is $c_\arrtime+c_\arrstop^\change$, or the target is unreachable, as there are no interstop footpaths.
If the traveler remains seated, then his arrival time will be the same, as the arrival time if he was sitting in the next connection of the trip.
This arrival time is stored in $T[c_\trip]$.
Incorporating $\tau_c$ into the trip data structure is also trivial, it consists of a single assignment: $T[c_\trip]\leftarrow \tau_c$.

Slightly more complex are the incorporation of $\tau_c$ into the profile and the efficient computation of $\tau_3$.
Incorporating $\tau_c$ consists of adding the pair $p=(c_\depstop-c_\depstop^\change,\tau_c)$ into the array if it is non-dominated.
Because the connections are scanned decreasing by departure time, there cannot be pair with an earlier departure time.
However, there can be a pair with the same departure time.
It is therefore sufficient for the domination test to look at the earliest pair $q$ already in the array.
If $p$ is not dominated by $q$, we either add $p$ or replace $q$, depending on whether the departure times are equal.

Evaluating the function is done by finding the pair $p$ in the array $S[c_\arrstop]$ with the earliest departure time no earlier than $c_\arrtime$.
The arrival time of $p$ is $\tau_3$.
As the array is sorted, the evaluation can be done in logarithmic running time using a binary search.
However, as $c_\arrtime-c_\deptime$ is usually small in practice, the requested pair is usually near the beginning of the array.
A sequential search is therefore faster in practice.

\begin{figure}
\begin{algorithm2e}[H]
\lFor{all stops $x$}{$S[x]\gets\{ (\infty, \infty) \}$}
\lFor{all trips $x$}{$T[x]\gets\infty$}
\BlankLine
\For{connections $c$ decreasing by $c_\deptime$}{
	\eIf{$c_\arrstop = \mathrm{target}$}{
		$\tau_1\gets c_\arrtime+c_\arrstop^\change$
	}{
		$\tau_1\gets \infty$
	}
	\BlankLine
	$\tau_2\gets T[c_\trip]$\;
	\BlankLine
	$p\gets$ earliest pair of $S[c_\arrstop]$\;
	\While{$p_\deptime<c_\arrtime$}{
		$p\gets$ next earlier pair of $S[c_\arrstop]$\;
	}
	$\tau_3\gets p_\arrtime$\;
	\BlankLine
	$\tau_c\gets \min \{\tau_1,\tau_2,\tau_3\}$\;	
	\BlankLine
	$p\gets(c_\deptime-c_\depstop^\change,\tau_c)$\;
	$q\gets$ earliest pair of $S[c_\depstop]$\;
	\If{$q$ does not dominate $p$}{
		\eIf{$q_\deptime \neq p_\deptime$}{
			Insert $p$ at the front of $S[c_\depstop]$\;
		}{
			Replace $q$ as the earliest pair of $S[c_\depstop]$ with $p$\;
		}
	}
	\BlankLine
	$T[c_\trip]\gets \tau_c$\;
} 
\end{algorithm2e}
\caption{Earliest arrival Connection Scan profile algorithm without interstop footpaths.}
\label{alg:earliest-arrival-profile-connection-scan} 
\end{figure}

\subsection{Optimizations}

Several optimizations exist for the Connection Scan profile algorithm. 
The first optimization, that we describe exploits a hardware feature called prefetching.
The next three optimizations exploit that in most cases we do not want to compute journeys from every stop to the target.
They exploit additional information in the input such as the source stop to accelerate the computation.

\paragraph{Memory Prefetching.}
The Connection Scan profile algorithm can be slightly accelerated by using processor memory prefetch instructions.
Modern processors are capable of detecting simple memory access patterns and to fetch data sufficiently early to hide memory access latency.
The sequential scan over the connection array is an example of such a simple memory access pattern.
However, detecting the stop profile access is more complex.
When scanning the $c$-th connection, we therefore execute prefetch instructions for the stop profiles $S[{(c-4)}_\depstop]$, and $S[{(c-4)}_\arrstop]$ and the trip arrival time $S[{(c-4)}_\depstop]$.
These instructions help hide memory latency by overlapping the processing of connection $c$ with the memory fetching of the four connections $c-4$, $c-3$, $c-2$, and $c-1$.

\paragraph{Bounded Time-Horizon.}
The minimum departure time $\tau_s$ and maximum arrival time $\tau_t$ can exploited by only scanning connections $c$ with $\tau_s \le c_\deptime \le \tau_t$.
The earliest connection can be determined using a binary search.
To determine the latest connection, a binary search can be used.
However, it is also a byproduct of the next optimization.

\paragraph{Scanning only Reachable Trips.}
The source stop and source times can be exploited by running a non-profile earliest arrival scan before the profile scan.
The objective of this initial scan is to determine, which trips are reachable.
If a trip is not reachable, then no connection in it can be reachable.
We do not have to scan non-reachable connections as they cannot influence the profile at the source stop.
We can thus skip connections for which the trip bit is not set.
An efficient implementation starts by finding the first connection not before $\tau_s$ using a binary search.
It then performs the earliest arrival scan increasing by departure time until a connection departing after $\tau_t$ is encountered.
The same connections are then scanned in the reverse order in the profile scan.

\paragraph{Source Domination.}
The source stop can be exploited in another way.
In the profile framework depicted in Figure~\ref{alg:profile-connection-scan-framework}, scanning a connection consists of two parts.
The first part determines the arrival time when sitting in the connection $\tau_c$.
The second part incorporates $\tau_c$ into the data structures.
Consider the pair $p=(c_\deptime,\tau_c)$. 
If $p$ is dominated by the pairs in the profile of the source stop, then the second part can be skipped.
This optimization is correct because every journey starting at the source stop and using $c$ would be dominated.

It remains to describe, how to efficiently implement the domination test.
For the test, we need to know the arrival time of the earliest pair $q$ in the profile of the source stop such that $q_\deptime \ge c_\deptime$.
This information can be obtained by evaluating the source stop's profile.
However, as the connections are scanned decreasing by departure time, we can do better by maintaining a pointer to the relevant pair in the source stop's profile.
When scanning a connection our algorithm first decreases the pointer if necessary and then looks up the arrival time.
As the pointer can only be decreased as often as there are pairs in the source stop's profile, we can bound the running time needed to perform these evaluations by the size of the source stop's profile.

\subsection{Interstop Footpaths}
\label{sec:interstop-footpath}

In this section, we expand the profile algorithm to handle interstop footpaths.
Initial and transfer footpaths are handled in the same way, but a different strategy is needed for final footpaths.
We start our description with the later, as the idea is simpler.
The pseudo-code for this algorithm variant is presented in Figure~\ref{alg:profile-connection-scan-with-footpaths}.

\paragraph{Final Footpaths.}

Handling final footpaths consists of modifying the computation of $\tau_1$ in the framework of Figure~\ref{alg:profile-connection-scan-framework}.
In the base algorithm, the traveler can only arrive at the target by train.
In the extended version, he can also walk at the end.
For this extension, we add a new array of integers $S$.
It stores for every stop the walking distance to the target or $\infty$, if walking is not possible.
Computing $\tau_1$ for a connection $c$ can be done in constant time by evaluation $c_\arrtime + D[c_\arrstop]$.

For efficiently reasons, we do not reset all elements of $D$ for each query.
Instead, we initialize all elements of $D$ to $\infty$ during the algorithm setup.
We do this initialization only once.
Each query begins by iterating over the incoming footpaths of the target stop.
It sets $D$ to the appropriate values for all stops from which the traveler can transfer to the target.
After the profile computation, our algorithms iterates a second time over the same footpaths to reset all values of $D$ to $\infty$.

\paragraph{Transfer and Initial Footpaths.}

Our algorithm handles transfer and initial footpaths by iterating over the incoming footpaths $f$ of $c_\depstop$ when incorporating $\tau_c$ into the profiles.
It inserts a pair $p=(c_\deptime-f_\dur,\tau_c)$ into the profile of the stop $f_\depstop$, if $p$ is not dominated in $f_\depstop$'s profile. 

Unfortunately, we can no longer guarantee, that the departure time of $p$ will be the earliest in each profile.
A slightly more complex insertion algorithm is therefore needed:
Our algorithm temporarily removes pairs departing before the new pair.
It then inserts $p$, if non-dominated, and then reinserts all previously removed pairs that are not dominated by $p$.

\paragraph{Limited Walking}

If the number of interstop footpaths is large, handling transfer and initial footpaths can be computationally expensive.
Especially, the iteration over the incoming footpaths of $c_\depstop$ can be costly.
Fortunately, the limited walking optimization can be adapted and can drastically reduce running time on some instances.
The idea is as follows: If the pair $(c_\deptime,\tau_c)$ is dominated in the profile of $c_\depstop$, then all pairs computed when scanning $c$ are dominated.
The correctness argument is essentially the same as for the non-profile algorithm.
One can prefix the journey of the dominating pair with each footpath and obtain at each stop a pair that would dominate each of the pairs created during the scanning of $c$.
We thus do not need to generate them as they would be dominated anyway, i.e., we do not need to iterate over the incoming footpaths.

\paragraph{Different Set of Footpaths for Initial and Final Footpaths.}

\begin{figure}
\begin{algorithm2e}[H]
\tcc{$D[x]\gets \infty$ for every stop $x$ in a preprocessing step}
\lFor{all footpaths $f$ with $f_\arrstop = \mathrm{target}$}{$D[x]\gets f_\dur$}
\lFor{all stops $x$}{$S[x]\gets\{ (\infty, \infty) \}$}
\lFor{all trips $x$}{$T[x]\gets\infty$}
\BlankLine
\For{connections $c$ decreasing by $c_\deptime$}{
	$\tau_1\gets c_\arrtime + D[c_\arrstop]$\; 
	$\tau_2\gets T[c_\trip]$\;
	$\tau_3\gets$~evaluate~$S[c_\arrstop]$ at $c_\arrtime$\;
	\BlankLine
	$\tau_c\gets \min \{\tau_1,\tau_2,\tau_3\}$\;	
	\BlankLine
	\If{$(c_\deptime,\tau_c)$ is non-dominated in profile of $S[c_\arrstop]$}{
		\For{all footpaths $f$ with $f_\arrstop = c_\depstop$}{
			Incorporate $(c_\deptime-f_\dur,\tau_c)$ into profile of $S[f_\depstop]$\;
		}
	}
	$T[c_\trip]\gets \tau_c$\;
} 
\lFor{all footpaths $f$ with $f_\arrstop = \mathrm{target}$}{$D[x]\gets \infty$}
\end{algorithm2e}
\caption{Earliest arrival Connection Scan profile algorithm with interstop footpaths and the limited walking optimization.}
\label{alg:profile-connection-scan-with-footpaths} 
\end{figure}

In our proposed transfer model, we only have one type of footpaths.
However, many applications have an extended set of footpaths for the initial and final footpaths.
In some applications the traveler can walk for a longer amount of time at the beginning or at the end of his journey than when changing trains.
Further, some applications have source and target locations that are not stops but might, for example, be city districts.
Luckily, our algorithm can easily be extended to handle these cases.

Final footpaths can be handled by iterating over the extended footpath set during the initialization of $D$.
Handling initial footpaths is slightly more complex.
Denote by $s$ the source location, for which the profile should be computed.
In a first step, we create a set $X$ of pairs that may contain dominated entries.
After removing the dominated entries, the profile of $s$ is obtained.

Our algorithm starts by iterating over all outgoing extended footpaths $f$ of $s$.
For every pair $(d,a)$ in the profile of $f_\arrstop$, there is a $(d-f_\dur,a)$ pair in $X$.
After removing dominated pairs from $X$, the profile of $s$ is obtained.

It is possible to generate the set of extended footpaths using Dijkstra's algorithm on the fly.
We can therefore drop the requirement that the set of extended footpaths must be transitively closed.
This allows us to have very long initial and final footpaths.
Unfortunately, the restrictions still apply for transfer footpaths.

\subsection{Optimizing the Number of Legs}

In the previous section, we presented the basic Connection Scan profile algorithm and extended it to a footpath-based transfer model.
In this section, we further extend it to optimize the number of legs beside the arrival time.
We present three ways to perform this optimization.
The first and easiest approach optimizes the number of legs as a secondary criterion.
The second approach is a refinement of the first that heuristically mitigates some of its problems.
Finally, we present as third approach an extension that optimizes the number of legs and the arrival time in the Pareto-sense.

The overhead of the first two approaches over the basic algorithm is negligible.
Unfortunately, the optimization in the Pareto-sense adds a significant overhead.
We therefore recommend to the reader to first try the first two approaches and only use the third if it is really necessary for the particular application at hand.

Our algorithm optimizes the number of legs by counting the number of times a traveler exits a train.
As there is an exit per leg, the number of exits and the number of legs coincide.
The exit counter is increased each time that a profile is evaluated, i.e., during the computation of $\tau_3$ in the framework.

\paragraph{Number of Legs as Secondary Criterion.}

Optimizing the number of legs as secondary criterion, i.e., computing a journey with a minimum number of legs among all journeys with a minimum arrival time, is surprisingly easy.
Denote by $\epsilon$ a negligibly small time value, i.e., think of $\epsilon$ as one millisecond.
The modification of our algorithm consists of increasing $\tau_3$ by $\epsilon$ after each profile evaluation, i.e., the modification consists of inserting a single addition compared to the base algorithm.
If two journeys have different arrival times, then the earlier journey is chosen.
If the arrival times are equal, the number of $\epsilon$s added determines which journey is chosen.
As an $\epsilon$ is added each time that the travelers exits a train, the number of $\epsilon$s corresponds to the number of legs.
The number of legs is thus optimized as secondary criterion.

In a real implementation, we multiply all departure and arrival times in the timetable with a small constant, such as for example $2^5$.
Timestamps, even with second resolution, usually require significantly fewer than 32 bits. 
For example, to encode all seconds within a year, 25 bits are enough.
We can therefore encode the modified timestamps using 32 bit integers.
The value of $\epsilon$ is set to 1.
The modifications to the algorithm depicted in Figure~\ref{alg:profile-connection-scan-with-footpaths} adding a ``+1'' in line 7 and perform the scaling using two bit shift operations between the lines 4 and 5.

Stated differently, we encode the number of legs in the lower 5 bits of a timestamp. 
The higher 27 bits encode the arrival time.
As an integer comparison only compares the lower bits if the higher bits are equal, we obtain the desired effect, that the journeys are tie-broken using the number of legs.

\paragraph{Rounding the Arrival Times.}

Optimizing the number of legs as secondary criterion, eliminates the most problematic earliest arrival journeys, such as those visiting a stop several times or those entering a trip multiple times.
However, a journey that arrives at $8{:}02$ with 10 legs is still preferred over a journey with 2 legs arriving at $8{:}03$.
While the former arrives earlier, most travelers prefer the later.
This problem can be avoid by optimizing the number of legs in the Pareto-sense.
Fortunately, a simpler partial solution to the problem exists that might be good enough for some applications.

The idea consists of rounding the value of $\tau_1$ in the framework of Figure~\ref{alg:profile-connection-scan-framework}.
If $\tau_1$ is rounded down the lowest multiple of say 5 minutes, then both journeys are equal with respect to arrival time and therefore the journey with 2 legs is chosen.
Rounding down to multiple of 5 minutes divides a day into 288 time buckets.
Journeys arriving within one bucket are regarded as arriving at the same time and thus one with a minimum number of legs is picked.
This avoids many problematic journeys, but it is only a partial solution as the problem remains at the time bucket borders.
Further, the trick has no effect, if the difference in journey arrival times is larger than the bucket size.

Notice, that we are only rounding the arrival times at the target stop.
We do not round the departure or arrival times of intermediate connections.
This trick therefore does not modify the transfer model.

A problem with this trick is that the profiles contain rounded arrival times.
However, we want to display the non-rounded arrival times to the user.
Further, there will only be one journey per bucket.
Fortunately, these problems can be solved by permuting some bits in the timestamps.

\begin{figure}
\begin{center}
\includegraphics{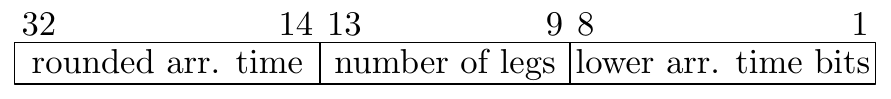}
\end{center}
\caption{Encoding used to represent timestamps. The numbers represent the bit-offsets within a 32 bit integer of the three data items.}
\label{fig:bit-encoding}
\end{figure}

Suppose that, we want to use 5 bits to encode the number of legs.
Further, assume that we round the arrival times down to $2^8=256$.
With seconds resolution that corresponds to rounding down to multiples of $\approx$4.2 minutes.
The idea consists of not encoding the number of legs in the lowest bits of a timestamp.
Instead, we use bits in the middle.
The lowest 8 bits are the lower bits of the arrival time.
The next higher 5 bits are the number of legs.
The remaining bits encode the higher bits of the arrival time.
Figure~\ref{fig:bit-encoding} illustrates the layout.
The effect of this modification is that our algorithm now optimizes three criteria.
These are:
\begin{enumerate}
\item The rounded arrival time,
\item the number of legs, and
\item the exact arrival time.
\end{enumerate}
Criteria 2 and 3 are used as second and third criteria, i.e., they are tie-breakers.
The exact arrival times can easily be reconstructed from this encoding.
Further, assume that there are two journeys that arrive within the same bucket and have the same number of legs but have different arrival times.
In the base version only one would be found.
Using the refined algorithm both are found as they are not identical with respect to the third criterion.

Unfortunately, as already mentioned this trick mitigates but does not resolve the problem of trading many addition transfers for a tiny improvement in arrival time.
However, for certain applications this trick reduces the number of problematic cases to a sufficiently small amount.
The main advantage of this trick is that it is significantly easier to implement than the more complex solution described in the next paragraph.
Further, the incurred overhead is comparatively low.

\paragraph{Pareto Optimization.}

\begin{figure}
\begin{algorithm2e}[H]
\lFor{all stops $x$}{$S[x]\gets\{ (\infty, (\infty,\infty\ldots \infty)) \}$}
\lFor{all trips $x$}{$T[x]\gets(\infty,\infty\ldots \infty)$}
\BlankLine
\For{connections $c$ decreasing by $c_\deptime$}{
	\eIf{$c_\arrstop = \mathrm{target}$}{
		$x\gets c_\arrtime+\mathrm{target}_\change$\;
	}{
		$x\gets \infty$\;
	}
	$\tau_1\gets (x,x \ldots x)$\;
	$\tau_2\gets T[c_\trip]$\;
	$\tau_3\gets \mathrm{shift}($evaluate~$S[c_\arrstop]$ at $c_\arrtime )$\;
	\BlankLine
	$\tau_c\gets \min (\tau_1,\tau_2,\tau_3)$\;	
	\BlankLine

	$y\gets$ arrival time of earliest pair of $S[c_\depstop]$\;
	\If{
		$y \neq \min (y, \tau_c)$
	}{
		Add $(c_\deptime-(c_\depstop)_\change,\min(y, \tau_c))$ at the front of $S[c_\depstop]$\;
	}
	$T[c_\trip]\gets \tau_c$\;
} 
\end{algorithm2e}
\caption{Pareto Connection Scan profile algorithm without interstop footpaths.}
\label{alg:pareto-profile-connection-scan} 
\end{figure}

The number of legs and the arrival times can be optimized in the Pareto-sense.
For a fixed target $t$, we want to compute for every source stop $s$, every source time $\tau_s$, and every number of legs $\ell$, the earliest arrival time $\tau_t$ over all journeys from $s$ to $t$ not departing before $\tau_s$ with at most $\ell$ legs.
To simplify this problem slightly, we bound $\ell$ by $\leg_{\max}$ which is a constant in the algorithm.
We usually set $\leg_{\max}$ to 8 or a similarly large value, exploiting that travelers in practice do not care about journeys with too many legs.

We modify our algorithm by replacing all arrival times by constant-sized vectors.
$\leg_{\max}$ is the dimension of the vectors.
We denote the elements of a vector $A$ as $A[1],A[2]\ldots A[\leg_{\max}]$.
The element $A[\ell]$ is the arrival time at the target, if the journey has at most $\ell$ legs.
We define two operations that modify these vectors.
The first is the \emph{component wise minimum}, i.e., the result of the minimum operation of two vector $A$ and $B$ is a vector $C$ such that $C[i]=\min\{A[i],B[i]\}$ for all indices $i$.
The second operation is the \emph{shift} operation, which is defined as follows: Shifting $A$ yields a vector $B$ such that $B[1]=\infty$ and $B[i]=A[i-1]$ for all other indices $i$.

The interpretation of the minimum operation consists of taking the best of two options.
Further, the shift operation can be interpreted as increasing the number of legs.

All $\tau$-variables in the framework from Figure~\ref{alg:profile-connection-scan-framework} become vectors. 
The trip data structure $T$ becomes an array of vectors.
The profile data structure $S$ becomes an array of dynamic-sized arrays of pairs of an integer and a vector.
The walking distance to the target $D$ remains an array of integers.

It is possible that a vector $A$ dominates another vector $B$ in one component, for example $A[1]<B[1]$, but $B$ dominates $A$ in another component, for example $A[2]>B[2]$.
For this reason, the vector insertion must be modified.
If all components of the new vector are dominated, then the profile is not modified.
Otherwise, we insert the minimum of the new vector and the minimum of the earliest vector already in the profile.
Two successive pairs can have the same arrival time with respect to certain but not all values of $\ell$ but different departure times.

In the base algorithm the profiles are initialized with a sentinel $(\infty,\infty)$ pair.
The arrival time of this pair is a vector in the extended algorithm, i.e., the new sentinel is $(\infty, (\infty,\infty\ldots \infty))$.

The computation of $\tau_1$ starts analogous to the non-Pareto case. 
Our algorithm starts by computing the walking time $x$ to the target.
Afterwards $x$ is converted to a vector $A$ by setting $A[i]=x$ for all indices $i$.
The operation of setting all components of a vector to one value is sometimes called \emph{broadcast}.

In Figure~\ref{alg:pareto-profile-connection-scan} we present the profile Pareto algorithm in pseudo-code form. 
To simplify its exposition, we omit interstop footpaths.
Fortunately, they can be incorporated in the same way as already described in Section~\ref{sec:interstop-footpath} and depicted in Figure~\ref{alg:profile-connection-scan-with-footpaths}. 

\begin{figure}
\begin{center}
\includegraphics{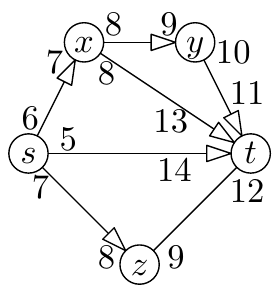}
\end{center}
\caption{
Example timetable.
The circles are stops and the arrows are connections annotated by their departure and arrival times.
All connections are part of different trips.
There are 4 journeys from $s$ to $t$ with a varying number of legs: 
$s@5{\rightarrow}t@14$, 
$s@7{\rightarrow}z{\rightarrow}t@12$,
$s@6{\rightarrow}x{\rightarrow}t@13$, and
$s@6{\rightarrow}x{\rightarrow}y{\rightarrow}t@11$
}
\label{fig:pareto-example}
\end{figure}

\paragraph{Example.}
Consider the example timetable depicted in Figure~\ref{fig:pareto-example}. 
We describe how the profile of $s$ evolves during the execution of our algorithm.
We set the target stop to $t$ and $\leg_{\max}$ to 3.
The profile is a dynamic array of pairs of departure time and arrival time vectors.
Initially it only contains an infinity sentinel, i.e., initially we have $S[s]=\{(\infty,(\infty,\infty,\infty))\}$.

The profile $S[s]$ is changed for the first time, when the connection from $s$ to $z$ is scanned.
The value of $\tau_c$ is $(\infty,12,12)$.
As there is no way to reach $t$ with at most 1 leg, the first component $\tau_c[1]$ is $\infty$.
$\tau_c[2]$ is $12$ as the target can be reached at 12 with 2 legs.
Further, $\tau_c[3]$ is $12$ also as the target can be reached at 12 with at most 3 legs.
Notice that $\tau_c[3]$ is $12$ even though that the corresponding journey only contains 2 legs.
$\tau_c$ is better in two components than the earliest vector in the profile, which is the $(\infty,\infty,\infty)$ sentinel.
The algorithm therefore inserts a new pair, namely $(7,(\infty,12,12))$ into the profile $S[s]$.
The profile $S[s]$ after the scan is  $\{(7,(\infty,12,12)), (\infty,(\infty,\infty,\infty))\}$.

The profile $S[s]$ is changed for the second time, when the connection from $s$ to $x$ is scanned.
The value of $\tau_c$ is $(\infty,13,11)$.
$\tau_c[1]$ is $\infty$ as $t$ cannot be reached without transfer.
$\tau_c[2]$ is $13$ because the journey $s@6{\rightarrow}x{\rightarrow}t@13$ contains 2 journeys.
Further, $\tau_c[3]$ is $10$ because the journey $s@6{\rightarrow}x{\rightarrow}y{\rightarrow}t@11$ with 3 legs exists.
As the later has more than 2 legs, we have that $\tau_c[2]\neq 11$.
$\tau_c$ is better in at least one component than the earliest vector in the profile, i.e., $(\infty,12,12)$.
However, it is not better in every component.
The algorithm therefore computes the minimum $\min\{(\infty,13,11),(\infty,12,12)\} = (\infty,12,11)$.
The pair $(6,(\infty,12,11))$ is added to the profile $S[s]$.
The resulting profile has the value $\{(6,(\infty,12,10)), (7,(\infty,12,12)), (\infty,(\infty,\infty,\infty))\}$.

The last time that the profile $S[s]$ might be change is when the connection from $s$ to $t$ is scanned.
The value of $\tau_c$ is $(14,14,14)$.
However, $\tau_c$ is not better in any component than the earliest vector in the profile, i.e., $(\infty,12,10)$.
No pair is thus added.

After the execution of the algorithm the profile $S[s]$ is $\{(6,(\infty,12,10)), (7,(\infty,12,12)), (\infty,(\infty,\infty,\infty))\}$.
To determine the arrival time for a source time $\tau_s$ and maximum number of legs $\ell$, find the earliest pair with a departure time no earlier than $\tau_s$.
The $\ell$-th component of the corresponding arrival time vector contains the answer.

For $\tau_s = 6.5$ and $\ell = 3$, we therefore first look up the first pair with a departure time after $6$. 
This is $(7,(\infty,12,12))$.
The $\ell$-th, i.e., third, component is 12.
The traveler can thus arrive at 12.

\paragraph{Earliest Arrival Time.}
In some cases, one is more interested in the minimum arrival time over all journeys than in the minimum arrival time over all journeys with at most $\leg_{\max}$ legs.
This can be implemented using a small change in the definition of the shift operation. 
The result of the modified shift of a vector $A$ is a vector $B$ such that $B[1]=\infty$, $B[\leg_{\max}] = \min\{A[\leg_{\max}-1], A[\leg_{\max}]\}$, and $B[i] = A[i-1]$ for all other indices $i$. 
With this modification, the $\leg_{\max}$-th vector component contains the earliest arrival times over all journeys.

\paragraph{SIMD.}
All vectors operations, i.e., component-wise minimum, component shifting, and broadcasting a value to all components, can be implemented using SIMD operations on all common processor architectures.
This includes x86 processors with the SSE and AVX2 instruction sets.
One SSE vector has 4 components with 32 bit integers.
Concatenating two vectors, yields an efficient implementation for $\leg_{\max}=8$.
Alternatively, AVX2 vectors have 8 components with 32 bit integers.
One AVX2 vector is therefore large enough.

\subsection{Journey Extraction}

\label{sec:extraction}
In the previous section, we introduced an algorithm to compute profiles. 
In this section, we describe how to extract corresponding journeys in a post processing step.

Similar to the extraction process for the earliest arrival time Connection Scan algorithm, the extraction comes in two variants.
The first conceptually simpler approach consists of storing journey pointers.
The second approach computes the journey pointers on the fly during the extraction.

The input consists of a source stop $s$ and source time $\tau_s$.
The output consists of an earliest arrival journey towards the target stop for which the profile was computed.
If transfers are optimized in the Pareto-sense, then the input contains additionally a maximum number of legs $\ell$.

Several journeys can exist that are identical with respect to all considered criteria, i.e., they depart at the same source stop at the same source time and arrive at the same target stop at the same target time and have the same number of transfers.
We only consider the problem setting of extracting one of these journeys.
Our algorithms guarantee that the extracted journey visits no stop or trip twice even when the number of legs is not optimized.

\subsubsection{Journey Pointers}

In the base profile algorithm, the pairs $(d,a)$ contain two pieces of information namely a departure time $d$ and an arrival time $a$. 
We extend the pairs with two connection IDs $l^\enter$, $l^\exit$, turning the pairs into quadruples $(d,a,l^\enter,l^\exit)$.
The meaning of such a quadruple is that there is an optimal journey $j$ that arrives at the target stop at time $a$ and departs at time $d$.
The extracted journey $j$ starts with a footpath towards $l^\enter_\depstop$.
$j$ leaves the stop using the connection $l^\enter$.
The traveler exits the train at the end of the connection $l^\exit$.
These quadruples can be used to iteratively extract an optimal journey.

The extraction starts by computing the time needed to directly transfer to the target.
Doing this trivial without interstop footpaths. 
With footpaths, we use the $D$ array of the base profile algorithm. 
In the next step, our algorithm determines the first quadruple $p$ after $\tau_s$ in the profile $P[s]$ of the source stop $s$.
If directly transferring to the target is faster, then the journey consists of a single footpath and there is nothing left to do.
Otherwise, $p$ contains the first leg of an optimal journey.
The algorithm then sets $s$ to $l^\exit_\arrstop$ and $\tau_s$ to $l^\exit_\arrtime$ and iteratively continues to find the remaining legs of the output journey.

It remains to describe how $l^\enter$ and $l^\exit$ are determined when inserting the quadruple into the profile during the scan.
$l^\enter$ is the connection being scanned and is therefore already known.
To determine $l^\exit$ efficiently, we extend the trip information $T$ with a connection ID for each trip, i.e., $T$ becomes an array of pairs of arrival times and connection IDs.
Each time that the arrival time stored in $T$ is decreased, the algorithm sets the trip's connection ID to the currently scanned connection.
When inserting the quadruple, $c^\exit$ is the connection ID stored with currently scanned connection's trip.

This approach can be combined with Pareto-optimization by replacing $l^\enter$, $l^\exit$, and the trip connection IDs with constant-sized vectors.
The input of the algorithm must be extended with the maximum number of desired legs.

\subsubsection{Without Journey Pointers}

Similarly to the earliest arrival Connection Scan, it is possible to implement a journey extraction without modifying the scan.

Our algorithms require enumerating the outgoing connections of a stop ordered by departure time.
To efficiently support this operation, we create an auxiliary data structure that consists of an adjacency array that maps a stop $s$ onto the departure time and the ID of all connections $c$ departing at $s$, i.e., onto the connections $c$ for which $c_\depstop=s$ holds.
The outgoing connections are ordered by departure time. 
Further, our algorithm needs to be able to enumerate all connections in a trip after a given connection.
To efficiently support the second operation, we create another auxiliary adjacency array that maps a trip $t$ onto the IDs of the connections $c$ in the trip, i.e., onto the connections $c$ for which $c_\trip = t$ holds.
The connections are ordered by their position in the trip.
To enumerate the connections in a trip after a given connection $c$, we enumerate the connections in $c_\trip$ from late to early and abort the enumeration once $c$ is encountered.
Notice, that all auxiliary data structures are independent of the target stop.
Further, both data structures can be computed by essentially sorting the connections by various criteria.
We can therefore compute the auxiliary data in a fast preprocessing step.

Similarly, to the journey pointer approach, our second approach starts by checking, whether directly walking from the source stop $s$ and the source time $\tau_s$ to the target $t$ is optimal.
It terminates, if this is the case.
Otherwise, our algorithm must compute a pair of valid $l^\enter$ and $l^\exit$.
In the first approach, these were stored in the pairs which is no longer the case in the second approach.
Our algorithm therefore needs to infer the values.
It does so by searching for the earliest pair $(d,a)$ after $\tau_s$ in $s$'s profile $P[s]$ using a binary search.
We know that there must be a footpath $f$ outgoing from $s$ towards $l^\enter_\depstop$ such that $l^\enter_\deptime = d+f_\dur$.
By iterating over the outgoing footpaths of $s$ and checking this condition, we obtain a set $\{c^1,c^2\ldots c^k\}$ of candidates for $l^\enter$.
We know that there must be an optimal first leg $l$, such that $l^\enter$ is among the candidates.

We can optionally prune the candidate set using the trip arrival times $T[x]$ computed during the profile scan.
$T[x]$ is the minimum arrival time over all optimal journeys departing in a connection of trip $x$.
We therefore know that if for a candidate $T[c^i_\trip]>a$ holds, that $c^i$ cannot be $l^\enter$ and we can therefore remove $c^i$ from the set.

For the remaining candidates, we need to look at the connections in their trips.
For each potential candidate $c^i$, our algorithm enumerates all connections $c$ in its trip that come after $c^i$, including $c^i$ itself.
For each $c$, our algorithm searches for the earliest pair $(d',a')$ in $c_\arrstop$'s profile after $c_\arrtime$ using a binary search.
If $a=a'$, then we found an optimal first leg and $c$ is the corresponding $l^\exit$.
If we only wish to extract one journey, then our algorithm can discard the remaining candidates.
Our algorithm iterates by setting $s$ to $l^\exit_\arrstop$ and $\tau_s$ to $l^\exit_\arrtime$.
To assure that no trip is used twice in a journey, we pick the latest valid $l^\exit$ in the trip.
As we enumerate connections from late to early, the first valid $l^\exit$ we encounter is automatically the latest.

\paragraph{Pareto Optimization.}

The candidate set is computed by finding the first pair $(d,a)$ departing after $\tau_s$.
This is correct for the base profile scan algorithm.
However, the Pareto-extension can insert several pairs with the same departure time with respect to $\ell$.
A modification to the extraction is therefore necessary.

Consider for example the example illustrated in Figure~\ref{fig:pareto-example}.
Suppose that the traveler departs at $s$ at 5 and wants to use at most 2 legs.
Already the first pair $(6,(\infty,12,10))$ in the profile departs later than 5.
However, there is no earliest arrival journey towards $t$ departing at 6 towards $t$ with at most 2 legs.
The corresponding journey departs at 7. 
Indeed, the second pair $(7,(\infty,12,12))$ in the profile has the correct departure time and arrives at the same time.

To fix this problem, we slightly modify the algorithm.
First we find the earliest pair $p$ departing no earlier than $\tau_s$.
In a second step, we iterated over the pairs in the profile from early to late starting at $p$ until we find the last pair $q$ with the same arrive time than $p$ for the requested number of legs.
The departure time of $q$ is used to determine the candidate set.

\subsection{Experiments}

\begin{table}
\begin{center}
\begin{tabular}{ccccccr}
\toprule
\multicolumn{7}{c}{Older Machine with 20 MiB of L3 cache}\\
\midrule
&  Pre- & Limited & Source & Range & Journeys  & Running \\
Instance & fetch & Walk. & Dom. & Query & Extraction & Time [ms]\\
\midrule
Germany & $\circ$ & $\circ$ & $\circ$ & $\circ$ & $\circ$ & 2\,132.1 \\
Germany & $\bullet$ & $\circ$ & $\circ$ & $\circ$ & $\circ$ & 1\,995.7 \\
Germany & $\bullet$ & $\bullet$ & $\circ$ & $\circ$ & $\circ$ & 1\,567.2 \\
Germany & $\bullet$ & $\bullet$ & $\bullet$ & $\circ$ & $\circ$ & 1\,119.3 \\
Germany & $\bullet$ & $\bullet$ & $\bullet$ & $\bullet$ & $\circ$ & 253.1 \\
Germany & $\bullet$ & $\bullet$ & $\bullet$ & $\circ$ & $\bullet$ & 1\,118.4 \\
Germany & $\bullet$ & $\bullet$ & $\bullet$ & $\bullet$ & $\bullet$ & 253.1 \\
\midrule
London & $\circ$ & $\circ$ & $\circ$ & $\circ$ & $\circ$ & 287.8 \\
London & $\bullet$ & $\circ$ & $\circ$ & $\circ$ & $\circ$ & 279.7 \\
London & $\bullet$ & $\bullet$ & $\circ$ & $\circ$ & $\circ$ & 162.3 \\
London & $\bullet$ & $\bullet$ & $\bullet$ & $\circ$ & $\circ$ & 119.9 \\
London & $\bullet$ & $\bullet$ & $\bullet$ & $\bullet$ & $\circ$ & 11.1 \\
London & $\bullet$ & $\bullet$ & $\bullet$ & $\circ$ & $\bullet$ & 121.2 \\
London & $\bullet$ & $\bullet$ & $\bullet$ & $\bullet$ & $\bullet$ & 11.2 \\
\midrule
\multicolumn{7}{c}{Newer Machine with 10 MiB of L3 cache, used in most experiments}\\
\midrule
Germany & $\circ$ & $\circ$ & $\circ$ & $\circ$ & $\circ$ & 2\,517.2 \\
Germany & $\bullet$ & $\circ$ & $\circ$ & $\circ$ & $\circ$ & 2\,391.0 \\
Germany & $\bullet$ & $\bullet$ & $\circ$ & $\circ$ & $\circ$ & 1\,684.4 \\
Germany & $\bullet$ & $\bullet$ & $\bullet$ & $\circ$ & $\circ$ & 1\,246.2 \\
Germany & $\bullet$ & $\bullet$ & $\bullet$ & $\bullet$ & $\circ$ & 217.9 \\
Germany & $\bullet$ & $\bullet$ & $\bullet$ & $\circ$ & $\bullet$ & 1\,246.4 \\
Germany & $\bullet$ & $\bullet$ & $\bullet$ & $\bullet$ & $\bullet$ & 218.0 \\
\midrule
London & $\circ$ & $\circ$ & $\circ$ & $\circ$ & $\circ$ & 242.3 \\
London & $\bullet$ & $\circ$ & $\circ$ & $\circ$ & $\circ$ & 238.7 \\
London & $\bullet$ & $\bullet$ & $\circ$ & $\circ$ & $\circ$ & 140.0 \\
London & $\bullet$ & $\bullet$ & $\bullet$ & $\circ$ & $\circ$ & 106.9 \\
London & $\bullet$ & $\bullet$ & $\bullet$ & $\bullet$ & $\circ$ & 9.4 \\
London & $\bullet$ & $\bullet$ & $\bullet$ & $\circ$ & $\bullet$ & 107.9 \\
London & $\bullet$ & $\bullet$ & $\bullet$ & $\bullet$ & $\bullet$ & 9.4 \\
\bottomrule
\end{tabular}

\end{center}
\caption{Earliest arrival profile computation running times.}
\label{tab:ea_profile}
\end{table}

We use the experimental setup described in Section \ref{sec:experimental-setup}.
In Table~\ref{tab:ea_profile}, we report the running times of the earliest arrival Connection Scan profile algorithm.
We report the running times for both main instances on both of our test machines.
We iteratively activate optimizations to show their impact.
Activating range queries also includes not processing unreachable trips.
We also report the running time needed to perform the scan and extract for every pair in the source stop's profile a corresponding earliest arrival journey.

The comparison between the two machines is interesting.
We expect the newer machine to be faster, as it has a faster processor, a newer architecture, and faster RAM.
This expected behavior is also nearly always the observed behavior, except on the Germany instance for non-range queries.
The differences in L3 cache sizes explains the effect.
The newer machine is better with respect to every criterion except L3 cache.
The old machine has 20 MiB while the newer one only has 10 MiB.
The London instance is smaller and therefore a larger part of the stop profiles fit into the 10 MiB.
If we compute range queries, only parts of the stop profiles are computed.
This part is smaller and therefore a greater percentage fits into the L3 cache.
The newer machine is therefore faster on range queries and slower on non-range queries.
The conclusion is that a sufficiently large cache is necessary for a good Connection Scan profile performance.

Activating prefetching decreases the running times.
On the newer machine and the London instance the speedup is only about 1.02.
However, on the Germany instance the gain is already 1.05.
This observation again illustrates that caching effects matter for good performance.
On the London instance large parts of the frequently used data structures are never evicted from L3 cache. 
The gain from prefetching comes therefore mostly from moving data to the lower cache levels.
On the Germany instance prefeching moves data from the RAM into L3 cache more often.
As the absolute differences in access speeds between L3 cache and RAM are greater than between L2 and L3 cache, the speedup is lower for the London instance.

Activating the limited walking optimization further reduces the running times.
The speedup is about 1.4 to 1.7, which is roughly comparable to the speedups achieved for the non-profile algorithm variants.

Activating source domination further reduces the running times.
As source domination prunes pairs from profiles except the source stop, the algorithm solves a more restricted problem setting.
Instead of computing the profiles from every stop towards the target stop, it now only computes a single profile from the source stop to the target.

Switching to range queries drastically reduces the running times.
The specified maximum travel time of twice the minimum travel time allows the algorithm to limit the connections that need to be scanned.
On the Germany instance the speedup is about a factor 6.
On the London instance the speedup of 11 is higher.
These speedups are roughly comparable to the speedups achieved by activating the start and stop criteria in the non-profile earliest arrival algorithm.
The reason is that the decrease in scanned connections is roughly comparable.
Further, as already observed, a traveler needs less time to traverse London than to traverse Germany. 
The relative decrease in scanned connections is thus higher on the London instance and as a consequence the achieved speedups are higher.

\begin{table}
\begin{center}
\begin{tabular}{ccccccr}
\toprule
& & Pre- & Limited & Source & Range & Running \\
Instance & SIMD & fetch & Walk. & Dom. & Query & Time [ms]\\
\midrule
Germany & --- & $\circ$ & $\circ$ & $\circ$ & $\circ$ & 8\,298.5 \\
Germany & --- & $\bullet$ & $\circ$ & $\circ$ & $\circ$ & 7\,109.3 \\
Germany & SSE & $\circ$ & $\circ$ & $\circ$ & $\circ$ & 4\,792.2 \\
Germany & SSE & $\bullet$ & $\circ$ & $\circ$ & $\circ$ & 4\,612.6 \\
Germany & SSE & $\bullet$ & $\bullet$ & $\circ$ & $\circ$ & 3\,519.9 \\
Germany & SSE & $\bullet$ & $\bullet$ & $\bullet$ & $\circ$ & 2\,834.6 \\
Germany & SSE & $\bullet$ & $\bullet$ & $\bullet$ & $\bullet$ & 279.5 \\
Germany & AVX & $\circ$ & $\circ$ & $\circ$ & $\circ$ & 4\,402.9 \\
Germany & AVX & $\bullet$ & $\circ$ & $\circ$ & $\circ$ & 4\,332.7 \\
Germany & AVX & $\bullet$ & $\bullet$ & $\circ$ & $\circ$ & 3\,220.7 \\
Germany & AVX & $\bullet$ & $\bullet$ & $\bullet$ & $\circ$ & 2\,489.6 \\
Germany & AVX & $\bullet$ & $\bullet$ & $\bullet$ & $\bullet$ & 259.2 \\
\midrule
London & --- & $\circ$ & $\circ$ & $\circ$ & $\circ$ & 777.5 \\
London & --- & $\bullet$ & $\circ$ & $\circ$ & $\circ$ & 749.1 \\
London & SSE & $\circ$ & $\circ$ & $\circ$ & $\circ$ & 424.1 \\
London & SSE & $\bullet$ & $\circ$ & $\circ$ & $\circ$ & 420.1 \\
London & SSE & $\bullet$ & $\bullet$ & $\circ$ & $\circ$ & 261.2 \\
London & SSE & $\bullet$ & $\bullet$ & $\bullet$ & $\circ$ & 213.8 \\
London & SSE & $\bullet$ & $\bullet$ & $\bullet$ & $\bullet$ & 11.9 \\
London & AVX & $\circ$ & $\circ$ & $\circ$ & $\circ$ & 355.6 \\
London & AVX & $\bullet$ & $\circ$ & $\circ$ & $\circ$ & 359.8 \\
London & AVX & $\bullet$ & $\bullet$ & $\circ$ & $\circ$ & 206.1 \\
London & AVX & $\bullet$ & $\bullet$ & $\bullet$ & $\circ$ & 170.2 \\
London & AVX & $\bullet$ & $\bullet$ & $\bullet$ & $\bullet$ & 10.7 \\
\bottomrule
\end{tabular}

\end{center}
\caption{Profile computation running times with optimization of the number of legs and the earliest arrival time in the Pareto-sense.}
\label{tab:pareto}
\end{table}

In Table~\ref{tab:pareto}, we report running times of the Connection Scan Pareto profile algorithm.
It optimizes the number of legs, the arrival time, and the departure time in the Pareto-sense.
The maximum number of legs is set to 8.
We use the algorithm variant that computes the earliest arrival time in the 8-th vector component.
We iteratively activate our proposed optimizations to demonstrate their effectiveness.

We present three SIMD variants.
All three use the same memory layout.
All use vectors with 256 bits that contain 8 components with a 32-bit timestamp.
They differ in what processor instructions are used to operate on the vectors.
The first variant uses no special instructions and works with loops with a fixed number of iterations.
The second variant uses SSE instructions. 
SSE registers are 128 bits wide.
To process one vector, two SSE instructions are thus required.
The third variant uses AVX registers.
Luckily, these are 256 bits wide and therefore a single instruction is sufficient.
We use integer AVX arithmetic instructions. 
These were introduced with AVX2, a feature introduced in the Haswell processor architecture.
Our AVX code can therefore not run on our older test machine, which does not yet support AVX2.

The first optimization that we consider consists of prefetching memory.
On the Germany instance without SSE nor AVX, a speedup of 1.16 was achieved.
This is significant, considering that no algorithmic changes were performed.
Interestingly, the speedup is only 1.02, when comparing the AVX prefetch and AVX non-prefetch running times.
It is also interesting that by using AVX compared to the base version a speedup of 1.9 is achievable.
Especially, the later is interesting, as we expect SIMD to have the largest benefit in compute-bound algorithms and our previous experiments suggest that the Connection Scan algorithm heavily depends on memory access speeds.
One explanation for these two effects is that the AVX code has fewer instructions, making it easier for processor to predict memory access patterns. 
This would explain why the benefit of prefetching nearly vanishes but running times drastically decrease.
This explanation is also consistent with the observation that using AVX is a benefit over SSE as the AVX code requires fewer instructions.

The speedups of the limited walking and source domination optimizations are comparable to those observed for the earliest arrival profile algorithm.
We refer to discussion of these experiments for an interpretation of the observed effects.
The speedup of the range query variant is about 10 on the Germany instance and 17-19 on the London instance.
These speedups are larger than those observed for the earliest arrival profile algorithm.
The difference is likely due to the Pareto algorithms having a larger overall memory consumption.
As a consequence caching effects have a larger impact and therefore a reducing of the memory footprint yields a large relative advantage.

\paragraph{Comparison with Related Work.}

\begin{table}

\begin{center}
\begin{tabular}{ccccr}
\toprule
Instance & Algorithm & Pareto & One-to-one & Running Time [s]\\
\midrule
Germany & CSA & $\circ$ & $\circ$ & 1.68 \\
Germany & CSA & $\circ$ & $\bullet$ & 1.25 \\
Germany & CSA & $\bullet$ & $\circ$ & 3.22 \\
Germany & CSA & $\bullet$ & $\bullet$ & 2.49 \\
Germany & SPCS-col & $\circ$ & $\circ$ & 10.95 \\
Germany & SPCS-col & $\circ$ & $\bullet$ & 8.40 \\
Germany & rRAPTOR & $\bullet$ & $\circ$ & 6.27 \\
Germany & rRAPTOR & $\bullet$ & $\bullet$ & 4.73 \\
\midrule
London & CSA & $\circ$ & $\circ$ & 0.14 \\
London & CSA & $\circ$ & $\bullet$ & 0.11 \\
London & CSA & $\bullet$ & $\circ$ & 0.21 \\
London & CSA & $\bullet$ & $\bullet$ & 0.17 \\
London & SPCS-col & $\circ$ & $\circ$ & 1.19 \\
London & SPCS-col & $\circ$ & $\bullet$ & 0.79 \\
London & rRAPTOR & $\bullet$ & $\circ$ & 0.97 \\
London & rRAPTOR & $\bullet$ & $\bullet$ & 0.68 \\
\bottomrule
\end{tabular}

\end{center}
\caption{Comparison of profile algorithms.}
\label{tab:profile_compare}
\end{table}

In Table~\ref{tab:profile_compare}, we compare the Connection Scan profile algorithm with two competitor algorithms.
The first is Self-Pruning Connection-Setting (SPCS) algorithm~\cite{dkp-pcbcp-12}.
It computes profiles that optimize departure and arrival time in the Pareto-sense but does not optimize transfers.
The algorithm can be combined with the colored timetable optimization, which was used in our experiments.
We therefore refer to the algorithm as SPCS-col in Table~\ref{tab:profile_compare}.
The second competitor is rRAPTOR~\cite{dpw-rbptr-14}.
Similar to the base RAPTOR algorithm, it inherently optimizes transfers in the Pareto-sense.
The Connection Scan algorithm (CSA) was run with AVX and limited walking activated.

Both rRATPOR and CSA clearly dominate SPCS-col in terms of running time.
The difference between CSA and rRATPOR is smaller.
CSA is always faster, but on the Germany instance the gap is only up to a factor of 2.
On the London instance, there is a speedup of up to 4.7.

\paragraph{Section Conclusions.}

Using CSA and by exploiting the full capabilities of modern processors, it is possible to answer Pareto range queries on the large Germany instance in a quarter of a second.
It is feasible to construct interactive timetable information systems upon these running times.
However, ideally lower running are desirable.
For example, spending a quarter of a second per query in a web server severely limits throughput.
Fortunately, we were able to achieve these running times without compromising the excellent data structure construction times of the base algorithm.
Flexible realtime updates are possible.

\section{Connection Scan Accelerated}

In the previous sections, we present the Connection Scan family of algorithms.
We demonstrate that queries can be answered very quickly on modern hardware.
Even Pareto range queries can be answered in well below a second even on the large Germany instance.
A significant advantage of the Connection Scan algorithms is that the preprocessing is very lightweight.
It mainly consists of sorting the connections, which can be done in very few seconds.
This allows us to quickly update the timetable to account for disturbances, such as delayed trains, blocked stops or tracks, or overbooked trains.

While, all of these properties make the Connection Scan family of algorithms a good fit for many applications, it is also interesting to investigate whether further gains are achievable by using more heavy-weight preprocessing techniques.
Further, even though the achieved running times on the Germany instance are low enough for interactive applications, we expect them to consume a significant amount of resources.
Lower running times are therefore very desirable in practice.
Investigating the combination of Connection Scan with more heavy-weight preprocessing techniques is therefore the topic of this section.

We investigate a multilevel overlay extension to the Connection Scan algorithms, which we call Connection Scan Accelerated (CSAccel).
The central ideas are similar to those used in~\cite{sww-daola-99,hsw-emlog-08,dgpw-crprn-13}.
In several studies, this approach has proven to enable very fast queries in road networks.
Compared to Dijkstra's algorithms, speedups on the order of 1000 are possible.
It is therefore reasonable to expect similar speedups on timetable networks.
We are not the first to investigate this question.
Unfortunately, previous research~\cite{bdgm-atdmc-09,bgm-fdsut-10} has shown that achieving similar speedups is harder than one would naively expect.
Our work is no exception to this observation.
Our multilevel extension manages to provide a significant speedup on the Germany instance.
However, the speedup lacks far behind of what is achievable in road networks.

The core idea of our extension is best illustrated using an example: 
When planning a journey from Karlsruhe to Stuttgart, do not scan rural bus connections around Hamburg.
We use overlays to formalize the concept of rural bus.
Our algorithm partitions the stop set into cells.
Karlsruhe and Stuttgart are put into the same cell.
Hamburg is in a different cell.
For every cell, our algorithms computes a subset of \emph{transit connections}.
For every pair of connections entering and leaving a cell $z$, there must be a journey with a minimum number of transfers, that only enters or exits trips at transit connections of $z$.
For rural buses, usually no such journey exists and thus they are not in the transit connection set.
When traveling from Karlsruhe to Stuttgart, our algorithm only looks at the transit connections of Hamburg's cell and thus skips the rural buses around Hamburg.

Following the setup and terminology of~\cite{dgpw-crprn-13}, our algorithm works in three phases.
In the first phase, called \emph{preprocessing phase}, a multilevel partition of the stop set is computed.
In the second phase, called \emph{customization phase}, our algorithm computes overlays for every cell.
Finally, in the third phase, called \emph{query phase}, our algorithm computes arrival times and journeys.
The second phase uses the results of the first phase.
Similarly, the third phase uses the results of the first and the second phases.
The preprocessing phase should only use data that rarely changes, such as for example what tracks exist and perhaps what tracks are highly frequented.
The idea is that the preprocessing phase does not have to be rerun very often and may therefore be slow.
To update the timetable, it should be sufficient to rerun the customization, which should be fast.
Our customization phase works with every stop partitioning, as long footpaths do not cross cell boundaries and the stop sets are identical.
However, if the timetables used during preprocessing and customization differ too much, then customization and query performance will significantly degrade. 

multilevel approaches inherently rely on the structure of the network.
Small, balanced graph cuts are a necessity.
Without these, the achievable speedups crumble.
Fortunately, as shown in many studies, road graphs typically have this structure.
However, for timetables, the situation is less clear.
Indeed, country-wide timetables that consist of many urban centers differ in structure from timetables that consist of a single large urban region.
There typically exist small, balanced cuts between cities, however, cutting through a city is significantly more difficult.
Many cities contain natural cuts such as rivers or large main roads.
This property is exploited to achieve fast shortest path queries in road networks.
Unfortunately, in timetable networks, rivers are not necessarily advantageous. 
Often, several trains or buses lines pass over a single bridge.
Cutting through tracks with a high public transit frequency is expensive, in the context of timetables, as we need to weight the cuts by the number vehicles that pass over it. 
We therefore expect the performance of all multilevel overlay extensions to perform poorer on pure urban instances. 
This differs from the basic Connection Scan algorithm, whose performance is nearly independent of the timetable structure.

Connection Scan algorithms find a journey $j$ with legs $l^1,l^2 \ldots l^k$, if the connections $l^1_\exit$, $l^1_\enter$ $\ldots$ $l^k_\exit$, $l^k_\enter$ are scanned in the correct order.
These are the connections where the traveler transfers, i.e., enters or exits.
A connection where the traveler does neither does not have to be scanned.
Scanning all connections ordered by departure time fulfills this property for all journeys.
This is the core observation exploited by the Connection Scan base algorithms.
For a fixed source and target stop it can be sufficient to only scan a subset of the connections.
Our algorithm exploits this observation.
Our query phase thus works in two phases.
In the first phase a sorted connection subset $\mathcal{C_S}$ is assembled.
For every pair in the $st$-profile, there must be a journey $j$, such that all transfer connections of $j$ are included in $\mathcal{C_S}$.
In the second phase the Connection Scan base algorithms are run restricted to the connections in $\mathcal{C_S}$.

Our algorithm computes $\mathcal{C_S}$ by merging arrays of sorted connections.
In the base setting every cell has an associated sorted array of transit connections.
To compute $\mathcal{C_S}$, one would identify all potentially relevant cells and merge their transit connections.
Unfortunately, the number of these cells can be large and merging sorted arrays is a task that requires some running time.
We therefore want to reduce the number of arrays merged.
We therefore introduce the concept of \emph{long distance connections}.
A transit connection of a cell $z$ is a long distance connection of its direct parent cell.
On the lowest level, all connections within a cell $z$ are long distance connections of $z$. 
For every cell, our algorithm stores a sorted array of long distance connections.
To assemble $\mathcal{C_S}$, our algorithm merges the long distance connections of all cells that contain the source or target stop or both.

If the long distance connections of a cell $z$ are merged into $\mathcal{C_S}$, then also the connections of $z$'s parent are merged.
We can exploit this observation to further thin out the long distance connection set.
If $c$ is a long distance connection of a cell $z$ and of $z$'s parent cell, then it is sufficient to store $c$ in the parent cell's array.
Further, we can construct the transit connections with the property that, if $c$ is a long distance connection of $z$'s parent, then $c$ is a long distance connection of $z$.
A consequence of this is that every connection is contained in at most one thinned out long distance connection set.
The memory consumption is therefore linear in the number of connections.

To prove that our algorithm is correct, we show that for every Pareto-optimal journey $j$, there exists a Pareto-optimal journey $j'$ that only enters or exists trips in the merged connection subset $\mathcal{C_S}$, such that $j$ and $j'$ have the same departure and arrival time and have the same number of legs.
Before formally proving the correctness, we illustrate the employed arguments using an example.

\begin{figure}
\begin{center}
\includegraphics[scale=2]{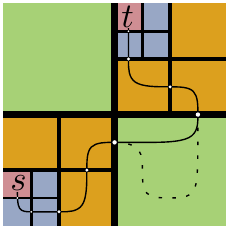}
\end{center}
\caption{
Multi level journey example from stop $s$ to stop $t$. 
}
\label{fig:multilevel-journey}
\end{figure}

Figure~\ref{fig:multilevel-journey} illustrates a stop set that was recursively partitioned along the straight solid lines.
At every level, every cell was partitioned into four parts.
The thickness of the lines indicates the level -- the thicker the line the higher a level.
The solid bent line represents the journey $j$.
The colored areas represent transit connections merged into $\mathcal{C_S}$.
Red means lowest level, blue is next higher, then orange and green is the highest level.
The white dots represent connections, where $j$ crosses cell boundaries.
The dotted line represents an alternative subjourney of $j$ within the green bottom-right cell.

The journey $j$ consists of a prefix from $s$ to the first boundary connection, several subjourneys that traverse cells, and a suffix from the last boundary connection to $t$.
The subjourneys are enclosed by the white dots in Figure~\ref{fig:multilevel-journey}.
The constructed journey $j'$ has the same prefix and suffix and crosses the cell boundaries in the same connections as $j$, i.e., in the white dots.
Because the prefixes and the suffixes are equal, the departure and arrival times of $j$ and $j'$ are equal.
The subjourneys within a cell can differ.
For example it is possible that $j$ uses the solid line, whereas $j'$ uses the dotted line.
By construction, we know that for every cell, entry connection, and exit connection, there exists a subjourney with a minimum number of transfers that only enters or exits trips at transit connections in $\mathcal{C_S}$. 
For every cell $z$ that $j$ traverses, replace the subjourney of $j$ within $z$ with the corresponding minimum transfer journey to obtain $j'$.
$j'$ cannot have more transfers than $j$ because otherwise one of the employed subjourneys would not have had a minimum number of transfers.
Further, as $j$ was Pareto-optimal, $j'$ cannot have fewer transfers than $j$.
$j$ and $j'$ therefore have the same number of transfers.

In the following, we describe the details of our multi level extension.
The text is organized along the three main phases.
It first describes the the preprocessing phase which mostly consists of a graph partitioning problem.
Afterwards, the customization phase is described, which primarily consists of computing the transit connections.
Next, we explain how to perform the queries, which consists of computing $\mathcal{C_S}$.
Finally, we present an experimental evaluation of the algorithm and a comparison with related work.

\subsection{Phase 1: Partitioning the Stop Set}

A $k$-partition of the stop set $V$ divides $V$ into $k$ cells such that every stop is in exactly one cell.
We require that stops connected by footpaths must be in the same cell, i.e., footpaths must not cross cell borders.
A connection is interior (exterior) to a cell if it departs at a stop inside (outside) the cell. 
In an $l$-level partition with $k$ children, the stop set is recursively split into $k$ cells over $l$ levels.
At the bottom level there are $k^l$ cells.
The top level consists of a single cell that contains all stops.
The parent $p$ of a cell $z$ is the cell which was split to create $z$.
Similarly, $z$ is a child of $p$.
The bottom level cells do not have children and the top level cell does not have a parent. 

The preprocessing step consists of computing an $l$-level partition with $k$ children where $l$ and $k$ are tuning parameters of the algorithm.
We perform the partitioning using a graph partitioner.
From the timetable, we build an undirected, weighted graph as follows:
The stops form the node set of the graph. 
There is an edge between two nodes if there is a connection or footpath between the corresponding stops.
If there is a footpath, we weight the corresponding edge with $\infty$, to assure that it is not cut.
Otherwise, the weight of an edge reflects the number of connections between the edge's endpoints.
We partition the graph into $k$ parts using KaHip v1.0c\footnote{We also tried Metis in a preliminary experiment and the resulting query and customization running times were dominated.} with 20\% imbalance.
We recursively repeat this operation $l$ times.
We run KaHip using the ``strong''-preconfiguration. 
Unfortunately, the results we get from KaHip vary significantly depending on the random seed given to it.
We therefore run KaHip at each level in a loop with varying seeds until for 10 iterations no smaller cut is found.
This setup is definitely not the fastest partitioning method.
Fortunately, it is fast enough and the obtained cuts are reliably small.

\subsection{Phase 2: Computing Transit Connections}

In this section, we describe how to compute the transit and long distance connections.
We start by describing how to compute journeys with a minimum number of transfers.
In the next step, we describe how to use this algorithm to compute transit and long distance connections sequentially.
Finally, we describe how the customization algorithm can efficiently be parallelized.

\paragraph{Minimum Number of Transfers.}

\begin{figure}
\begin{algorithm2e}[H]
\lFor{all stops $x$}{$S[x]\gets\{ (\infty, \infty) \}$}
\lFor{all trips $x$}{$T[x]\gets\infty$}
\BlankLine
\For{all footpaths $f$ with $f_\arrstop = (c_t)_\depstop$}{
	Incorporate $((c_t)_\deptime-f_\dur,0)$ into profile of $S[f_\depstop]$\;
}
$T[c_t]\gets 0$\;
\BlankLine
\For{all connections $c$ strictly before $c_t$ decreasing by $c_\deptime$}{
	$\tau_2\gets T[c_\trip]$\;
	$\tau_3\gets($evaluate~$S[c_\arrstop]$ at $c_\arrtime)+1$\;
	\BlankLine
	$\tau_c\gets \min \{\tau_2,\tau_3\}$\;
	\BlankLine
	\If{$(c_\deptime,\tau_c)$ is non-dominated in profile of $S[c_\arrstop]$}{
		\For{all footpaths $f$ with $f_\arrstop = c_\depstop$}{
			Incorporate $(c_\deptime-f_\dur,\tau_c)$ into profile of $S[f_\depstop]$\;
		}
	}
	$T[c_\trip]\gets \tau_c$\;
	\BlankLine
	\If{$c\in C_s$}{
		Extract journey from $c$ to $c_t$ and mark transit connections\;
	}
} 
\end{algorithm2e}
\caption{Minimum transfer profile algorithm between connections.}
\label{alg:profile-min-trans} 
\end{figure}

To compute overlays, our algorithm needs to quickly compute journeys with a minimum number of transfers between each pair of connection entering and leaving a cell $z$.
We implement this using a variant of the earliest arrival Connection Scan profile algorithm with secondary transfer optimization.
We run this algorithm on a part of the network restricted to the connection inside of $z$.
Contrary to the algorithms described in Section~\ref{sec:base-csa}, the traveler does not start and end at a stop but starts in an entry connection $c_s$ and ends in an exit connection $c_t$.
A key observation is that, for a fixed target exit connection $c_t$, the arrival times of all journeys are the same.
The algorithm therefore only optimizes the number of transfers.

Our algorithm iterates in an outer loop over the exit connections $c_t$ and computes a backward profile for each.
In an inner loop, it iterates over all entry connections and evaluates the profile. 
It extracts a corresponding journey $j$ from $c_s$ to $c_t$.
All connections where $j$ exits or enters a trip are marked as transit connections including $c_s$ and $c_t$.

In the following, we describe the inner loop of our algorithm in greater detail.
In the inner loop, we have a fixed target exit connection $c_t$ and a set of source enter connections $C_s$.
For every connection $c_s$ in $C_s$, a minimum transfer journey from $c_s$ to $c_t$ should be computed.
The pseudo-code of the algorithm is given in Figure~\ref{alg:profile-min-trans}.
We start our scan with the connection $c_t$, as all connections after it are obviously not reachable.
The body of the loop is left mostly unchanged compared to the base algorithm. 
There is only one major modification:
We no longer compute a walk-to-target time $\tau_1$. 
It would be $\infty$ for every connection except $c_t$, which is not useful.
Instead, we introduce a special case for $c_t$ outside of the loop.
As all journey end in $c_t$, it does not matter what arrival time we give $c_t$.
For simplicity, we use 0.

To quickly extract journeys, we use journey pointers. 
The extraction works analogous to the base algorithm with one modification.
To extract the first leg of a journey starting in a connection $c_s$, we need to look at the exit connection stored with $c_s$'s trip.
However, this exit connection may be overwritten, if there are several entry connection inside of this trip.
We therefore extract the journey directly after processing $c_s$.
A trip can contain multiple entry connections, if it leaves and enters a cell multiple times. 

\paragraph{Computing Transit and Long-Distance Connections.}

We compute transit connections bottom-up, i.e., the transit connections of the lowest level are computed first.
To accelerate the computations on the higher levels, we use transit connections of lower levels.
A central observation is that for every cell $z$ there is a valid transit connection set $T_z$ that is a subset of the long-distance connection set $L_z$ of $z$.
Our algorithm thus works as following:
For all levels $l$ from bottom to the top and all cells $z$ in the level $l$, first compute the long-distance connections $L_z$ of $z$, then compute the transit connections $T_z$ of $z$ by restricting the search to the long-distance connections $L_z$.
For the lowest level cells, the long-distance connection set contains all interior connections.
In a second faster step, we iterate a second time over the levels and cells and thin out the long distance connection sets.

\paragraph{Parallelization.}

Significant speedups can be achieved by parallelizing the transit connection computation.
There are two levels of granularity on which we can parallelize: (1) we can compute the transit connections of cells on the same level in parallel, and (2) we can compute the journeys for different exit connections within a cell in parallel.
The former has the advantage that the data structures of different cells are completely disjoint, minimizing the necessary communication and synchronization. 
However, the boundary sizes of cells are very skewed because of urban centers.
It is therefore difficult to keep all threads fully occupied.
The later is more fine-grained and therefore allows us to fully occupy all threads.
However, more communication and synchronization is needed.

We use a hybrid approach that combines the best properties of both.
In a first step, we sort all cells first by level from bottom to top and as a secondary criterion by decreasing boundary size.
The obtained list is a topological sorting of the dependencies between the cells.
We sort the cells by boundary size, to assure that the more expensive cells, i.e., those with a larger boundary, are processed first. 
We attach to every cell an atomic counter, that indicates the number of children cells have not yet been computed.
If this counter reaches zero then processing can start.
The bottom level cells start with a counter of zero.
The higher level cells start with the number of children used in the partitioning.
We spawn as many threads as the hardware can process simultaneously.
Every thread iterates over the list of cells once.
If it finds a cell with counter 0, it grabs the cells by atomically increases the counter to prevent other threads from seeing the 0 counter value.
The thread then processes the cell and once it is finished decreases the counter of its parent.
When a thread reaches the end of the list, it puts itself into a pool of idle threads.
The threads that are still processing cells, look at whether this pool is non-empty between processing two target exit connections.
If it is non-empty, they extract an idle thread atomically and the thread helps processing the cell.
At the end of processing a cell, all threads but the main one are put back into the idle pool.

\subsection{Phase 3: Answering Queries}

In this section, we describe how to compute the connection subset $\mathcal{C_S}$ and how the query algorithms need to be modified.
\paragraph{2-way vs $k$-way.}
Efficiently computing $\mathcal{C_S}$ is a crucial component of an efficient implementation of our query algorithms.
The input consists of several arrays of sorted data that should be merged.
Three major strategies exists~\cite{k-taocp-97}.
The first consists of iteratively performing a two-way merge to combine pairs of arrays.
The other two are direct $k$-way merges. 
The idea consists of storing a pointer into each array and iteratively determining the smallest element and increasing the corresponding pointer.
Determining which element is the smallest is the challenging part.
There are two approaches. 
One can use a binary heap or one can use tournament trees.
All three variants have a worst case running time of $O(n \log k)$, where $n$ is the total number of elements.
We implemented all three variants and in preliminary experiments on our data set, the iterative two-way merge was the fastest, followed by the binary heap, and the tournament heaps came last.
Unfortunately, the iterative two-way merge can only compute $\mathcal{C_S}$ as a whole.
The direct $k$-way approach allows us to perform a partial merge, i.e., only merge the first $x$ connections, which is enough for some of our applications.

\paragraph{Profile Queries.}

We implement the earliest arrival and Pareto profile algorithms in the straight forward way.
First, our algorithm computes $\mathcal{C_S}$ using an iterative two-way merge.
In a second step, the Connection Scan base algorithm is applied restricted to $\mathcal{C_S}$.

\paragraph{Earliest Arrival Queries.}
Earliest arrival queries have a start and stop criterion.
We therefore use a direct $k$-way merge to avoid computing parts of $\mathcal{C_S}$ that will not be scanned.
For each of the $k$ arrays, we run a binary search to determine the first connection not before the source time.
We then start the $k$-way merge. 
We run the merging process until the stop criterion aborts the scan.

\paragraph{Range Queries.}
For range queries, we use a similar approach.
We first perform the $k$ binary search and then start with the $k$-way merge.
To determine the reachable trips, we execute a non-profile earliest arrival scan.
Once the stop-criterion activates, we continue the merge until all connections departing within the desired range have been computed.
We store the output of the merging process into a temporary array.
We then run the profile algorithm restricted to the connections in this temporary array.

\subsection{Experiments}

\label{sec:csa_accel_exp}

In this section, we experimentally evaluate CSAccel.
We use the experimental setup described in Section \ref{sec:experimental-setup}.
We start by comparing various multilevel configuration in terms of preprocessing, earliest arrival query, and profile query running time.
For one of the best configurations, we present an evaluation of range queries.
We conclude with a comparison of experimental results with related work.

\begin{table}
\begin{center}
\begin{tabular}{lrrrrrrr}
\toprule
 & Low & \multicolumn{2}{c}{Setup [s]} &  & \multicolumn{3}{c}{Query [ms]} \\
\cmidrule(lr){3-4} \cmidrule(lr){6-8}
Instance & Cell & Part. & Cust. & Conn [K] & EA & EA-Prof & Par-Prof \\
\midrule
Germany-2-9 & 512 & {\fontseries{b}\selectfont 2\,483.7} & 157.7 & 897.2 & 6.6 & 49.1 & {\fontseries{b}\selectfont 75.0}\\
Germany-2-12 & 4\,096 & 7\,300.5 & 329.1 & {\fontseries{b}\selectfont 751.0} & {\fontseries{b}\selectfont 6.2} & {\fontseries{b}\selectfont 47.3} & 78.9\\
Germany-3-5 & {\fontseries{b}\selectfont 243} & 2\,604.8 & 114.5 & 1\,184.6 & 7.2 & 56.7 & 85.9\\
Germany-3-7 & 2\,187 & 4\,918.8 & 220.3 & 868.8 & 6.5 & 51.0 & 79.3\\
Germany-4-5 & 1\,024 & 3\,746.7 & 157.7 & 1\,023.4 & 7.2 & 55.6 & 84.6\\
Germany-4-6 & 4\,096 & 7\,214.2 & 229.1 & 988.9 & 7.0 & 57.1 & 89.0\\
Germany-8-3 & 512 & 3\,170.2 & {\fontseries{b}\selectfont 113.6} & 1\,331.4 & 7.9 & 66.2 & 99.3\\
Germany-8-4 & 4\,096 & 7\,367.9 & 176.0 & 1\,252.2 & 7.7 & 67.5 & 102.3\\
\midrule
London-2-7 & 128 & 253.5 & 101.2 & 1\,933.6 & 2.6 & {\fontseries{b}\selectfont 91.7} & {\fontseries{b}\selectfont 134.4}\\
London-2-10 & 1\,024 & 838.1 & 126.6 & {\fontseries{b}\selectfont 1\,920.8} & 2.6 & 96.3 & 137.5\\
London-3-3 & {\fontseries{b}\selectfont 27} & {\fontseries{b}\selectfont 124.3} & 54.1 & 2\,181.2 & 2.5 & 99.2 & 140.6\\
London-3-5 & 243 & 338.3 & 74.1 & 2\,085.0 & 2.3 & 92.6 & 137.5\\
London-4-3 & 64 & 230.0 & 51.7 & 2\,226.2 & 2.3 & 95.0 & 141.2\\
London-4-5 & 1\,024 & 718.6 & 67.1 & 2\,193.6 & 2.2 & 97.1 & 141.5\\
London-8-2 & 64 & 186.7 & {\fontseries{b}\selectfont 32.9} & 2\,490.1 & 2.0 & 97.7 & 147.9\\
London-8-3 & 512 & 579.9 & 40.6 & 2\,464.9 & {\fontseries{b}\selectfont 1.9} & 97.1 & 147.0\\
\bottomrule
\end{tabular}

\end{center}

\caption{Preprocessing, and customization running times, number of lowest level cells, number of connections in filter, and average query running times for earliest arrival time, earliest arrival profile, and Pareto profile.
Preprocessing and customization were run on the older machine. Customization was parallelized with 16 threads.}
\label{tab:accel_times}
\end{table}

\paragraph{Query Experiments.}
In Table~\ref{tab:accel_times}, we experimentally evaluate Connection Scan Accelerated for various configurations.
A label X-$c$-$l$ refers to a recursive partitioning of timetable X, over $l$ levels, with $c$ children per level.
The number of lowest level cells is $c^l$.
We report $c^l$ in the table to give an overview over the granularity of the partition.
We report the time needed to compute the multilevel partitioning with KaHip version 1.0c.
In preliminary experiments, we also tried using Metis. 
The partition running times were significantly lower but the customization and query running times were higher.
As we focus on the later two values, we therefore refrain from reporting these experiments.
Further, we report the customization running times.
Both the preprocessing and customization experiments were performed on our older Xeon E5-2670 machine with 16 physical hardware threads. 
The customization running times are parallelized, whereas the preprocessing running times are sequential.
We also report running times for various query variants.
The query experiments were run sequentially on the newer Xeon E5-1630v3 machine.
We report the average running times for the earliest arrival time, the earliest arrival profile, and the Pareto profile problem settings.
Journey extraction was not performed. 
Range query experiments are reported in Table~\ref{tab:accel_range} and discussed later in this section.
We activated all optimizations of the base algorithm, i.e., start and stop criteria, source domination, limited walking, and AVX.
Beside the query running times, we also report the number of connections in $\mathcal{C_S}$.
These are the number of connections that are scanned by the profile algorithms.
The earliest arrival algorithm only needs to scan a subset of these connections because of the start and the stop criteria.

The preprocessing running times roughly grows with the number of lowest level cells. 
This is non-surprising, as the number of partitioner invocations follows this trend.
The customization running times follow the same general trend and grow with the number of cells.
However, having a large number of children helps the customization but hampers the partitioning.
The minimum partitioning running times are therefore achieved for Germany-2-9 and London-3-3, which have a low number of children, whereas the customization running times are minimum for Germany-8-3 and London-8-2, i.e., a high number of children.
To minimize the number of connections, a recursive bisection strategy with many levels performs best.
Scanning fewer connections reduces the running time spent in the Connection Scan algorithm.
Query running times are therefore comparatively fast for nested dissection configurations.
The only exception to this trend is the earliest arrival running time on London, which is fastest for London-8-3 and London-8-2.
The explanation is that the $k$-way merge step dominates the running time.
Having more levels results in more arrays to be merged and thus increases the running time of the merge step.
London-8-3 and London-8-2 have the fewest levels and therefore the fastest merge steps.

Compared to the non-accelerated running times, we observe a significant decrease in running times for every query type on the Germany instance.
However, the speedups are significantly less impressive on the London instance.
The explanation is that the London instance only has 4\,850K connections but even for London-2-10 1\,921K connections have to be scanned.
The speedup is therefore very slim.
In fact for the earliest arrival time problem, the base algorithm is even faster.
The explanation is that it is faster to scan the few additional connections, than to perform the $k$-way merge.

It is very surprising that CSAccel is faster in absolute terms on the Germany instance compared to the London instance.
There are several reasons for this effect.
London has at the time of writing nearly 9M inhabitants. 
This contrasts with the largest German city Berlin that has only 3.5M inhabitants.
As a consequence the London urban transit is larger than any urban transit contained in the Germany instance.
Another explanation is the difference in stop modeling.
The London instance has a reflexive transfer model with usually one stop per platform.
The Germany instance groups nearby platforms into one stop and uses loops in the footpath graph.
London is thus modeled in greater detail than Berlin.
Having more stops increases computation times.

\begin{table}
\begin{center}

\begin{tabular}{ccr}
\toprule
Instance & Pareto & Running Time [ms]  \\
\midrule
Germany-2-12 & $\circ$ & 17.9 \\
Germany-2-12 & $\bullet$ & 24.7 \\
\midrule
London-2-7 & $\circ$ & 11.2 \\
London-2-7 & $\bullet$ & 12.0 \\
\bottomrule
\end{tabular}

\end{center}
\caption{Accelerated range queries average running times.}
\label{tab:accel_range}
\end{table}

\paragraph{Range Queries.}

In Table~\ref{tab:accel_range}, we report range query results.
We restrict our exposition to Germany-2-12 and London-2-7, as we obtained very good results for these configurations for non-range profile queries.
Compared to the profile query running times, we observe significant speedups.
These speedups are similar to those observed when comparing profile with range queries in the non-accelerated Connection Scan base algorithm.
The speedups are due to cache effects and fewer connections being scanned.

\begin{table}
\begin{center}
\begin{tabular}{lrrrrrrr}
\toprule
 & & & & \multicolumn{4}{c}{Query Running Time [ms]}\\

\cmidrule(lr){5-8}
 & \#Stop & \#Conn & Prepro & \multicolumn{2}{c}{Fixed-Dep} & \multicolumn{2}{c}{Profile} \\
\cmidrule(lr){5-6} \cmidrule(lr){7-8}

Algo                         & [K] &  [M] &  [min] &       EA & Pareto & EA  & Pareto \\

\midrule

RAPTOR~\cite{dpw-rbptr-14}   &      252.4 &       46.2 &          --- &      --- &             325.8 &      --- &             4\,730 \\
CSA                          &      252.4 &       46.2 &          0.1 &     44.9 & ${259.2}^\dagger$ &   1\,246 &             2\,490 \\
CSAccel-2-12                 &      252.4 &       46.2 &     (122)+88 &      6.2 &    $24.7^\dagger$ &     47.3 &               78.9 \\
TP~\cite{bs-fbspt-14}        &      248.4 &       13.9 &      22\,320 &      --- &               0.3 &      --- &                5.0 \\
S-TP~\cite{bhs-stp-16}       &      250.0 &       15.0 &          990 &      --- &              32.0 &      --- &                --- \\
TB-ST~\cite{w-tbptr-16}      &      247.9 &       27.1 &      13\,878 &      --- &             0.156 &      --- &              0.512 \\
TB~\cite{w-tbptr-15}         &      249.7 &       46.1 &           39 &      --- &              40.8 &      --- &              301.7 \\

\midrule

RAPTOR~\cite{dpw-rbptr-14}   &       20.8 &        4.9 &          --- &      --- &              6.4 &      --- &             680 \\
CSA                          &       20.8 &        4.9 &      $<$ 0.1 &      1.2 & ${10.7}^\dagger$ &    106.9 &           170.2 \\
CSAccel-2-7                  &       20.8 &        4.9 &       (4)+27 &      2.6 &   $12.0^\dagger$ &     91.7 &           134.4 \\
PTL~\cite{ddpw-ptl-15}       &       20.8 &        5.1 &           54 &   0.0028 &              --- &   0.074 &             --- \\
Pareto-PTL~\cite{ddpw-ptl-15}&       20.8 &        5.1 &       2\,958 &      --- &           0.0266 &      --- &             --- \\
TB-ST~\cite{w-tbptr-16}      &       20.8 &        5.0 &          696 &      --- &              1.7 &      --- &            16.1 \\
TB~\cite{w-tbptr-15}         &       20.8 &        5.0 &            6 &      --- &              1.2 &      --- &            70.0 \\

\bottomrule
\end{tabular}

\end{center}

\caption{Comparison of various preprocessing-based algorithms for timetable routing.
The top results are for Germany instances and the bottom results for London instances.
}
\label{tab:accel_comp}
\end{table}

\paragraph{Comparison with Related Work.}

In Table~\ref{tab:accel_comp}, we compare various algorithms for timetable routing.
Some make use of very heavy-weight preprocessing, while others are very lightweight.
We compare RAPTOR~\cite{dpw-rbptr-14}, our Connection Scan algorithm (CSA), our multilevel extension (CSAccel), public transit labeling (PTL,Pareto-PTL),~\cite{ddpw-ptl-15}, Trip-Based routing (TB)~\cite{w-tbptr-15,w-tbptr-16}, and transfer patterns (TP)~\cite{bceghrv-frvlp-10,bs-fbspt-14,bhs-stp-16}.
Two PTL variants exist: the base version (PTL), and an extension that supports optimizing transfers in the Pareto-sense (Pareto-PTL).
There are also two variants of Trip-Based routing: the base variant TB~\cite{w-tbptr-15} and a newer version~\cite{w-tbptr-16} (TB-ST) that precomputes prefix and suffix trees.
Transfer patterns were introduced in~\cite{bceghrv-frvlp-10} and overhauled in~\cite{bs-fbspt-14}.
We refer to the overhauled version as TP.
Another variant called ``Scalable Transfer Patterns'' was introduced in~\cite{bhs-stp-16}.
We refer to it as S-TP.

The various papers use different instances that are based upon the same input data.
The only exception is S-TP which uses a newer version of the Deutsche Bahn data set.
Unfortunately, the papers significantly differ in how they extract a formal timetable from the input.
The variations on the London instance are comparatively small and originate from differences in how data errors are repaired.

The differences on the Germany instance are more significant.
S-TP is based on newer input data than TP and therefore the corresponding numbers differ.
The TP instance is based on the same input as the other papers.

CSAccel, TB, and TB-ST extract a two day instance.
TP and S-TP extract a single day but have days-of-operation flags.
Using these flags multiple days discerned.
The difference between a two day instance and a one day instance with flags explains the different number of connections between TP and TB.
The difference in size between TB-ST and TB originates from a different interpretation.
Following our original CSAccel paper, TB extracts all connections regardless of the day of operation.
This is done because some local operators do not have a schedule for every day.
The downside of this approach is that several variations of the same trip appear simultaneously.
For example some trips drive differently on Sundays than on workdays.
Fortunately, having more connections will most likely not decrease the running times.
The reported numbers of CSAccel and TB are therefore upper bounds.
The difference between CSAccel and TB is the result of correcting data errors differently.

These differences in instances makes a detailed comparison difficult, if not impossible.
We can only confidently compare orders of magnitude between the running times reported in the various papers.
We therefore refrain from scaling running times with respect to machines as the numbers are not directly comparable anyway.
Further, cache sizes can have a larger impact on the running time than the processor clock speed as demonstrated in Table~\ref{tab:ea_profile}.
Unfortunately, cache sizes are rarely reported in papers.
Scaling by processor clock speed is therefore not meaningful, even if the instances were equal.

All reported running times are sequentially.
The reason that the preprocessing times seem large, stems from the fact that papers usually report parallelized running times.
CSAccel is the only algorithm to split preprocessing into two phases.
We therefore report its preprocessing as $(p)+c$ where $p$ is the preprocessing and $c$ the customization running time.

Unfortunately, we cannot report numbers for every query type and algorithm.
This has various reasons.
For RAPTOR, we do not report non-Pareto running times because RAPTOR does not benefit from not optimizing transfers.
We report no preprocessing time for RAPTOR, because the original implementation that we use was not tuned for this criteria.
For CSA, we report range query running times instead of non-profile Pareto running times.
The reason is that we do not know how to implement non-profile Pareto queries in a way that significantly outperforms range queries.
Range queries usually compute more journeys because they allow for a flexible departure time.
In some sense the problem is therefore harder, 
However, the latest arrival time is bounded, which makes the problem also somewhat easier.
Fortunately, both problems have similar applications and therefore we present the results in the same column.
The CSA numbers are marked with a $\dagger$ to illustrate that range queries are computed.
PTL's preprocessing can optionally optimize transfers.
This explains the two PTL variants in the table.
The authors evaluated the non-transfer variant for earliest arrival time and earliest arrival profiles.
Unfortunately, the authors were not able to evaluate PTL on the Germany instance because of legal restrictions.
Further, they did not evaluate Pareto-profile queries.
The trip-based techniques TB and TB-ST, just as RAPTOR, do not benefit from not optimizing transfers in the Pareto-sense and thus no earliest-arrival-only numbers exist.
The transfer patterns techniques TP and S-TP could in theory be implemented in a variant that only optimizes arrival time.
This theoretical variant would probably benefit from smaller query graphs but it was, to the best of our knowledge, never implemented and thus we cannot report numbers.
Unfortunately, TP was not evaluated on the London instance.

\paragraph{Discussion of the Germany instance.}

Ordering the algorithms by preprocessing running times yields: CSA, RAPTOR, TB, CSAccel, S-TP, TB-ST, and finally TP.
With the exception of TB-ST and TP, the gaps between each of these techniques are large enough that we can be confident, that the differences are not solely due to differences in experimental setup.
Comparing query running times is more difficult because of the various query types.
Further, the differences between running times are smaller. 
It is thus possible that a number is only lower because of a different experimental setup.
With respect to non-profile Pareto query running times, the group of fastest algorithms clearly contains TP and TB-ST. 
The next-slower group contains CSAccel, S-TP, and TB.
The slowest group contains CSA and RAPTOR.
Meaningfully comparing algorithms within a group requires a more similar experimental setup.
Overall, CSAccel strikes a good trade-off between the various criteria. 
No query running time is above 100ms and preprocessing running times are manageable.

\paragraph{Discussion of the London instance.}

On the London instance, only PTL achieves a speedup above a factor of 11 over the CSA baseline.
Given the simplicity and near-instant preprocessing running times, this makes CSA a perfect fit for this instance.
PTL achieves an interesting performance trade-off when not optimizing transfers. 
The preprocessing time is slightly below an hour, which is still somewhat manageable.
The benefit is that PTL achieves query running times are on the microsecond scale.
Unfortunately, when additionally optimizing transfers the preprocessing running time of PTL becomes prohibitively large.
Overall, assuming that some form of transfer optimization is required, we recommend using CSA as it is never drastically slower than the alternatives but is simple to implement and can update the timetable almost instantly.

\paragraph{Section Conclusions.}

The conclusions we draw from the experiments are mixed and depend on the test instance.

On the Germany instance, CSAccel can answer Pareto range queries on average in about 25ms.
This is a significant improvement over the 250ms of CSA.
Interactive timetable systems with a high throughput can be constructed with an average query running times of 25ms.
However, the factor 10 speedup comes at a high cost.

The obvious cost is the increased preprocessing time.
CSA needs 10 seconds single core to adjust to a completely new timetable.
On the other hand, CSAccel requires 2min with 16 cores.
Requiring 2min to update the timetable is probably acceptable in practice but far from ideal.
Further, CSAccel requires that the new timetable is sufficiently similar to the old one.
CSA does not have this restriction.

A further cost associated with CSAccel is the significant increase in code and algorithm complexity compared to CSA.
Arguably the most important selling point of CSA is its simplicity. 
It is so simple that not even a heap-based priority queue is needed as a component.
The earliest arrival CSA base is arguably even easier than Dijkstra's algorithm.
CSAccel requires solving among other things a graph partitioning problem as subroutine.
This is an NP-hard task and the state-of-the-art heuristics alone have a complexity far exceeding that of CSA.
Depending on the application, the increase in complexity of CSAccel compared to CSA might even be worse than the increased preprocessing times.

However, for applications where query running times of 250ms are prohibitive and realtime updates are needed, CSAccel is still attractive because of the lack of alternatives.
None of the other techniques achieves preprocessing running times on the order of only a few minutes and similar query running times.

On the urban London network, the decrease in query running time of CSAccel over the CSA baseline is slim.
We do not believe that it outweighs the significantly larger preprocessing costs and especially not the significant increase in code complexity.
Use CSA in primarily urban networks.

An advantage of CSA is that its performance is nearly independent of the timetable structure and mostly depends on its size.
On the other hand, the performance of CSAccel is heavily dependent on the timetable structure, as the differences between the test instances shows.

When starting a new timetable information system, using CSA until the query running times get prohibitive is a good approach.
CSA is easy to implement and therefore not much effort is lost when switching to other approaches.
Further, chances are high that the size of your timetables will never reach the prohibitive size.
For example, we have not been able to assemble a realistic timetable with only rail-bound vehicles that was large enough.
The Germany test instance is only large enough because buses are included.

\subsection{Differences from original CSAccel publication~\cite{sw-csa-13}}

In this section, we briefly explain where the differences in experimental results between the experimental evaluation presented here and the original conference article~\cite{sw-csa-13} stems from.
If you have only read this journal article then you can ignore this section.

In~\cite{sw-csa-13} we report 1794.7 seconds to perform a customization using 16 threads on the same test machine on the Germany instance. 
We improved this to 113.6 seconds, which constitutes an improvement of over a magnitude.
This improvement is the combination of three smaller changes.

The first and largest improvement is due to an improved parallelization scheme.
The cell boundary sizes differ significantly.
However, in \cite{sw-csa-13} we only parallelized over the cells in a level but not not across level and not within a cell. 
The result was that during large parts of the customization only a single thread was working. 
The new parallelization approach manages to keep all threads occupied over nearly the whole process.

The second improvement is due to using KaHip instead of Metis. 
The newer KaHip versions achieve smaller cut sizes than Metis. 
This translates into slightly smaller query running times and drastically lower customization times.

The third and smallest improvement is due to implementing the minimum transfer CSA more carefully.

\section{Minimum Expected Arrival Time}
\label{sec:MEAT}

The Connection Scan profile framework is very flexible.
In the previous sections, we have seen how the timetable can be adjusted to account for known delays.
In this section, we want to plan ahead and compute a journey that is robust with respect to unknown, future delays.
We do this by computing for every transfer backup journeys.
For every transfer in a journey from train $A$ to train $B$, we compute a list of backup trains $C_1,C_2\ldots$ that the traveler can take if he cannot reach train $B$ because $A$ is delayed.
If a transfer breaks, then a traveler should take the backup train with the earliest departure time that he can get.

\begin{figure}
\begin{center}
\includegraphics{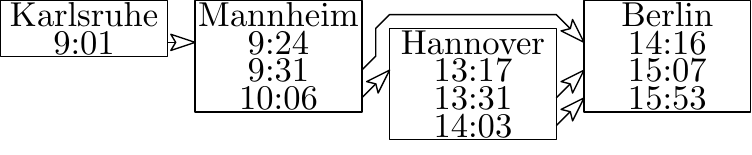}
\end{center}
\caption{Delay-Robust journey from Karlsruhe at 9:00 to Berlin.}
\label{fig:meat-example}
\end{figure}

An example of such a delay-robust journey from Karlsruhe to Berlin is depicted in Figure~\ref{fig:meat-example}.
We refer to the depicted graph as \emph{decision graph}.
If no train is delayed, then the traveler should take the train from Karlsruhe to Mannheim, transfer there, and take the direct train to Berlin.
The question is what the traveler should do, if he misses his connecting train in Mannheim. 
The answer is to take a later journey with an additional transfer in Hannover.
This additional transfer might also break.
The backup journey therefore needs its own backup journey.
Fortunately, the second backup goes directly to the target and therefore no further backup trains are necessary.

Further, examples can be generated using our proof-of-concept demonstration at \url{http://meatdemo.iti.kit.edu}.
As we expect readers to be unfamiliar with the concept, we highly recommend readers to experiment with the demonstration to get a basic understanding on an intuitive level before reading on. 

Computing such journeys is a very different setting than computing an earliest arrival journey with respect to a timetable aware of the realtime delay situation.
All the algorithm's decision need to be performed in advance, when the exact delays are not yet known.
To be able to do this, we assume that we have an estimation of how likely it is that a train is delayed.
We refer to this estimation as \emph{delay model}.
One way to obtain such an estimation is to aggregate historic delay data.
If a train was never delayed in the past, then we can assume that it is unlikely that it is delayed today.
If a train was nearly always delayed, then we definitely need to have a good backup journey.

Consider a train $A$ that is part of a very fast journey but no reasonable backup trains exist and $A$ is likely delayed.
A risk averse traveler will not want to take $A$ because it is too risky.
The algorithmically interesting question consists of identifying risky trains and avoiding them.
Neither optimizing the arrival time nor the number of trains achieves this.

While developing the Connection Scan algorithms, we have discovered a surprisingly easy way to solve this problem.
Consider the Connection Scan earliest arrival profile algorithm.
Suppose that the traveler arrives with a train $A$, transfers at a stop $X$, to take a train $B$.
We want to find the next best backup $C$ in case that the transfer breaks.
To compute this route our, the Connection Scan profile algorithm looks up $B$'s pair in $X$'s profile.
What will a traveler do, if he misses $B$? 
He will wait at $X$ for the next train $C$ heading into the correct direction to depart.
Formulated differently, $C$ is the backup train.
Computing $C$ is easy. 
The pair in $X$'s profile after $B$'s pair corresponds to $C$.

A problem remains.
If no reasonable backup exists, i.e., the transfer is very risky, then the so computed backup $C$ will arrive very late.
To solve this issue, we do not store the arrival time of the next train in the profiles. 
Instead, we store the average over all trains in the profile weighted by the probability of the traveler taking the train.
In probability theory, the expected value of a random variable is the average of all possible outcomes weighted by their probability.
Following this terminology, we refer to the modified arrival times in our profiles as \emph{expected arrival time}.
If the first train $B$ has an early arrival time but no good alternative exists, then the expected arrival time will be large.
Minimizing the expected arrival time therefore solves the problem.
We refer to the corresponding problem setting as \emph{minimum expected arrival time (MEAT)} problem.

The decision graph in Figure~\ref{fig:meat-example} is tiny.
Unfortunately, not all cities are as well connected as Karlsruhe and Berlin.
Decision graphs between more remote areas can quickly grow in size and contain backups over numerous layers.
We therefore investigate approaches to reduce the graph size and to represent it more compactly.

\subsection{Related Work to MEAT}
\label{sec:related-work-meat}
 
There has been a lot of research in the area of train networks and delays.
In contrast to our algorithm most of them compute single paths through the network instead of subgraphs containing all backups. 
To make this distinction clear we refer to such paths as \emph{single-path-journeys}.
The authors of \cite{dms-mcspt-08} define the reliability of a single-path-journey and propose to optimize it in the Pareto-sense with other criteria such as arrival time or the number of transfers. 
The availability of backups is not considered. 
The authors of \cite{bmpvw-rrupt-13}, based on delays occurred in the past, search for a single-path-journey that would have provided close to optimal travel times in every of the observed situations. 
Again, backups do not play a role.
The authors of \cite{gkmss-tprti-11} propose to first compute a set of safe transfers (i.e. those that always work).
They then develop algorithms to compute single-path-journeys that arrive before a given latest arrival time and only use safe transfers or at least minimize use of unsafe transfers.
The problem with this is approach is that unsafe transfers are avoided at all costs.
In the example of Figure~\ref{fig:meat-example}, the direct train from Mannheim to Berlin would be missed because the transfer is unsafe.
In \cite{ghms-rrti-13}, a robust primary journey is computed such that for every transfer stop a good backup single-path-journey to the target exists. 
However, the backups do not have their own backups. 
The approach optimizes the primary arrival time subject to a maximum backup arrival time.
The authors of \cite{fils-itrar-13} study the correlation between real world public transit schedules in Rom and compare them against the single-path-journeys computed by state-of-the-art route planners based on the scheduled timetable. 
They observe a significant discrepancy and conclude that one should consider the availability of good backups already at the planning stage. 
The authors of \cite{bss-drtpp-13} examine delay-robustness in a different context: 
Having computed a set of transfer patters on a scheduled timetable in a urban setting, they show that single-path-journeys based on these patterns are still nearly optimal even when introducing delays. 
The conclusion is that these sets are fairly robust (i.e., the paths in the delayed timetable often use the same or similar patterns). 
In \cite{bs-fbgr-14} the authors propose to present to the user a small set of transfer patterns that cover most optimal journeys. 
They show that in an urban setting few patterns are enough to cover most single-path-journeys.
In a different line of work, the authors of \cite{bgmo-sdplt-11} investigate how a delay-perturbed timetable will evolve over time using stochastic methods. 
Their study shows that this is a computationally expensive task (running time in the seconds) if the delay model accounts many real-world details. 
Using a model with such a degree of realism therefore seems unfeasible for delay-robust route planning (requiring query times in the milliseconds).

\subsection{Delay Model}
\label{sec:delay_model}

Every random variable $X$ in this work is denoted by capital letters, is continuous, non-negative, and has a maximum value $\max X$. 
We denote by $P[X\le x]$ the probability that the random variable is below some constant $x$ and by $E[X]$ the expected value of $X$.

A crucial component of any delay-robust routing system is choosing against which types of delays the system should be robust and how to model these delays. 
This choice has deep implications throughout the whole system. 
While a too simplistic model does not yield useful routes, a too complicated model makes routing algorithms too inefficient to be useful in interactive timetable information systems. 
We therefore propose a simple stochastic model.
While our model that does not cover every situation and is not delay-robust in every possible scenario, it works well enough to give useful routes with backups. 
Further, we were not able to construct a proof-of-concept implementation for a more complex model while maintaining reasonable query running times.

The central simplification is that we assume that all random variables are independent. 
Clearly, in reality this is not always the case.
However, if delays between many trains interact then the timetable perturbation must be significant. 
An example of a significant perturbation is a train track that is blocked for an extended period of time.
As reaction to such a perturbation even trains in the medium or distant future need to be rescheduled (or arrive at least not on-time). 
The set of possible outcomes and the associated uncertainty is huge. 
Accounting for every outcome seems infeasible to us. 
We argue that if the perturbation is large then we cannot account for all possible recovery scenarios in advance. 
Instead, the user should replan his journey based on the realtime delay situation. 
Furthermore, even if we could account for all scenarios, we would still face the problem of explaining every possible outcome to the user, which is a show-stopper in practice. 
Our model therefore only accounts for small disturbances as we only intend to be robust against these.
We believe that assuming independence for small disturbances is a model simplification that is acceptable in practice.

Formally, our model contains one random variable $\mathcal{D}_{c}$ per connection $c$. 
This variable indicates with which delay the train will arrive at $c_{\arrstop}$. 
We assume that all connections depart on time. 
This assumption does not induce a significant error because it roughly does not matter whether the incoming or the outgoing train is delayed. 
Furthermore, we assume that every connection $c$ has a maximum delay, i.e., $\max\mathcal{D}_{c}$ is a finite value.
Finally, we assume that all random variables are independent. 
Delays between trips are independent because if they were not then the perturbation would be large. 
We can assume that delays within a trip are independent as there nearly never exists an optimal decision graph that uses a trip more than once.

We assume that the changing times at stops are encoded in $\mathcal{D}_{c}$.
An transfer with a slack time below the regular change time should have a very low success probability but it should not be zero.
This way the computed decision graph will also include the very risky transfers that in practice only have a chance of working if the outgoing train departs delayed.
However, as the probability is low not much weight is attributed to them.
We further assume that inter-stop footpaths are handled by contracting adjacent stops and adjusting the change time.
These simplifications allows us to omit footpath and change time handling from the algorithm.
Fortunately, for applications that require them they can be incorporated analogously to how they are handled in the earliest arrival profile Connection Scan algorithm.

The only remaining modeling issue is to define what distribution the random variables $\mathcal{D}_{c}$ should have. 
An obvious choice is to estimate a distribution based on historic delay data. 
However, this has two shortcomings: 
\begin{itemize}
\item it is hard to get access to delay data (we do not have it), and 
\item you need to have records of many days with precisely the same planned schedule. 
\end{itemize}
Suppose for example that the user is in the middle of his journey and a significant perturbation occurs. 
The operator then adjusts the short-term timetable to reflect this and the user wants to reroutes based on this adjusted data.
With historic data this often is not possible because this exact recovery scenario may never have occurred in the past and almost certainly not often enough to extrapolate from the historic data.

\begin{figure}
\begin{center}
\includegraphics{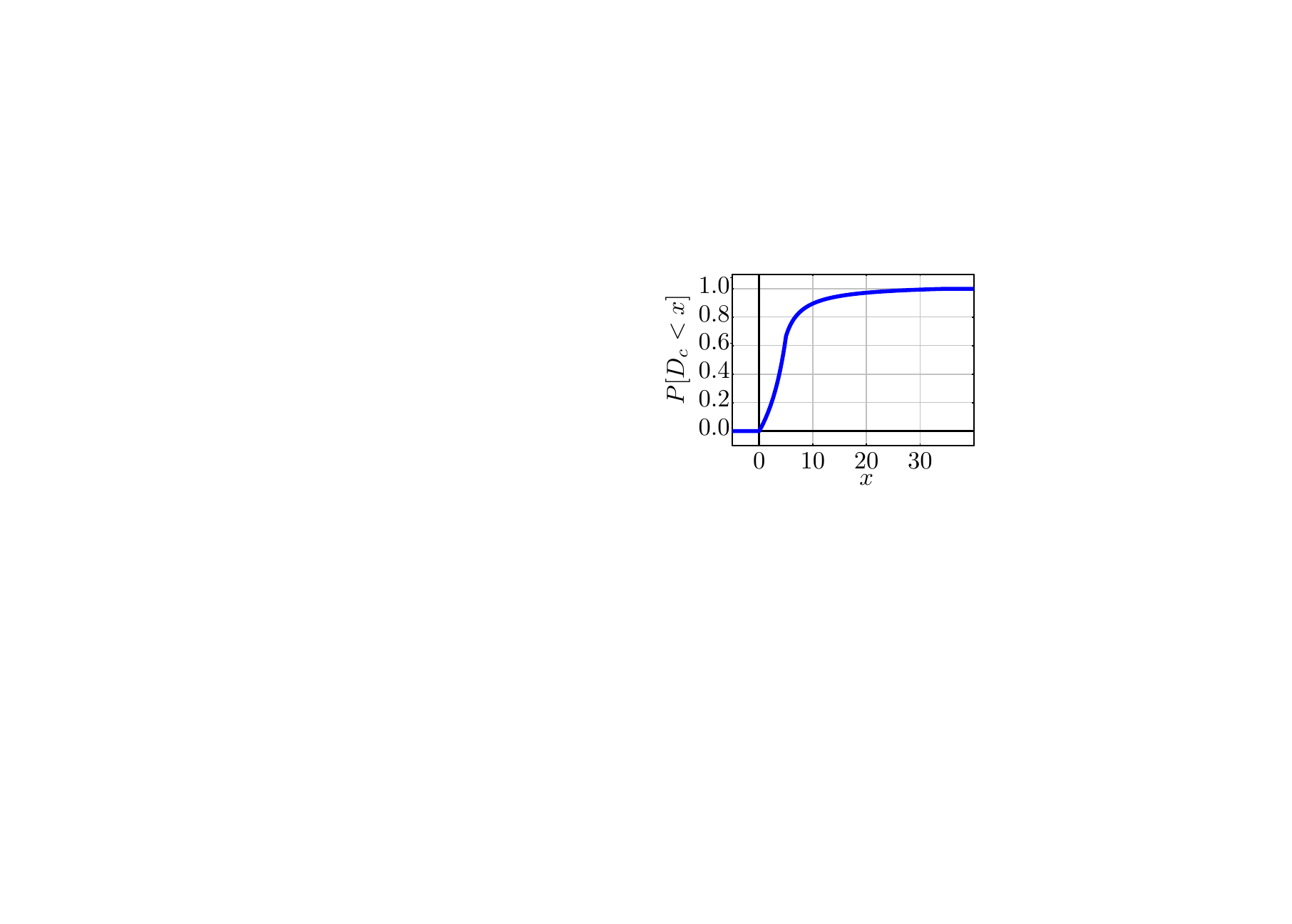}
\end{center}

\caption{\label{fig:delay-function}Plot showing $P[\mathcal{D}_{c}\le x]$ in function of $x$ for $m=5$ and $d=30$.}
\end{figure}

For these reasons, we propose to use synthetic delay distributions that are only parametrized on the planned timetable. 
We propose to add to each connection $c$ a synthetic delay variable $\mathcal{D}_{c}$ that depends on the change time $m$ of $c_{\arrstop}$ and on a global\footnote{$d$ is global since we lack per-train data. Our approach can be adjusted, if such data became available.} \emph{maximum delay parameter} $d$. 
We define $\mathcal{D}_{c}$ as follows: $\forall x\in(-\infty,0]:P[\mathcal{D}_{c}\le x]=0$, $\forall x\in(0,m]:P[\mathcal{D}_{c}\le x]=\frac{2x}{6m-3x}$, $\forall x\in(m,m+d]:P[\mathcal{D}_{c}\le x]=\frac{31(x-m)+2d}{30(x-m)+3d}$, and $\forall x\in(m+d,\infty):P[\mathcal{D}_{c}\le x]=1$. 
The function is illustrated in Figure~\ref{fig:delay-function} and the rational for its our design is given in the next section.

\subsubsection{Synthetic Delay Distribution} 

There are many methods to come up with formulas for synthetic delays. 
The lack of any effectively accessible ground truth makes any conclusive experimental evaluation of their quality very difficult. 
The only real criteria that we have is ``intuitively reasonable''. 
The approach presented here is by no way the final answer to the question of how to design the best synthetic delay distribution. 
In this section, we describe the rational for our design decisions.

We define for every connection $c$ its delay $\mathcal{D}_{c}$ by defining its cumulative distribution function $f_{m,d}(x)$, where $d$ is the maximum delay of $c$ and $m$ the minimum change time at $c_{\arrstop}$. 
Our delays do not depend on any other parameter than $m$ and $d$. 
We have the following hard requirements on $f_{m,d}$ resulting from our algorithm:
\begin{itemize}
\item $f_{m,d}(x)$ is a probability, i.e., $\forall x:0\le f(x)\le1$
\item $f_{m,d}(x)$ is a cumulative distribution function and therefore
non-decreasing, i.e., $\forall x:f_{m,d}'(x)\ge0$
\item $\max\mathcal{D}_{c}$ should be $m+d$, i.e., $\forall x\ge m+d:f(x)=1$ 
\item Our model does not allow for trains that arrive too early, i.e.,$\forall x<0:f(x)=0$
\end{itemize}
These requirements already completely define what happens outside of $x\in(0,m+d)$. 
Because of the limitations of current hardware, there are two additional more fuzzy but important requirements:
\begin{itemize}
\item We need to evaluate $f_{m,d}(x)$ many times. The formula must therefore
not be computationally expensive.
\item Our algorithm computes a lot of $\left(f_{m,d}(x_{1})+a_{1}\right)\cdot\left(f_{m,d}(x_{2})+a_{2}\right)\cdot\left(f_{m,d}(x_{3})+a_{3}\right)\cdots$ chains. 
The chain length reflects the number of rides in the longest journey considered during the computations. 
As 64-bit-floating points only have a limited precision, we must make sure that order of magnitude of the various values of $f_{m,d}$ do not differ too much. 
If they do differ a lot then the less likely journeys have no impact on the overall EAT because their impact is rounded away.
\end{itemize}
Finally there are a couple of soft constraints coming from our intuition:
\begin{itemize}
\item $f(m)$ is the probability that everything works as scheduled without the slightest delay. 
In practice this does happen and therefore this should have reasonable high probability. 
On the other hand a too high $f(m)$ can lead to problems with rounding. 
We set $f(m)=\frac{2}{3}$ as we believe that it is a good compromise.
\item We want $f$ to be continuous.
\item The maximum variation should be at $x=m$, i.e., $f'(m)$ should be the unique local maximum of $f'$. 
\item Initially the function should grow slowly and then once $x=m$ is reached the growth should slow down. 
This can be formalized as $f''(x)>0$ for $x\in(0,m)$ and $f''(x)<0$ for $x\in(m,m+d)$.
\end{itemize}
We define $f$ using two piece function $f_{1}$ and $f_{2}$. 
For these pieces we assume $m=5\mbox{min}$ and $d=30\mbox{min}$ and scale them to accommodate for different values, as following:
\[
f_{m,d}(x)=\begin{cases}
0 & \mbox{ if }x<0\\
f_{1}(\frac{5x}{m}) & \mbox{ if }0\le x\le m\\
f_{2}\left(\frac{30(x-m)}{d}\right) & \mbox{ if }m<x<m+d\\
1 & \mbox{ if }m+d\le x
\end{cases}
\]
It remains to define $f_{1}$ and $f_{2}$. 
We started with a $-1/x$ function and shifted and stretched the function graphs until we ended up with something that looks ``intuitively reasonable''.
\begin{eqnarray*}
f_{1}(x) & = & \frac{2x}{3(10-x)}\\
f_{2}(x) & = & \frac{31x+60}{30(x+3)}
\end{eqnarray*}
The resulting function $f$ fulfills all requirements and is illustrated in Figure~\ref{fig:delay-function}. 
To sum up: We define the $f_{m,d}$ as following:
\[
f_{m,d}(x)=\begin{cases}
0 & \mbox{ if }x<0\\
\frac{2x}{6m-3x} & \mbox{ if }0\le x\le m\\
\frac{31(x-m)+2d}{30(x-m)+3d} & \mbox{ if }m<x<m+d\\
1 & \mbox{ if }m+d\le x
\end{cases}
\]

\subsection{Decision Graphs}
\label{sec:decision_graphs}

In this subsection, we first introduce the notion of safe journey, then formally define decision graphs, and then introduce three problem variants: (i) the unbounded, (ii) the bounded, and (iii) the $\alpha$-bounded MEAT problems. 
The first two are of more theoretical interest, whereas the third one has the highest practical impact.
We prove basic properties of the unbounded and bounded problems and show a relation to the earliest safe arrival problem. 

\subsubsection{Formal Definition}

A \emph{safe $(s,\tau_s,t)$-journey} is a $(s,\tau_s,t)$-journey such for every transfer the time difference between the arrival of the incoming train and the departure of the outgoing train is at least the maximum delay of the incoming train.
We denote by $\mathrm{eat}(s,\tau_s,t)$ the arrival time of an optimal earliest arrival journeys and by $\mathrm{esat}(s,\tau_s,t)$ the arrival time of an optimal safe earliest arrival journey.

A \emph{$(s,\tau_s,t)$-decision graph} from source stop~$s$ to target stop~$t$ with the traveler departing at time~$\tau_s$ is a directed reflexive-loop-free multi-graph $G=(V,A)$ whose vertices correspond to stops and whose arcs correspond to legs $l$ directed from $l_{\depstop}$ to $l_{\arrstop}$. 
There may be several legs between a pair of stops, but they must be of part of different trips and depart at different times. 
We formalize this as: $\forall l^{1},l^{2}\in A:l_{\deptime}^{1}\neq l_{\deptime}^{2}\vee l_{\depstop}^{1}\neq l_{\depstop}^{2}$.
We require that the user must be able to reach every leg and must always be able to get to the target. 
Formally, we require that for every $l\in A$ there exists a $(s,\tau_s,l_{\depstop})$-journey $j$ with $j_{\arrtime}\le l_{\deptime}$ to reach the leg, and a safe $(l_{\arrstop},l_{\arrtime}+\max\mathcal{D}_{r},t)$-journey $j'$ to reach the target. 
To exclude decision graphs with unreachable stops, we require that every stop in $V$ except $s$ and $t$ have non-zero in- and out degree. 

We first recursively define the \emph{expected arrival time} $e(l)$ (short EAT) of a leg $l\in A$ and define in terms of $e(l)$ the EAT $e(G)$ of the whole decision graph $G$. 
If $l_{\arrstop}=t$, we define $e(l)=l_{\arrtime}+E[\mathcal{D}_{l}]$. 
Otherwise $e(l)$ is defined in terms of other legs. 
Denote by $q_{1}\ldots q_{n}$ the sequence of legs in $G$ ordered by departure time, departing at $l_{\arrstop}$ after $l_{\arrtime}$, i.e., every leg that the user could reach after $l$ arrives. 
Denote by $d_{1}\ldots d_{n}$ their departure times and set $d_{0}=l_{\arrtime}$. 
We define $e(l)=\sum_{i\in\{1\ldots n\}}P[d_{i-1}<\mathcal{D}_{l}<d_{i}] \cdot e(q_{i})$, i.e., the average of the EATs of the connecting legs weighted by the transfer probability. 
Note that this definition is well-defined because $e(l)$ only depends on $e(q)$ of legs with a later departure time, i.e., $l_{\deptime}<q_{\deptime}$. 
Further notice that $P[\mathcal{D}_{l}<d_{n}]=1$. 
Otherwise no safe journey to the target would exist invalidating the decision graph.

We denote by $G^\first$ the leg $l\in A$ with minimum $l_{\deptime}$.
This is the leg that the user must initially take at $s$.
We define the \emph{expected arrival time} $e(G)$ (short EAT) of the decision graph $G$ as $e(G^\first)$. 
Furthermore, the \emph{latest arrival time} $G_{\max \arrtime}$ is the maximum $l_{\arrtime}+\max\mathcal{D}_{l}$ over all $l\in A$. 
Note that by minimizing $G_{\max \arrtime}$, we can bound the worst case arrival time giving us some control over the arrival time variance.

The \emph{unbounded $(s,\tau_s,t)$-minimum expected arrival time }(short MEAT) problem consists of computing a \emph{$(s,\tau_s,t)$-}decision graph $G$ minimizing $e(G)$. 
The bounded \emph{$(s,\tau_s,t)$-MEAT} problem consists of computing a \emph{$(s,\tau_s,t)$}-decision graph $G$ minimizing $e(G)$ subject to a minimum $G_{\max \arrtime}$.
As a compromise between bounded and unbounded, we further define the $\alpha$-bounded MEAT problem: We require that $G_{\max \arrtime}-\tau_s\le\alpha\left(\mathrm{esat}\left(s,\tau_s,t\right)-\tau_s\right)$, i.e., the maximum travel time must not be bigger than $\alpha$ times the delay-free optimum. 
Notice that the bounded and 1-bounded MEAT problems are equivalent.

\subsubsection{Decision Graph Existence}
\begin{lemma}
There is a \emph{$(s,\tau_s,t)$-}decision graph $G$ iff there exists a safe \emph{$(s,\tau_s,t)$}-journey $j$.
\end{lemma}
\begin{proof}
If there exists a \emph{$(s,\tau_s,t)$-}decision graph $G$ then by the decision graph definition we know that there exists a safe $(G_{\arrstop}^\first,G_{\arrtime}^\first+\max\mathcal{D}_{G^\first},t)$-journey $j'$. 
Prefixing $j'$ with $G^\first$ yields the required $(s,\tau_s,t)$-journey $j$.

Conversly, if there exists a $(s,\tau_s,t)$-journey $j$, we can construct a (non-optimal) $(s,\tau_s,t)$-decision graph $G$ that contains exactly the same legs as $j$.
\end{proof}
A direct consequence of this lemma is that the minimum $G_{\max \arrtime}$ over all \emph{$(s,\tau_s,t)$-}decision graphs $G$ is equal to $\mathrm{esat}(s,\tau_s,t)$.
Using this observation we can reduce the bounded MEAT problem to the unbounded MEAT problem. 
Formally stated:
\begin{lemma}
An optimal solution $G$ to the bounded \emph{$(s,\tau_s,t)$-}MEAT problem on timetable $T$ is an optimal solution to the unbounded $(s,\tau_s,t)$-MEAT problem on a timetable $T'$ where $T'$ is obtained by removing all connections $c$ from $T$ with $c_{\arrtime}$ above the $\mathrm{esat}(s,\tau_s,t)$.
\end{lemma}
\begin{proof}
There are two central observations needed for the proof: 
First, every $(s,\tau_s,t)$-decision graph on timetable $T'$ is a $(s,\tau_s,t)$-decision graph on the strictly larger timetable $T$. 
Second, every safe \emph{$(s,\tau_s,t)$-}journey in $T'$ is an earliest safe \emph{$(s,\tau_s,t)$-}journey in $T$. 
Suppose that a \emph{$(s,\tau_s,t)$}-decision graph $G'$ on $T'$ would exist with a suboptimal $G'_{\max \arrtime}$ then there would also exist a safe \emph{$(s,\tau_s,t)$-}journey $j'$ in $T'$ with a suboptimal $j_{\arrtime}'$, which is not possible by construction of $T'$, which is a contradiction.
\end{proof}

\begin{figure}
\begin{center}
\includegraphics{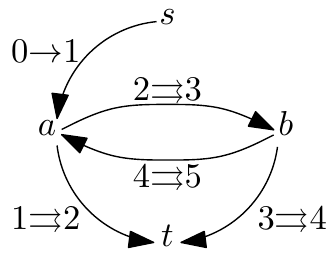}
\end{center}

\caption{\label{fig:infinite-network}A timetable $T_{p}$ has 4 stops: $s$,
$a$, $b$ and $t$. The arrows denote connections. An arrow annotated
with its departure time and arrival time. A simple arrow ($\rightarrow$)
denotes a single non-repeating connection. A double arrow ($\rightrightarrows$)
is repeated every 4 time units, i.e. $1\rightrightarrows2$ is a shorthand
for $1+4i\rightarrow2+4i$ for every $i\in\mathbb{N}$. All connections
are part of their own trip and have the same delay variable $\mathcal{D}$.
We define $P[\mathcal{D}=0]=p$ (with $p\neq0$) and $P[\mathcal{D}<1]=1$.}
\end{figure}%

Having shown how to explicitly bound $G_{\max \arrtime}$ it is natural to ask what would happen if we dropped this bound and solely minimized $e(G)$. 
For this we consider the timetable $T_{p}$ with an infinite connection set illustrated and defined in Figure~\ref{fig:infinite-network}.
Notice that $T_{p}$ is constructed such that it does not matter whether the user arrives at $a$ at moments $1+4\mathbb{N}$ or at $b$ at moments $3+4\mathbb{N}$ as the two states are completely symmetric with the stops $a$ and $b$ swapping roles. 
By exploiting this symmetry we can reduce the set of possibly optimal $(s,0,t)$-decision graphs to 2 elements: 
the decision graph $G^{1}$ that waits at $a$ and never goes over $b$, and the decision graph $G^{2}$ that oscillates between $a$ and $b$. 
The corresponding expected arrival times are $e(G^{1})=p\left(2+E[\mathcal{D}]\right)+(1-p)\left(7+E[\mathcal{D}]\right)$ and $e(G^{2})=p\left(2+E[\mathcal{D}]\right)+(1-p)\left(3+e(G^{2})\right)$.
The later equation can be resolved to $e(G^{2})=E[\mathcal{D}]-1+\frac{3}{p}$.
We can solve $e(G^{1})<e(G^{2})$ in terms of $p$. 
The result is that $G^{1}$ is better if $p<\frac{\sqrt{43}-4}{9}\approx0.28$.
If they are equal then $G^{1}$ and $G^{2}$ are equivalent and otherwise $G^{2}$ is better. 

This has consequences even for timetables with a finite connection set.
One could expect that to compute a decision graph it is sufficient to look at a time-interval proportional to its expected travel time: 
It seems reasonable that a connection scheduled to occur in ten years would not be relevant for a decision graph departing today with an expected travel time of one hour. 
However, this intuition is false in the worst case: Consider the finite sub-timetable $T'$ of the periodic timetable $T_{p}$ that encompasses the first ten years (i.e., we ``unroll'' $T_p$ for ten years). 
For $p{>}0.28$, an optimal $(s,0,t)$-decision graph will use all connections in $T'$, including the ones in ten years (as $G^2$ would).
Fortunately, the bounded MEAT problem does not suffer from this weakness:
No connection arriving after $\mathrm{esat}(s,0,t)$ can be relevant. 
Therefore, even on infinite networks the bounded MEAT problem always admits finite solutions.
This property is the main motivation to study the bounded MEAT problem.

\subsubsection{Non-dominated Pairs and Decision Graphs}

In this section, we only consider decision graphs and journeys arriving at a fixed target stop $t$. 
All lemmas and definition are therefore with respect to $t$. 
To simplify our notation, we omit $t$ in this section.

We consider, for every connection $c$, the pair $p_c=(c_\deptime, e(G))$ where $G$ is a decision graph that minimizes the expected arrival time, subject to $c$ being the first connection, i.e., $G^\first_\enter=c$.
Denote by $O$ the outgoing connections of a stop.
Every connection has an associated pair, which can be dominated within $O$.
This allows us to define when a connection is dominated: It is dominated when its pair is dominated.
A leg $l$ is dominated if $l_\enter$ is dominated.
Non-dominated connections have an important role in the computation of optimal decision graphs as the following lemma shows.

\begin{lemma}
\label{lem:meat-only-non-dominated}
For every source stop $s$ and source time $\tau_s$, if there exists a decision graph, then there exists an optimal decision graph, such that for every leg $l$ of $G$, the entry connection $l^\enter$ is non-dominated at $l^\enter_\depstop$.
\end{lemma}
\begin{proof}
We know that an optimal decision graph $H$ exists as we required the existence of a decision graph.
If $H^\first_\enter$ is dominated, then there is another optimal decision graph associated with the dominating connection.
Without loose of generality we can therefore assume that, $H^\first_\enter$ is non-dominated.

Suppose that $H$ contained some other leg $l$ such that $l^\enter$ is dominated. 
Further denote by $l'$ an incoming leg from which the traveler might transfer to $l$.
$l'$ must exist because $l$ is not the first leg in the decision graph.
As $l$ is dominated, removing it and all legs that can only reached via $l$ from $H$ improves $e(l')$, which in terms improves $e(H)$, which is a contradiction to $H$ being optimal. 
\end{proof}

\subsection{Solving the MEAT Problem}

The unbounded MEAT problem can be solved to optimality on finite networks, and by extension also the bounded and $\alpha$-bounded MEAT problems. 
We first describe an algorithm to optimally solve the unbounded MEAT problem.
By applying this algorithm to a restricted timetable we solve the bounded and $\alpha$-bounded MEAT problems. 

\subsubsection{Solving the Unbounded MEAT problem}

Our algorithm works in two phases:
\begin{itemize}
\item Compute the minimum expected arrival times for all connections $c$,
\item extract a desired $(s,\tau_s,t)$-decision graph. 
\end{itemize}
The first phase is a variant of the earliest arrival profile Connection Scan algorithm.
The second phase is an extension of the journey extraction algorithm.

\paragraph{Phase 1: Computing all Expected Arrival Times.}

Recall the basic Connection Scan profile framework depicted in Figure~\ref{alg:profile-connection-scan-framework} and especially the earliest arrival time instantiation depicted in Figure~\ref{alg:earliest-arrival-profile-connection-scan}.
We first describe the algorithmic differences to the later and then explain why the proposed algorithm is correct.
In the context of this subsection $c$ always refers to the connection being scanned.

The first key idea consists of replacing all earliest arrival times with minimum expected arrival times.
This works similarly to the profile Pareto-optimization where all earliest arrival times were replaced by vectors. 
The stop data structure becomes an array of dynamic arrays of pairs of departure time and expected arrival time.
The trip data structure becomes an array of expected arrival times.
The computation of the expected arrival time when arriving at the target $\tau_1$ is only modified in a minor way: 
We need to add $E\mathcal{D}_c$, a constant, to the arrival time of the connection.
The arrival time when the traveler remains sitting $\tau_2$ is computed in exactly the same way by reading the value of $T[c_\trip]$.
The computation of the arrival time when changing trains $\tau_3$ is significantly modified and is described below.
The value of $\tau_c$ is still computed as the minimum of $\tau_1$, $\tau_2$, and $\tau_3$.
$\tau_c$ is the minimum expected arrival time over all decision graphs starting in $c$. 
Formulated differently, $\tau_c$ is the minimum $e(G)$ over all decision graphs such that $G^\mathrm{first}_\enter = c$. 
Incorporating $\tau_c$ into the trip data structure $T$ and the stop profiles $S$ works completely analogous to the earliest arrival profile algorithm.

The computation of $\tau_3$, i.e., the computation of the arrival time when transferring trains is changed. 
The reason for this change is that the arriving train $c$ has a random arrival time between $c_\arrtime$ and $c_\arrtime + \max \mathcal{D}_c$.
Our algorithm starts by determining using a sequential scan all pairs $p^1\ldots p^k$ in the profile $S[c_\arrstop]$ that might be relevant.
These are all pairs departing between $c_\arrtime$ and $c_\arrtime + \max \mathcal{D}_c$ and the first pair after $c_\arrtime + \max \mathcal{D}_c$.
These correspond to all outgoing trains that are worth taking.
It then computes $\tau_3$ as the weighted sum over the expected arrival times of all $p^i$.
A pair is weighted by the probability of the incoming being delay in such a way that the traveler will take it.
Formally this means: $p^1$ is weighted by the probability $P[c_\arrtime + \max \mathcal{D}_c \le p^1_\deptime]$ and all other $p^i$ are weighted by $P[p^{i-1}_\deptime \le c_\arrtime + \max \mathcal{D}_c \le p^i_\deptime]$.
Formulated differently, $\tau_2$ is the average over the expected arrival time of the non-dominated outgoing trains weighted by the probability of the traveler transferring to them. 

The correctness of our algorithm relies on optimal decision graphs not containing any dominated legs as shown in Lemma~\ref{lem:meat-only-non-dominated}.
The domination test in the profile insertion filters dominated pairs and pairs which appear several times.
In the later case there are two or more connections that depart at the same time and have the same expected arrival time.
In this case, it does not matter which we insert into the decision graph but we may only insert one.
Our algorithm picks the connection that appears last in the connection array.

It remains to show why our strategy of selecting all outgoing non-dominated connections during the evaluation is optimal. 
This directly follows from the pairs being ordered.
One does not want to skip earlier pairs because they have lower expected arrival times than the later trains.
One cannot remove the later trains because it is not guaranteed that the earlier trains can be reached.
Connection not in the profile are dominated.
From Lemma~\ref{lem:meat-only-non-dominated} follows that we can ignore dominated connections.

\paragraph{Phase 2: Extracting Decision Graphs.}

We extract a $(s,\tau_s,t)$-decision graph $G=(V,A)$ by enumerating all legs in $A$. 
The stop set $V$ can then be inferred from $A$. 
At the core, our algorithm uses a min-priority queue that contains connections ordered increasing by their departure time. 
Initially, we add the earliest connection in the profile of $s$ to the queue. 
While the queue is not empty we pop the earliest connection $c^{1}$ from it. 
Denote by $c^{2}\ldots c^{n}$ all subsequent connections in the trip $c_{\trip}^{1}$. 
The desired leg $l=(c^{1},c^{i})$ is given by the first $i$ such that $e(c^{1})\neq e(c^{i+1})$ (or $i=n$ if all are equal). 
We add $l$ to $G$. 
If $c_{\arrstop}^{i}\neq t$ we add the following connections to the queue: 
(i) All connections in the profile of $c_{\arrstop}^{i}$ departing between $c_{\arrtime}^{i}$ and $c_{\arrtime}^{i}+\max\mathcal{D}_{c^{i}}$, and 
(ii) the first connection in the profile of $c_{\arrstop}^{i}$ departing after $c_{\arrtime}^{i}+\max\mathcal{D}_{c^{i}}$.

\subsubsection{Solving the $\alpha$-Bounded MEAT Problem}
\label{sub:Bounded-MEAT-Algo}
We assume that the connection set is stored as an array ordered by departure time. 
To solve the $\alpha$-bounded $(s,\tau_s,t)$-MEAT problem we perform the following steps: 
(i) run a binary search on the connection set to determine the earliest connection $c^\first$ departing after $\tau_s$, 
(ii) run a one-to-one Connection Scan from $s$ to $t$ that assumes all connections $c$ are delayed by $\max \mathcal{D}_c$ to determine $\mathrm{esat}\left(s,\tau_s,t\right)$
(iii) let $\tau_{\mathrm{last}}=\tau_s+\alpha\cdot\left(\mathrm{esat}\left(s,\tau_s,t\right)-\tau_s\right)$ and run a second binary search on the connection set to find the last connection $c^{\mathrm{last}}$ departing before $\tau_{\mathrm{last}}$,
(iv) run a one-to-all Connection Scan from $s$ restricted to the connections from $c^\first$ to $c^{\mathrm{last}}$ to determine all $\mathrm{eat}\left(s,\tau_s,\cdot\right)$, 
(v) run Phase 1 of the unbounded MEAT algorithm scanning the connections from $c^{\mathrm{last}}$ to $c^\first$ skipping connections~$c$ for which $c_\arrtime > \tau_{\mathrm{last}}$ or $\mathrm{eat}(s,\tau_s,c_{\depstop})\le c_{\deptime}$ does not hold, and finally 
(vi) run Phase 2 of the unbounded MEAT algorithm, i.e., extract the $(s,\tau_s,t)$-decision graph.

\subsection{Decision Graph Representation}
\label{sec:representation}

In the previous section we described how to compute decision graphs.
In practice this is not enough and we must be able to represent the graph in a form that the user can effectively comprehend. 
The main obstacle here is to prevent the user from being overwhelmed with information.
A secondary obstacle is how to actually layout the graph. 
In this section, we solely focus to reducing the amount of information. 
The presented drawings were created by hand.
In the demonstration, we use GraphViz~\cite{egknw-gdsdg-03}.

\subsubsection{Expanded Decision Graph Representation}

\begin{figure}
\begin{center}
\subfloat[\label{fig:meat-expanded}Expanded]{\includegraphics{ka_be_exp}}
~~
\subfloat[\label{fig:meat-compact}Compact]{\includegraphics{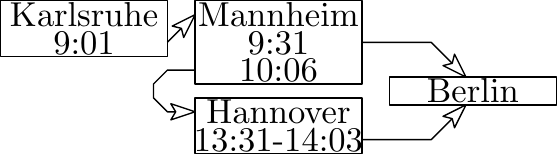}}
\end{center}
\caption{\label{fig:meat-decision-graph}Decision graph representations from Karlsruhe at 9:00 to Berlin.}
\end{figure}

The expanded decision graph subdivides each node $v$ into slots $s_{v,1}\ldots s_{v,n}$ that correspond to moments in time that an arc arrives or departs at $v$. 
The slots in each node are ordered from top to bottom in chronological order. 
Each arc $(u,v)$ connects the corresponding slots $s_{u,i}$ and $s_{v,j}$. 
To determine his next train the user has to search for the box corresponding to his current stop and pick the first departure slot after the current moment in time. 
The arrows guide him to the box corresponding to his next stop. 
Figure~\ref{fig:meat-expanded} illustrates this graph drawing style.

\subsubsection{Compact Decision Graph Representation}

The scheduled arrival time of trains is an information contained in the expanded decision graph that is not strictly necessary. 
A traveler decides what outgoing train to take when he arrives. 
At that moment, he can look at any clock to figure out the precise arrival time.
The scheduled arrival time recorded in the timetable is not needed for his decision.

The compact decision graph exploits this observation by removing the arrival time information from the representation. 
Each arc $(u,v)$ connects the corresponding departure slot $s_{u,i}$ directly to the stop $v$ instead of a slot. 
Time slots that only appear as arrival slots are removed. 
If two outgoing arcs of a node $u$ have the same destination and depart subsequently, they are grouped and only displayed once.
The compact decision graph is never larger than the expanded one and most of the time significantly smaller. 
See Figure~\ref{fig:meat-compact} for an example of a compact decision graph.

\subsubsection{Relaxed Dominance}

Decision graphs exist that contain legs that have near to no impact on the EAT. 
Removing them increases the EAT by only a small amount, resulting in an almost optimal decision graph that can be significantly smaller. 
To exploit this, we introduce a \emph{relaxation tuning parameter} $\beta$. 
EATs are regarded as equal if their difference is below $\beta$. 
Formulated in terms of the framework depicted in Figure~\ref{alg:profile-connection-scan-framework}, we only insert a new pair into the profile $S[x]$, if the expected arrival time of the earliest pair of $S[x]$ is at least $\beta$ time units later than $\tau_c$.

\subsubsection{Displaying only the Relevant Subgraphs}

In many scenarios, we have a canvas of fixed size.
If even the compact relaxed decision graph is too large to fit, we can only draw parts of it.
We observe that the decision graph extraction phase does not rely on the actual distributions of the delay variables $\mathcal{D}_{c}$ but only on $\max\mathcal{D}_{c}$. 
It extracts all connections departing in an interval $I$, plus the first connection directly afterwards. 
The full decision graph is extracted when $I=[c_\arrstop,c_\arrstop+\max \mathcal{D}_c]$.
Reducing the size of $I$ reduces the number of legs displayed, while still guaranteeing that backup legs exist. 
For example a smaller partial decision graph is extracted, if we only follow the connections departing in $I=[c_\arrstop,c_\arrstop+\kappa]$ for $\kappa=1/2 \cdot \max \mathcal{D}_c$.
Valid values for $\kappa$ are from $0$ to $\max \mathcal{D}_c$.
We refer to $\kappa$ as \emph{display window}. 
Given an upper bound $\gamma$ on the number of arcs in the compact or expanded representation, we use a binary search to determine the maximum display window $\kappa$ and draw the corresponding subgraph. 
Note that in the worst case the display window has size~0. 
In this case the decision graph degenerates to a single-path-journey.

\subsection{Experiments}
\label{sec:experiments}

\begin{table}
\begin{center}
\begin{tabular}{c@{\,\,}crrrrrrrrrrrr}
\toprule
\multicolumn{1}{c}{} &  & \multicolumn{4}{c}{Unbounded} & \multicolumn{4}{c}{2.0-Bounded} & \multicolumn{4}{c}{1.0-Bounded}\\
\cmidrule(lr){3-6}\cmidrule(lr){7-10}\cmidrule(lr){11-14}
\multicolumn{1}{c}{} &  & \multicolumn{1}{c}{\begin{sideways}
\negthinspace{}Time~
\end{sideways}} & \multicolumn{1}{c}{\begin{sideways}
\negthinspace{}Stops~
\end{sideways}} & \multicolumn{1}{c}{\begin{sideways}
\negthinspace{}Legs~
\end{sideways}} & \multicolumn{1}{c}{\begin{sideways}
\negthinspace{}Arcs~
\end{sideways}} & \multicolumn{1}{c}{\begin{sideways}
\negthinspace{}Time~
\end{sideways}} & \multicolumn{1}{c}{\begin{sideways}
\negthinspace{}Stops~
\end{sideways}} & \multicolumn{1}{c}{\begin{sideways}
\negthinspace{}Legs~
\end{sideways}} & \multicolumn{1}{c}{\begin{sideways}
\negthinspace{}Arcs~
\end{sideways}} & \multicolumn{1}{c}{\begin{sideways}
\negthinspace{}Time~
\end{sideways}} & \multicolumn{1}{c}{\begin{sideways}
\negthinspace{}Stops~
\end{sideways}} & \multicolumn{1}{c}{\begin{sideways}
\negthinspace{}Legs~
\end{sideways}} & \multicolumn{1}{c}{\begin{sideways}
\negthinspace{}Arcs~
\end{sideways}}\\
\midrule
\multirow{5}{*}{\begin{sideways}
\negthinspace{}0min-Relax
\end{sideways}} & Avg & 6\,452 & 12 & 98 & 42 & 138 & 12 & 87 & 35 & 26 & 9 & 45 & 19\\
 & 33\% & 6\,209 & 7 & 22 & 10 & 84 & 7 & 22 & 10 & 16 & 7 & 15 & 7\\
 & 66\% & 7\,407 & 13 & 70 & 31 & 162 & 13 & 69 & 31 & 27 & 10 & 40 & 19\\
 & 95\% & 7\,635 & 25 & 349 & 125 & 312 & 24 & 330 & 119 & 66 & 19 & 149 & 57\\
 & Max & 7\,805 & 280 & 35\,450 & 28\,848 & 817 & 173 & 5\,540 & 4\,703 & 288 & 38 & 1\,607 & 366\\
\midrule 
\multirow{5}{*}{\begin{sideways}
\negthinspace{}1min-Relax
\end{sideways}} & Avg & 5\,122 & 12 & 88 & 39 & 116 & 12 & 73 & 31 & 25 & 9 & 39 & 17\\
 & 33\% & 4\,628 & 8 & 26 & 12 & 75 & 8 & 25 & 12 & 16 & 6 & 14 & 7\\
 & 66\% & 6\,026 & 13 & 66 & 31 & 136 & 13 & 64 & 30 & 26 & 10 & 36 & 17\\
 & 95\% & 6\,368 & 24 & 284 & 110 & 249 & 24 & 257 & 100 & 64 & 18 & 123 & 52\\
 & Max & 6\,595 & 50 & 12\,603 & 6\,558 & 685 & 50 & 1\,576 & 478 & 240 & 37 & 1\,390 & 289\\
\midrule
\multirow{5}{*}{\begin{sideways}
\negthinspace{}5min-Relax
\end{sideways}} & Avg & 4\,180 & 11 & 66 & 33 & 100 & 11 & 51 & 25 & 24 & 9 & 29 & 15\\
 & 33\% & 3\,845 & 8 & 24 & 12 & 66 & 8 & 23 & 11 & 15 & 6 & 13 & 6\\
 & 66\% & 4\,808 & 13 & 53 & 26 & 115 & 12 & 51 & 25 & 25 & 10 & 30 & 15\\
 & 95\% & 5\,028 & 22 & 178 & 82 & 216 & 22 & 155 & 74 & 61 & 17 & 84 & 42\\
 & Max & 5\,159 & 54 & 6\,640 & 3\,220 & 553 & 54 & 760 & 285 & 196 & 34 & 590 & 183\\
\bottomrule
\end{tabular}
\end{center}

\caption{\label{tab:meat-results}The time (in ms) needed to compute a decision graph and its size. 
Arcs is the number of arcs in the compact representation.
The number of rides corresponds to the number of arcs in the expanded representation. 
The maximum delay parameter is set to 1h.
We report average, maximum and the 33\%-, 66\%- and 95\%-quantiles.}
\end{table}

\begin{table}
\begin{center}
\begin{tabular}{cr}
\toprule
\#Stop & 16~991\\
\#Conn. & 55~930~920\\
\#Trip & 3~965~040\\
\bottomrule 
\end{tabular}
\end{center}

\caption{\label{tab:meat-instance}Instance Size}

\end{table}%

For our experiments, we used on a single core of a Xeon E5-2670 at 2.6 GHz, with 64 GiB of DDR3-1600 RAM, 20 MiB of L3 and 256 KiB of L2 cache.
This is the ``older'' machine used in the experiments of the previous sections.
We implemented the algorithm in C++ and compiled it using GCC 4.7.1 with -O3. 

The timetable is based on the data of \url{bahn.de} during winter 2011/2012. 
This the same primary data source as used for the experiments of Section~\ref{sec:csa_accel_exp}.
However, we extracted a different formal timetable.
We extracted every vehicle except for most buses as we mainly focus on train networks. 
Not having buses explains the significant instance size difference compared to the Germany instance of the previous sections. 
Not having buses allows us to get the running times onto a manageable level.
Further, it allows us to focus on long-distance trains where delays have a significantly larger impact than in high-frequent inner-city transit. 
We removed footpaths longer than 10\,min, connected stops with a distance below 100\,m, and then contracted stops connected through footpaths adjusting their minimum change times resulting in an instance without footpaths.
Not having footpaths again benefits query running times.
We pick the largest strongly connected component to make sure that there always exists a journey (assuming enough days are considered).
We extract one day of maximum operation (i.e. extract everything regardless of the day of operation and remove exact duplicates).
We then replicated this day 30 times to have a timetable spanning about one month of operation.
The detailed sizes are in Table~\ref{tab:meat-instance}. 
We ran 10\,000 random queries. 
Source and target stop are picked uniformly at random. 
The source time is chosen within the first 24h. 
We filter queries out that have an minimum delay-free travel time above 24h.

Our experimental results are presented in Table~\ref{tab:meat-results}.
The compact representation is smaller by a factor of 2 in terms of arcs than the expanded one. 
As expected, a larger relaxation parameter gives smaller graphs.
Increasing the $\alpha$-bound leads to larger graphs and running times grow.
The running times of unbounded queries are proportional to the timespan of the timetable (i.e. 30 days). 
On the other hand, the running times of bounded queries depend only on the maximum travel time of the journey. 
This explains the gap in running time of two orders of magnitude. 
As the maximum values are significantly higher than the 95\%-quantile, we can conclude that the graphs are in most cases of manageable size with a few outlines that distort the average values. 
Upon closer manual inspection, we discover that most outliers with large decision graphs connect remote rural areas, where even no ``good'' delay-free journey exists. 
We can therefore not expect to find any form of robust travel plan.

\begin{figure}
\begin{center}
\subfloat[1.0-Bounded]{\includegraphics[width=4cm]{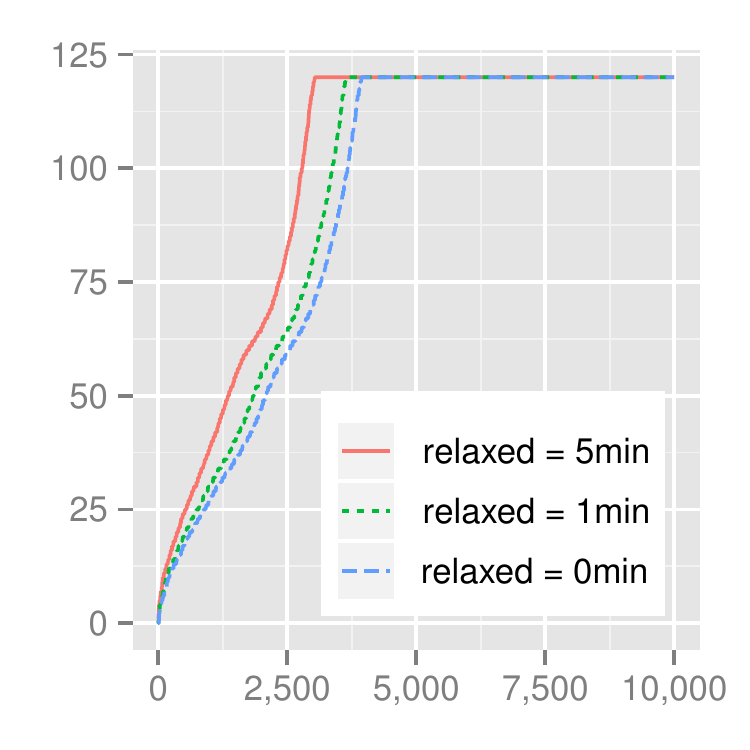}}
\subfloat[2.0-Bounded]{\includegraphics[width=4cm]{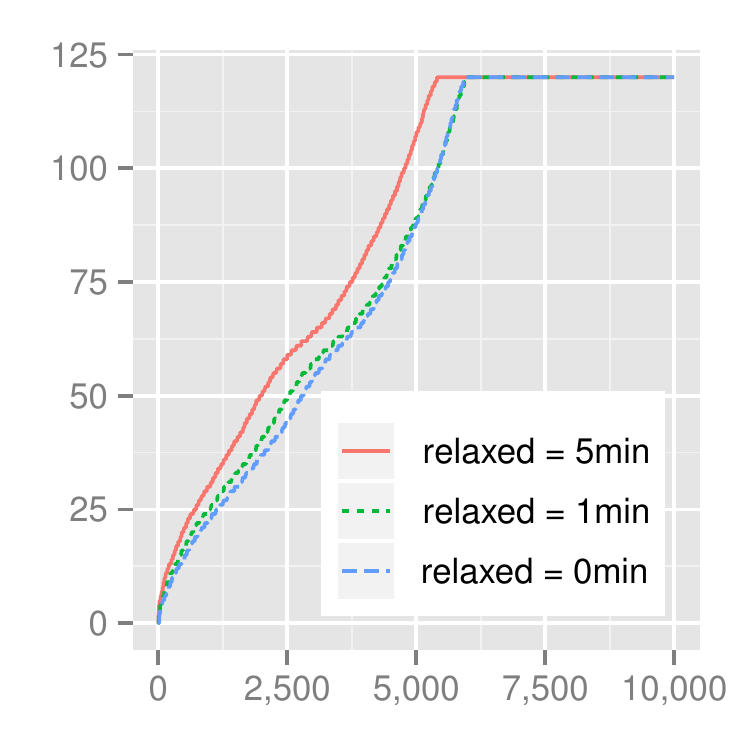}}
\subfloat[Unbounded]{\includegraphics[width=4cm]{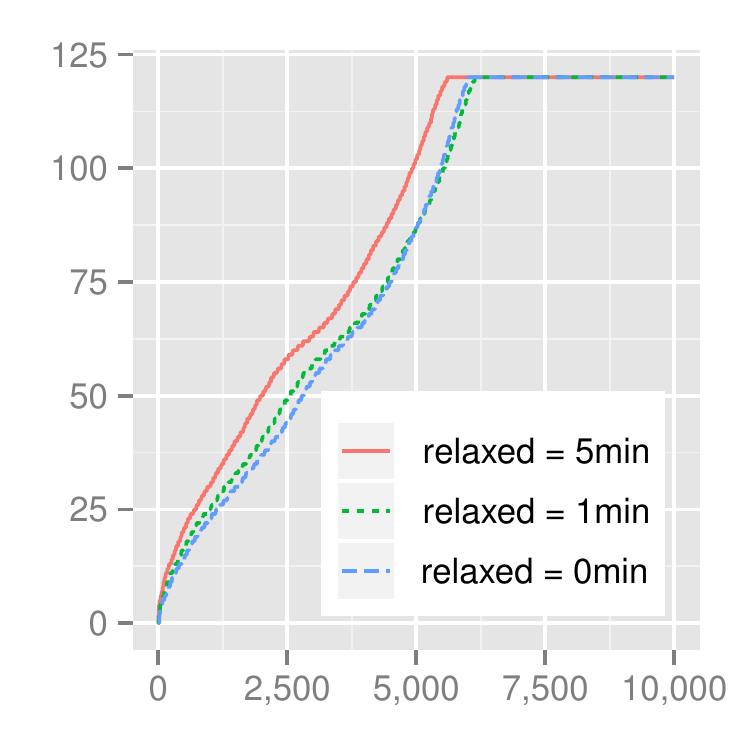}}
\end{center}

\caption{Display windows in min (y-axis) for each of the 10\,000 test queries (x-axis) ordered increasingly. 
The maximum delay parameter is set to 2h.}
\label{fig:display}
\end{figure}

In Figure~\ref{fig:display} we evaluate the value of the display window such that the extracted graphs have less than 25 arcs in the compact representation. 
Recall that this modifies what is displayed to the user.
It is still guaranteed that backups exist. 
As the 1.0-bounded graphs are smaller than 2.0-bounded graphs we can display more, explaining the larger display window. 
The difference between 2.0-bounded graphs and unbounded graphs is small. 
A greater relaxation parameter also reduces the graph size and thus allows for slightly larger display windows. 
If there is no ``good'' way to travel the decision graphs degenerate to single-path-journeys.

\paragraph{Section Conclusions.}

We described the Minimum Expected Arrival Time (MEAT) problem and described an efficient CSA-based algorithm to solve it.
This demonstrates that the CSA-framework is very flexible and can adapt to complex problem settings.
The achieved query running times of 100ms on average are fast enough for interactive systems.
This is further demonstrated by our proof of concept implementation accessible at \url{http://meatdemo.iti.kit.edu}.

However, the fast query running times were bought by removing most buses from the instance.
For the full Germany instance the running times are unfortunately prohibitively large.
Fortunately, decision graphs make most sense in long-distance travel where most high-frequency local bus lines do not play a role.
The size of the computed decision graphs can become large.
Fortunately, using careful engineering it is possible to sufficiently reduce their size to a manageable size.

Overall, we believe that the MEAT problem and our CSA-based algorithm are a promising basis on which an innovative timetable information system can be built.

\section{Conclusion}

We described the Connection Scan family of algorithms (CSA).
The algorithms are a simple solution to various routing problems in timetable-based networks.
We presented profile and non-profile variants of the algorithm.
CSA optimizes the arrival time and optionally the number of transfers in the Pareto sense.
CSA can adjust to a new timetable in mere seconds enabling the computation of journeys with respect to the current delay situation.
We combined CSA with multilevel overlay techniques yielding Connection Scan Accelerated (CSAccel).
CSAccel improves over CSA in terms of query running time on large country networks at the expense of an increased preprocessing running time and an increased code complexity.
Finally, we described the Minimum Expected Arrival Time (MEAT) problem and a CSA-based solution algorithm.
All algorithms were experimentally evaluated in an in-depth study.

\end{document}